\RequirePackage[l2tabu, orthodox]{nag}
\documentclass[preprint, 11pt, reqno]{amsart}
\usepackage[normalem]{ulem}
\usepackage{mathrsfs}
\usepackage{fullpage}
\usepackage[dvipsnames]{xcolor}
\usepackage{graphicx}
\usepackage{braket}
\usepackage{hyperref}
\usepackage{breqn}
\usepackage{csquotes}
\usepackage{xcolor}
\usepackage{multirow} 
\usepackage{algorithm}
\usepackage{algpseudocode}
\usepackage{physics}
\usepackage{tikz}
\usepackage{algcompatible}
\usetikzlibrary{quantikz}
\allowdisplaybreaks 
\usepackage{makecell} 
\definecolor{newgreen}{HTML}{228B22}
\definecolor{newred}{HTML}{ff382e}

\newcommand\redsout{\bgroup\markoverwith{\textcolor{red}{\rule[0.5ex]{2pt}{0.4pt}}}\ULon}

\usepackage{palatino}
\newcommand{\numtargates}{M}
\newcommand{\sqzrot}{R_Z\left(\frac{\pi}{2^\ell}\right)}
\newcommand{\addrsqzrot}{R_Z\left(\frac{\pi}{2^\ell}\right)_A}
\newcommand{\mczrot}{C^{(m-1)}\left(R_Z\left(\frac{\pi}{2^{\ell}} \right) \right)}
\newcommand{\czrot}{C\left(R_Z\left(\frac{\pi}{2^{\ell}} \right) \right)}
\newcommand{\addmczrot}{C^{(m-1)}_A\left(R_Z\left(\frac{\pi}{2^{\ell}}\right)\right)}
\newcommand{\addczrot}{C_A\left(R_Z\left(\frac{\pi}{2^{\ell}}\right)\right)}

\renewcommand{\phi}{{\varphi}}

\def\<{\langle}
\def\>{\rangle}

\theoremstyle{plain} 
\newtheorem{theorem}{Theorem}
\newtheorem{lemma}[theorem]{Lemma}

\newtheorem{proposition}[theorem]{Proposition}

\newcommand{\tn}[1]{\textnormal{#1}}

\newcommand{\custlogicalU}[1]{U_{\mathrm{L}}^{(#1)}}

\def\diag{\textnormal{diag}}
\newcommand{\bracket}[1]{\left(#1\right)}

\newcommand{\cs}{\text{CS}}
\newcommand{\cz}{\text{CZ}}
\newcommand{\ccz}{\text{CCZ}}
\newcommand{\supp}{\operatorname{Supp}}
\newcommand{\physicalU}{U_{\mathrm{P}}}
\newcommand{\logicalU}{U_{\mathrm{L}}}
\newcommand{\globalphase}{\gamma_{\mathrm{P}}}
\newcommand{\overlinelogicalU}{\overline{U}_{\mathrm{L}}}
\newcommand{\phaseP}[1]{\varphi_{{#1}}}
\newcommand{\phaseL}[1]{\lambda_{{#1}}}
\theoremstyle{definition} 
\newtheorem{definition}{Definition}

\newtheorem{remark}{Remark}
\newtheorem{example}{Example}
\newtheorem*{example*}{Example}
\newtheorem*{remark*}{Remark}

\renewcommand{\thefootnote}{\fnsymbol{footnote}}

\newcommand{\benu}{\begin{enumerate}\setlength\itemsep{4pt}}
\begin{document}

\title{Realizing Logical Diagonal Gates \\ via Transversal Physical $Z$-Rotations in CSS Codes$^\ast$}

\author{\NoCaseChange{
    K. Sai Mineesh Reddy and Navin Kashyap \\
    \textnormal{Dept.\ Electrical Communication Engineering, Indian Institute of Science} \\
    \textnormal{\{mineeshk,nkashyap\}@iisc.ac.in}
}}

\maketitle

\begin{abstract}
Calderbank-Shor-Steane (CSS) codes, constructed from nested classical codes $C_2 \subseteq C_1$, are typically optimized for good code parameters. However, practical quantum computing equally demands fault-tolerant logical gates. In this work, we characterize nested pairs $(C_1, C_2)$ whose resulting CSS codes realize a target logical diagonal gate via transversal physical $Z$-rotations. In doing so, we recover a result of Camps-Moreno et al.\ that CSS codes can realize only logical single-qubit $Z$-rotations and multi-qubit controlled $Z$-rotations via transversal physical $Z$-rotations. Building on our characterization, we develop the ``appending construction'', that takes as input an $[[n',k']]$ CSS code $\mathcal{Q}'$ and a target logical $Z$-rotation (single-qubit or multi-controlled) $U_L$, and extends $\mathcal{Q}'$ by systematically appending $n''$ physical qubits to obtain an $[[n,k]]$ CSS code $\mathcal{Q}$ with $n = n'+n''$ and $k=k'$. The target logical gate $U_L$ is realized in $\mathcal{Q}$ by applying a well-chosen physical transversal $Z$-rotation to the $n''$ appended physical qubits. Moreover, any logical gate realized via transversal physical $Z$-rotations in the input code $\mathcal{Q}'$ remains transversally realizable in the extended code $\mathcal{Q}$. The CSS code $\mathcal{Q}$ may incur a loss in minimum distance, but the loss can be controlled through the parameter choices made in the construction. By repeatedly applying the appending construction, we can extend any CSS code $\mathcal{Q}'$ to obtain a CSS code $\mathcal{Q}$ that supports fault-tolerant implementations of multiple desired logical $Z$-rotations. The cost to be paid for this is the increased physical qubit overhead as the number of target logical gates grows. To illustrate our methodology, we construct CSS code families that transversally realize addressable logical single-qubit and multi-controlled $Z$-rotations. 
\end{abstract}
\footnotetext[1]{This work was supported in part by the MATRICS grant ANRF/ARGM/2025/000814/MTR awarded by the Anusandhan National Research Foundation (ANRF).}
\renewcommand{\thefootnote}{\arabic{footnote}}

\section{Introduction}\label{sec:introduction}

A central challenge in scalable quantum computing is designing quantum error-correcting codes~\cite{S1995} that simultaneously possess high encoding rates, robust error correction, and diverse fault-tolerant logical gates. Formally, a quantum code $\mathcal{Q}$ with parameters $[[n,k,d]]_2$ encodes $k$ logical qubits (logical dimension) into $n$ physical qubits (physical blocklength) and corrects Pauli errors of weight up to $\lfloor \frac{d-1}{2} \rfloor$. A logical operator is a unitary applied to the $n$ physical qubits that preserves the code space $\mathcal{Q}$, effectively executing a logical gate on the $k$ encoded qubits.

Because physical gates are inherently noisy, logical operators must be implemented fault-tolerantly. A canonical strategy for achieving this is the use of transversal gates (see~\cite[Chapter~5]{G1997}). When a quantum system is partitioned into code blocks with independent error-correcting capabilities, a unitary is called transversal if it never couples physical qubits within the same block. This constraint ensures that physical gate failures remain localized: qubit errors cannot propagate within a block and are reliably corrected by the underlying code block. In this work, we adopt the simplest and most practical notion of transversality, wherein logical operators are implemented via tensor products of single-qubit unitaries. 

Thus, an ideal objective is to find quantum codes capable of realizing a universal logical gate set entirely via transversal physical unitaries, but the Eastin-Knill theorem~\cite{EK2009} forbids this. Quantum gates are typically partitioned into gates belonging to the Clifford group---generated by the Hadamard ($H$), Phase ($S$), and controlled-$Z$ ($\cz$) gates---and non-Clifford gates. Because the Clifford group along with any single non-Clifford gate constitutes a universal gate set (see Section~\ref{sec:preliminaries_and_notation}), no quantum code can transversally implement both simultaneously. 

A natural relaxation is to find codes that transversally realize logical Clifford gates~\cite{TTR2026,BB2024,TTF2025}. For instance, puncturing an $[n,n/2,d]$ classical self-dual code yields an $[[n-1, 1, \ge d-1]]$ Calderbank-Shor-Steane (CSS) code transversally supporting logical $H$, $S$, and logical $\cz$ across any two such code blocks~\cite{JA2025}. Considering $\ell$ such code blocks produces an $[[(n-1) \ell, \ell, \ge d-1]]$ code that transversally realizes the full logical Clifford group, meaning that logical $H$ and $S$ on each logical qubit, and logical $\cz$ between any pair of logical qubits. Instantiating this with asymptotically good classical self-dual codes~\cite{SMT1972,PW2007} for $n$ blocks yields a CSS family realizing the full logical Clifford group wherein both the number of logical qubits and minimum distance scale as the square root of the number of physical qubits.  However, by the Gottesman-Knill theorem~\cite{NC2010}, Clifford-only circuits are efficiently simulatable on classical computers. Consequently, realizing only the Clifford group is insufficient for quantum advantage.

Typically, a diagonal gate such as the $T$ (or the controlled-controlled-$Z$, i.e., $\ccz$) gate is selected as the required non-Clifford gate. Consequently, the set $\{H, \cz, T\}$ (as $S=T^2$) constitutes a universal gate set. In practice, logical Clifford gates are executed transversally within one code, while non-Clifford gates are implemented via protocols like magic state distillation~\cite{CTV2017,KT2019} or code switching~\cite{ADP2014}. However, a crucial observation is that standard universal gate sets consist entirely of diagonal gates supplemented by a superposition-creating Hadamard gate. This structure features prominently in major quantum algorithms, including Shor's factoring~\cite{S1997} and Grover's search~\cite{G1996} algorithms. Therefore, our methodology focuses on constructing codes that transversally realize a diverse set of logical diagonal gates, relying on code switching solely for the logical Hadamard gate. Because practical algorithms demand a rich variety of diagonal gates, this framework offers a viable paradigm for fault-tolerant quantum computation. 

\subsection{Literature Review}As motivated above, characterizing quantum codes that transversally implement logical diagonal gates is crucial for fault-tolerant quantum computing. Prior works~\cite{WQAS2023,AT2016} investigate logical diagonal gates realized transversally within stabilizer codes. For CSS codes specifically, constructions derived from classical divisible codes transversally realize logical diagonal gates~\cite{HLC2022,HLC2025, HLC2022div,Haah_2018,BC2026,BCHDK2026}. Additionally, several studies establish characterizations, constructions, and necessary conditions for specific non-Clifford gates, primarily the $T$ gate~\cite{CYCXSZ2026,MN2025,RCNP2020, RCNPisit2020, BH2012, BDMJL2025, CLMRSS2024,HH2018,BADEFSM2025,HHccz18, CHGE2024,BCR2024,AOHEM2026,PCR2025,BHMR2026,CHGDRI2024,W2026,SJ2026,LWC2026,AM2026} and multi-controlled-$Z$ gates such as the $\ccz$ gate~\cite{GG2025,N2025,WHY2025,MBJMRADL2026,JCD2026,GPRW2026}. Furthermore, recent works have explored either the transversal or constant-depth realization of logical non-Clifford gates---especially multi-controlled $Z$-rotations---within quantum LDPC codes~\cite{GT2024,THLGH2025,CT2021,ET2026,TB2026,GVG2025,YZZQ2025,LLL2026,LLL2026a,GTZ2026,LWHL2026}. Note that in some of these works, the stabilizer weights grow sub-linearly, meaning they do not strictly satisfy the LDPC criterion. However, a limitation of several of these families is their inability to apply targeted logical multi-controlled-$Z$ gates; there is no flexible mechanism to dictate exactly which logical qubits participate in the multi-controlled-$Z$ operation. Moreover, whenever these logical gates are realized transversally, they are restricted to inter-block operations rather than intra-block operations.

While the aforementioned literature primarily focuses on algebraic constructions, another important line of research investigates topological codes. For instance, $2$-D doubled color codes transversally realize the logical $T$ gate on a single logical qubit~\cite{BA2015}. Recently, this doubling technique was utilized to construct a class of quantum color codes encoding a single logical qubit that transversally realize any fixed logical single-qubit $Z$-rotation on the encoded qubit~\cite{DOCJ2026}. Furthermore, higher-dimensional color codes also transversally support logical diagonal gates, specifically the $T$ gate~\cite{B2015,KB2015}, the controlled-$S$ ($\cs$) gate~\cite{B2025}, and the $\ccz$ gate~\cite{PR2013,KYP2015,VB2019,Z2025,ZSPWB2025}. 

In this work, we characterize CSS codes that realize a fixed target logical diagonal gate via transversal physical $Z$-rotations. Camps-Moreno et al.~\cite{CLMRS2026} independently derived a closely related characterization. However, our methodologies differ: given a specific transversal physical $Z$-rotation, their approach determines whether it acts as a logical operator and computes the resulting logical diagonal gate using the code's logical-$X$ basis. In contrast, our target-driven approach \emph{a priori} fixes both the target logical diagonal gate and the transversal physical $Z$-rotation to determine the exact conditions required of the CSS code. 

Furthermore, we specifically focus on realizing \emph{addressable} logical gates, which selectively target specific subsets of logical qubits. Addressability is crucial for practical quantum algorithms, which require localized gates rather than gates acting on all encoded qubits simultaneously. Consequently, the transversal realization of addressable logical gates has become a major focus in recent literature~\cite{W2026,AM2026,SJ2026,LWC2026,GJ2026,PB2025,QWV2023,HVAR20251,HVAR2025,G2025,CB2026}, as well as in our current work. 


\subsection{Our Main Contributions} 
We now highlight the primary contributions of our paper.

\subsubsection{Characterization of CSS Codes That Realize Logical Diagonal Gates Via Transversal $Z$-Rotations}
Restricting physical gates to dyadic transversal $Z$-rotations, we characterize the CSS codes capable of realizing a target logical diagonal gate. As defined in Sections~\ref{sec:phy_tran_css_notation} and~\ref{sec:phy_tran_css}, an $n$-qubit dyadic transversal $Z$-rotation $U(p,w)$ is a diagonal gate characterized by a non-negative integer $p$ and an integer vector $w \in \mathbb{Z}^n$. We provide, in Theorem~\ref{thm:refined_fault_tolerant_logical_diag_gates}, a precise characterization of when $U(p,w)$ realizes a given logical diagonal gate within a specific CSS code. The theorem can be loosely paraphrased as follows:
\begin{quote}   
Let $\mathcal Q$ be an $[[n,k]]$ CSS code defined by a pair of nested binary linear codes $C_2 \subseteq C_1$, and let $\logicalU = \diag \left(\exp{\iota \frac{\pi}{2^{\ell}} f(a)} : a \in \mathbb{F}_2^k \right)$ be a target $k$-qubit logical diagonal gate specified by an integer-valued function $f:\mathbb{F}_2^k \rightarrow \mathbb{Z}$. The transversal physical $Z$-rotation $U(p,w)$ realizes $\logicalU$ if and only if $p \ge \ell$ and certain modular equations hold on the support of the vector $w$. These modular equations involve $p$, the codewords of $C_2$ and the coset space $C_1/C_2$, and the parameters $f, \ell$ that specify $\logicalU$.
\end{quote}

Furthermore, as established by Camps-Moreno et al.~\cite{CLMRS2026} (see also Remark~\ref{rmk:logicalgatestructure}), CSS codes can realize only logical single-qubit $Z$-rotations and multi-controlled $Z$-rotations via transversal physical $Z$-rotations. Consequently, whenever we restrict physical gates to transversal $Z$-rotations, we assume that target gates consist entirely of single-qubit and multi-controlled $Z$-rotations.

Finally, we provide a relaxed characterization of CSS codes that realize a target logical diagonal gate via transversal physical $Z$-rotations. This relaxation allows the logical gate to be realized via transversal $Z$-rotations up to the application of physical $Z$ gates. 

\subsubsection{A Framework to Construct CSS Codes Transversally Realizing A Set of Logical Diagonal Gates}
While our characterization overlaps with prior literature~\cite{CLMRS2026}, our primary contribution lies in leveraging it to propose a systematic constructive methodology: the appending framework. As established in Theorem~\ref{thm:refined_fault_tolerant_logical_diag_gates}, which we have paraphrased above, realizing a target logical gate $\logicalU$ depends only on satisfying specific modular equations on the support of the physical $Z$-rotation's integer vector $w$. This localized dependency directly motivates our framework. 

Given an $[[n,k]]$ primary CSS code and a set of target logical diagonal gates, we systematically append extra coordinates (physical qubits) by attaching auxiliary \emph{appending matrices} to the $X$-stabilizer check matrix and logical-$X$ generator matrix of the primary code. This process yields an extended CSS code wherein each target logical gate is assigned a dedicated subset of the newly appended coordinates. Realizing a specific target gate in this extended code involves choosing a transversal $Z$-rotation that acts exclusively on the physical qubits associated with the target logical gate's dedicated subset of extra coordinates. Furthermore, all logical diagonal gates realized via transversal physical $Z$-rotations in the primary code remain transversally realizable in the extended code.


The modularity of the appending framework allows the primary code to be flexibly extended to transversally support any desired set of logical diagonal gates. By construction, the extended code preserves the logical dimension $k$. While this extension may incur a loss in minimum distance, the loss can be controlled via parameter choices made in the appending matrices. However, on the flip side, the total number of physical qubits grows with the number of implemented logical gates. 

\subsubsection{Appending Matrix Constructions for Addressable Single-Qubit and Multi-Controlled $Z$-Rotations}

To construct extended CSS codes capable of realizing any target logical gate achievable via transversal $Z$-rotations, we adopt an inductive construction strategy:

We first establish the appending matrix construction for addressable single-qubit $Z$-rotations, which serves as the base case. Building upon this, we systematically derive the appending matrices for addressable $m$-controlled $Z$-rotations by recursively leveraging the matrices for lower-order $\widetilde{m}$-controlled $Z$-rotations and single-qubit $Z$-rotations, where $\widetilde{m} < m$. 

Consequently, given any primary CSS code and a desired set of achievable target diagonal gates, we can systematically extend the primary code such that every target logical gate is realized transversally within the extended code. In fact, this approach provides an effective method for designing practical, short-blocklength codes capable of transversally realizing a diverse set of target logical gates.    
\subsubsection{Families of CSS Codes Realizing Addressable Single-Qubit and Multi-Controlled $Z$-Rotations} 

By leveraging our appending matrix constructions, we derive several CSS code families that transversally realize addressable logical single-qubit $Z$-rotations and multi-controlled $Z$-rotations across two distinct regimes.

First, we consider an arbitrary but fixed sequence of target address configurations for a fixed single-qubit $Z$-rotation ($\sqzrot$) or multi-controlled $Z$-rotation ($\mczrot$). In this targeted regime, for sub-linearly growing target address sizes, we construct CSS code families that achieve a linear logical dimension and a polynomially growing minimum distance. Specifically, for any fixed $\sqzrot$, we construct a CSS family achieving parameters $[[n,\Theta(n),\Omega(n^{\eta})]]$ while transversally realizing $\sqzrot$ on any fixed subset of at most $n^\epsilon$ logical qubits, where $0 < \epsilon < 1$. Similarly, for any fixed $\mczrot$, we construct a CSS family achieving parameters $[[n,\Theta(n),\Omega(n^{\eta})]]$ while transversally realizing $\mczrot$ on any fixed subset of at most $n^\epsilon$ logical $m$-qubit tuples. Here, the exponent $0 < \eta < 1$ depends on the specific target gate being realized. Furthermore, for arbitrary target support sizes, we construct similar CSS code families that achieve polynomial scaling for both the logical dimension and the minimum distance. 

Second, we broaden our scope to CSS code families that realize any addressable logical gate for a fixed $\sqzrot$ or $\mczrot$. In this regime, the resulting CSS code families similarly achieve a polynomial logical dimension and minimum distance. 

Note that across both regimes, we provide both explicit and non-explicit code families; the non-explicit families yield improved code parameters. The exact code parameters, along with the corresponding logical gates, are summarized in Tables~\ref{sq_families},~\ref{cz_families}, and~\ref{mcz_families}.

More generally, utilizing the flexibility of the appending framework, given any fixed, constant number of desired addressable logical single-qubit and multi-controlled $Z$-rotations, we can construct a single CSS code family that achieves polynomial parameter scaling while ensuring all target logical gates are realized transversally.

We now compare our constructed code families with the literature. Recent algebraic constructions~\cite{HVAR20251,HVAR2025,G2025} achieve asymptotically good parameters while transversally realizing addressable logical controlled-controlled-$Z$ $(\ccz)$ gates. However, these constructions realize logical inter-block $\ccz$ gates by applying physical $\ccz$ gates across code blocks. Furthermore, they realize intra-block logical $\ccz$ gates via physical gates which are constant-depth but not transversal.

At the expense of asymptotically good parameter scaling, we construct code families that achieve a polynomial logical dimension and minimum distance while realizing intra-block multi-controlled $Z$-rotations entirely via transversal physical single-qubit $Z$-rotations. To the best of our knowledge, these are the first such algebraically constructed code families. We also note a concurrent quantum lego framework~\cite{CB2026} that similarly constructs CSS codes supporting transversal intra-block multi-controlled $Z$-rotations. However, this lego framework is distinct from our appending framework.

The appending framework is designed primarily for operational flexibility, aiming to realize as many target logical diagonal gates within a single code as possible. Due to this generalized approach, the framework does not always yield the best parameter scaling. For instance, while recent works~\cite{SJ2026,LWC2026,AM2026} construct asymptotically good codes capable of transversally realizing any addressable $\sqzrot$ gate for a fixed $\ell$, our single-qubit constructions do not achieve these asymptotic parameters.

\subsection{Outline}
The paper is organized as follows. Section~\ref{sec:preliminaries_and_notation} establishes the notation and preliminaries. Section~\ref{sec:characterizations} parametrizes the transversal $Z$-rotations that preserve a CSS code and characterizes CSS codes that realize a target logical diagonal gate via transversal physical $Z$-rotations. Building on this, Section~\ref{app_framework} introduces our general appending framework. Section~\ref{ft_sq} provides appending-matrix constructions for addressable logical single-qubit $Z$-rotations. For clarity of exposition, we focus on appending-matrix constructions for $1$-controlled $Z$-rotations in Section~\ref{ft_cs} of the main text, deferring constructions for arbitrary multi-controlled $Z$-rotations to Appendix~\ref{ft_mc_zrot}. Section~\ref{css_families} applies these techniques to construct families of CSS codes that transversally realize addressable logical single-qubit and $1$-controlled $Z$-rotations. 
The paper has several appendices which contain proofs deferred from the main text.

\section{Notation and Preliminaries}\label{sec:preliminaries_and_notation}

To lay the groundwork for our analysis, this section formalizes the core mathematical notation and preliminary concepts used throughout this work.

\subsection{Pauli Operators and the Clifford Hierarchy}

We begin with the single-qubit Pauli operators:
\begin{align*}
    I_2 := \begin{bmatrix}
        1 & 0 \\ 0 & 1
    \end{bmatrix}, \qquad X := \begin{bmatrix}
        0 & 1 \\ 1 & 0
    \end{bmatrix}, \qquad Z := \begin{bmatrix}
        1 & 0 \\ 0 & -1
    \end{bmatrix}, \qquad Y := \iota XZ = \begin{bmatrix}
        0 & -\iota \\ \iota & 0
    \end{bmatrix},
\end{align*}
where $\iota = \sqrt{-1}$. These operators are Hermitian, unitary, and involutory (i.e., $X^2=Y^2=Z^2 = I_2$).

Let $\mathbb{F}_2$ denote the binary field. For any positive integer $n$ and binary vectors $a = (a(1), \ldots, a(n)), b = (b(1), \ldots, b(n)) \in \mathbb{F}_2^{n}$, the $n$-qubit Pauli operators are defined via the tensor product
\begin{equation*}
    E(a,b) :=  \left( \iota^{a(1) b(1)} X^{a(1)} Z^{b(1)} \right) \otimes \cdots \otimes \left( \iota^{a(n) b(n)} X^{a(n)} Z^{b(n)} \right).
\end{equation*}
The $n$-qubit Pauli group is defined as the set
\begin{equation*}
    \mathcal{P}_n := \left\{ \iota^{\kappa} E(a, b) : a, b \in \mathbb{F}_2^n, \kappa \in \{0,1,2,3\} \right\}.
\end{equation*}
Like their single-qubit counterparts, the operators $E(a,b)$ are Hermitian, unitary, and involutory. 

Setting $N=2^n$, let $\mathbb{U}_N$ represent the group of all $N \times N$ unitaries. The Clifford hierarchy is defined recursively, starting with the Pauli group as its first level: $ \mathcal{C}^{(1)} := \mathcal{P}_n$. For $\ell \geq 2$, the $\ell$-th level is defined as the set of unitaries that map the Pauli group into the preceding level under conjugation:
\begin{equation*}
    \mathcal{C}^{(\ell)}:= \left\{   U \in \mathbb{U}_N : U P U^{\dagger} \in \mathcal{C}^{(\ell-1)} \; \forall \; P \in \mathcal{P}_n  \right\}.
\end{equation*}
The second level, $\mathcal{C}^{(2)}$, constitutes the Clifford group. It is generated by the Hadamard $(H)$, Phase $(S)$, and Controlled-NOT $(\text{CNOT})$ gates:
\begin{align*}
   H  := {\frac{1}{\sqrt{2}}} \begin{bmatrix}
      1 & 1 \\ 1 & -1  
    \end{bmatrix},  \qquad S := \begin{bmatrix}
        1 & 0 \\ 0 & \iota
    \end{bmatrix},  \qquad \text{CNOT} := \begin{bmatrix}
        I_2 & 0 \\ 0 & X
    \end{bmatrix}.
\end{align*}
The controlled-$Z$ $(\text{CZ})$ gate is also a Clifford gate given by\footnote{The set $\{H,S, \text{CZ}\}$ also generates the Clifford group.}
\begin{equation*}
    \text{CZ} := \left( I_2 \otimes H \right)  \cdot \text{CNOT} \cdot \left( I_2 \otimes H \right) =  \begin{bmatrix}
        I_2 & 0 \\ 0 & Z
    \end{bmatrix}.
\end{equation*}
As is well-established, the Clifford group supplemented with any unitary from level $\ell \geq 3$ of the Clifford hierarchy forms a universal gate set for quantum computation. 

\subsection{Binary Linear Codes and Asymptotics}

We next review standard coding theory notation. The Hamming weight of a binary vector $x$ is denoted by $w_H(x)$. For a binary linear code $C$, we denote its dimension by $\dim(C)$ and its minimum distance by $d_{\min}(C)$. A code $C$ is called $m$-divisible if every codeword $x \in C$ has Hamming weight $w_H(x) = 0 \pmod{m}$. 

The Schur product of two vectors $x = (x(1), \ldots, x(n))$ and $y=(y(1), \ldots, y(n))$ in $\mathbb{F}_2^n$ is defined as their component-wise multiplication:
\begin{equation*}
    x \ast y := (x(1) y(1), \ldots, x(n) y(n)).
\end{equation*}
This operation extends naturally to codes. The Schur product of two codes $C_1$ and $C_2$ is defined as the $\mathbb{F}_2$-linear span of the pairwise Schur products of their respective codewords:
\begin{equation*}
    C_1 \ast C_2 := \langle x \ast y : x \in C_1, y \in C_2 \rangle.
\end{equation*}

Let $\mathbf{1}\{ \cdot \}$ denote the indicator function and $[n] := \{ 1, 2, \ldots, n \}$. We denote the standard basis of $\mathbb{F}_2^n$ by $\{e_i : i \in [n]\}$, where $e_i$ has a $1$ at the $i$-th coordinate and $0$ elsewhere. The support of a vector $x \in \mathbb{F}_2^n$ specifies the coordinates of its non-zero entries:
\begin{equation*}
    \supp(x) := \{ i \in [n] : x(i) \neq 0 \}.
\end{equation*}
Extending this definition to tensors, the support of an order-$m$ binary tensor $A \in \mathbb{F}_2^{n \times \cdots \times n}$ is given by the coordinate tuples of its non-zero entries:
\begin{equation*}
    \supp(A) := \{ (i_1, \ldots, i_m) \in [n] \times \cdots \times [n] : A(i_1, \ldots, i_m) \neq 0 \}.
\end{equation*}

Using inclusion-exclusion and induction, it can be shown that the binary linear combination $\bigoplus_{i=1}^k a(i) x_i$ (where $a(i) \in \mathbb{F}_2$ and $x_i \in \mathbb{F}_2^n$) can be expanded as an arithmetic sum over $\mathbb{Z}$ as follows:
\begin{equation*}
    \bigoplus_{i=1}^{k} a(i) x_i = \sum_{i=1}^{k} (-2)^{i-1}  \sum_{1 \leq j_1 < \cdots < j_i \leq k} \left( \prod_{p=1}^{i} a({j_p}) \right)  \left( x_{j_1} \ast \cdots \ast x_{j_i} \right),
\end{equation*}
where $\bigoplus$ and $\sum$ denote modulo-$2$ and integer addition, respectively. By the linearity of the dot product $w \cdot x = \sum_{i=1}^n w_i x_i$, for any integer vector $w \in \mathbb{Z}^n$, we have
\begin{equation*}
    w \cdot \left( \bigoplus_{i=1}^{k} a(i) x_i \right) = \sum_{i=1}^{k} (-2)^{i-1} \sum_{1 \leq j_1 < \cdots < j_i \leq k} \left( \prod_{p=1}^{i} a({j_p}) \right) w \cdot (x_{j_1} \ast \cdots \ast x_{j_i}).
\end{equation*}
Let $\mathbf{1}_n$ and $\mathbf{0}_n$ denote the all-ones and all-zeros vectors of length $n$, respectively. Setting $w = \mathbf{1}_n$ yields the Hamming weight
\begin{equation*} 
    w_H \left( \bigoplus_{i=1}^{k} a(i) x_i \right) = \sum_{i=1}^{k} (-2)^{i-1} \sum_{1 \leq j_1 < \cdots < j_i \leq k} \left( \prod_{p=1}^{i} a({j_p}) \right) w_H(x_{j_1} \ast \cdots \ast x_{j_i}).
\end{equation*}
Crucially, if $C$ is a $2^m$-divisible code, the preceding equation implies that for any $\ell \in [k]$ and codewords $x_1, \ldots, x_\ell \in C$, we have
\begin{equation*}
    w_H(x_1 \ast \cdots \ast x_\ell) = 0 \pmod{2^{m-\ell+1}}.
\end{equation*}

To evaluate the asymptotic scaling of code parameters, we adopt the Bachmann-Landau order notations $\mathcal{O}$, $\Omega$, $\Theta$, $o$, and $\omega$ as defined in~\cite{CLRS2009}.

\subsection{CSS Construction and Encoding}

Consider a nested pair of binary linear codes $C_2 \subseteq C_1$ with parameters $[n, k_2]$ and $[n, k_1]$, respectively. 

The stabilizer group $\mathcal{S}$, generated by the $X$-type operators $E(x, 0)$ for $x \in C_2$ and the $Z$-type operators $E(0, z)$ for $z \in C_1^\perp$, defines an $[[n, k]]$ quantum code, where $k := k_1 - k_2$. The resulting CSS code is denoted by $(C_1, C_2)_{\tn{CSS}}$ or $\mathcal{Q}_{\mathcal{S}}$. Its minimum distance is given by $d_{\min}(\mathcal{Q}_{\mathcal{S}}) := \min\{d_X, d_Z\}$, where the $X$- and $Z$-distances are given by
\begin{equation*}
    d_X := \min_{x \in C_1 \setminus C_2} w_H(x) \quad \text{and} \quad d_Z := \min_{z \in C_2^\perp \setminus C_1^\perp} w_H(z). 
\end{equation*}
Here, $C_1 \setminus C_2$ denotes the set difference, distinct from the coset space $C_1 / C_2$. The projector $P_{\mathcal{S}}$ onto the code space $\mathcal{Q}_{\mathcal{S}}$ is defined as
\begin{equation*}
    P_{\mathcal{S}} := \frac{1}{|C_2|} \left( \sum_{x \in C_2} E(x, 0) \right) \frac{1}{|C_1^\perp|} \left( \sum_{z \in C_1^\perp} E(0,z) \right).
\end{equation*}
Since the $Z$-type operators preserve $\ket{0^n}$, applying $P_{\mathcal{S}}$ to this state yields the logical all-zero basis state, $\ket{0^k}$:
\begin{align*}
    {\ket{0^k}}_L &:= \sqrt{|C_2|} P_{\mathcal{S}} \ket{0^n} \\
    &= \frac{1}{\sqrt{|C_2|}} \sum_{x \in C_2} \ket{x}.
\end{align*}

To construct all the remaining logical basis states $\ket{a}$ for $a \in \mathbb{F}_2^{k} \setminus \{0^k\}$, we first define logical-$X$ operators. Any non-zero coset representative $x \in C_1 / C_2$ acts as a non-trivial logical-$X$ operator $E(x,0)$. Therefore, we fix an arbitrary basis $\{ y_1, \ldots, y_k \}$ for the coset space $C_1 / C_2$.\footnote{Throughout this paper, the choice of coset representatives is not unique; any such choice should be regarded as arbitrary but fixed.} Employing  the standard overline notation $\overlinelogicalU$ to represent a physical operator implementing a corresponding logical gate $\logicalU$ on the code space, the logical-$X$ operator acting as the Pauli-$X$ gate on the $i$-th logical qubit is defined as   
\begin{equation*}
    \overline{X}_i := E(y_i,0).
\end{equation*}
For any $a = (a(1), \ldots, a(k)) \in \mathbb{F}_2^k$, the multi-qubit logical-$X$ gate acting over the qubits within the support of $a$ is obtained by the product
\begin{equation*}
    \overline{X}_1^{a(1)} \cdots \overline{X}_k^{a(k)} = E(y_a,0),
\end{equation*}
where $y_a := \bigoplus_{i=1}^k a(i) y_i$ is the coset representative corresponding to the vector $a$. Applying $E(y_a,0)$ to the logical all-zero state yields the target logical basis state $\ket{a}$:
\begin{align*}
    \ket{a}_L &:= \overline{X}_1^{a(1)} \cdots \overline{X}_k^{a(k)} \ket{0^k}_L = E(y_a, 0) \ket{0^k}_L \\
    &= \frac{1}{\sqrt{|C_2|}} \sum_{x \in C_2} \ket{x \oplus y_a}.
\end{align*}

\subsection{Diagonal Gates and $Z$-Rotations}\label{sec:phy_tran_css_notation}

We first formalize the notation for diagonal unitaries.
\begin{definition}\label{defn:diagonalgate}
    An $n$-qubit unitary operator $U$ is \emph{diagonal} if it acts on the computational basis as
    \begin{equation*}
        U\ket{x} = \exp{ \iota \varphi_x } \ket{x}, \quad \forall \ x \in \mathbb{F}_2^n,
    \end{equation*}
    where $\varphi_x \in \mathbb{R}$. We compactly denote such an operator by
    \begin{equation*}
        U = \diag \left( \exp{  \iota \varphi_x } : x \in \mathbb{F}_2^n \right).
    \end{equation*}
\end{definition}
\begin{remark}\label{rmk:globalphase}    
    Without loss of generality, we set $\varphi_{0^n}=0$, ensuring that $U \ket{0^n} = \ket{0^n}$. Any non-zero $\varphi_{0^n}$ introduces a phase factor $\exp{ \iota  \varphi_{0^n}}$, which can be factored out as a global phase.
\end{remark}

Because the transversal gates we consider are tensor products of single-qubit unitaries, we first define single-qubit diagonal gates, namely $Z$-rotations.
\begin{definition}
    A single-qubit $Z$-rotation by an angle $\theta \in [0, 2\pi)$ is defined by the unitary operator
    \begin{equation*}
        R_Z(\theta) = \begin{bmatrix}
        1 & 0 \\
        0 & \exp{\iota \theta}
        \end{bmatrix}.
    \end{equation*}
    We term this a \emph{dyadic $Z$-rotation} if the phase $\exp{\iota \theta}$ is a $2^\ell$-th root of unity (i.e., $\theta = \frac{2\pi k}{2^\ell}$ for some integer $k$) for a non-negative integer $\ell$. 
\end{definition}
Notably, for $\ell \geq 2$, the dyadic $Z$-rotation $R_Z\left(\frac{\pi}{2^\ell} \right)$ is a non-Clifford gate. We next formalize multi-controlled $Z$-rotations.
\begin{definition}\label{def:multi_controlled_z}
    For any positive integer $m$ and non-negative integer $\ell$, the $m$-controlled $Z$-rotation $C^{(m)}\left( R_Z \left(\frac{\pi}{2^\ell} \right) \right)$ is an $(m+1)$-qubit unitary that applies a phase of $\exp{\iota \frac{\pi}{2^\ell}}$ to a computational basis state if and only if all $m$ control qubits and the target qubit are in the $\ket{1}$ state.
\end{definition}

Since the operator $C^{(m)}\left(R_Z \left(\frac{\pi}{2^\ell} \right) \right)$ applies its phase exclusively to the all-ones state $\ket{1}^{\otimes (m+1)}$, the target and control qubits are completely interchangeable, rendering the action of the unitary invariant under any permutation of its input qubits. To specify its action on a particular subset of an ordered $k$-qubit system ($k \ge m+1$), we introduce the indexed notation $C_{i_1, \ldots, i_{m+1}}^{(m)}\left( R_Z\left(\frac{\pi}{2^\ell} \right) \right)$, where $1 \le i_1 < \cdots < i_{m+1} \le k$. By convention, we designate the highest-indexed qubit, $i_{m+1}$, as the target, while the first $m$ indices, $i_1,\ldots,i_m$, identify the control qubits.

\subsection{Addressable $Z$-Rotations}We conclude this section by formalizing the notation for the addressability of single-qubit and multi-controlled $Z$-rotations.
\begin{definition}\label{def:addressed_single_qubit}    
For a given binary vector $A \in \mathbb{F}_2^k$ and a dyadic $Z$-rotation $\sqzrot$, we define the unitary $\addrsqzrot$ via the following $k$-fold tensor product:
\begin{equation*}
      \addrsqzrot := \bigotimes_{i=1}^k R_Z\left( \frac{\pi}{2^\ell} A(i)\right) = \diag \left( \exp{\iota \frac{\pi}{2^\ell}  w_H(a \ast A) : a \in \mathbb{F}_2^k}\right).
\end{equation*}
We refer to $A$ as the \emph{address vector} and $\addrsqzrot$ as the corresponding \emph{addressed gate}.
\end{definition}

\begin{definition}\label{def:addressed_multi_controlled} 
    Given an order-$(m+1)$ binary tensor $A \in \mathbb{F}_2^{k \times \cdots \times k}$ and an $m$-controlled $Z$-rotation $C^{(m)}\left(R_Z\left(\frac{\pi}{2^\ell}\right)\right)$, we define the gate  $C_A^{(m)}\left(R_Z\left(\frac{\pi}{2^\ell}\right)\right)$ as the following product over all ordered $(m+1)$-qubit subsets:
    \begin{equation*}
        C_A^{(m)}\left(R_Z\left(\frac{\pi}{2^\ell}\right)\right) := \prod_{1 \leq i_1 < \cdots < i_{m+1} \leq k} C^{(m)}_{i_1, \ldots, i_{m+1}} \left( R_Z\left(\frac{\pi}{2^{\ell}} A\left(i_1, \ldots, i_{m+1} \right) \right) \right).
    \end{equation*}    
    The tensor $A$ is referred to as the \emph{address tensor}, and the unitary $C_A^{(m)}\left(R_Z\left(\frac{\pi}{2^\ell}\right)\right)$ is the corresponding \emph{addressed gate}.
    Because the product is exclusively over $(m+1)$-tuples $(i_1,\ldots,i_{m+1})$ that are strictly increasing, the gate's action depends only on the tensor entries $A(i_1, \ldots, i_{m+1})$ for which $1 \le i_1 < \cdots < i_{m+1} \le k$. To eliminate ambiguity and to ensure that $A$ is uniquely identified for an $m$-controlled $Z$-rotation, we enforce the convention that $A(i_1, \ldots, i_{m+1})=0$ unless we have $1 \le i_1 < \cdots < i_{m+1} \le k$.
\end{definition}
\begin{example}
    To illustrate this framework, consider a $4$-qubit register where the two-qubit interactions are controlled-$S$ gates: 
    \begin{equation*}
        \text{CS} := C^{(1)}\left( R_Z\left(\frac{\pi}{2} \right) \right).
    \end{equation*}
    Suppose the target addressability configuration is specified by the matrix $A \in \mathbb{F}_2^{4 \times 4}$, whose only non-zero entries are $A(1,2) = A(2,3) = A(2,4) = 1$. The resulting addressed unitary, $\cs_A = C_A^{(1)} \left( R_Z\left(\frac{\pi}{2} \right) \right)$, evaluates to
    \begin{equation*}
        \text{CS}_{A} = \text{CS}_{1,2} \cdot \text{CS}_{2,3} \cdot \text{CS}_{2,4},
    \end{equation*}
    where $\text{CS}_{i,j} = C_{i,j}^{(1)}\left( R_Z\left( \frac{\pi}{2} \right) \right)$. The equivalent quantum circuit for this addressed operation is depicted below:
    
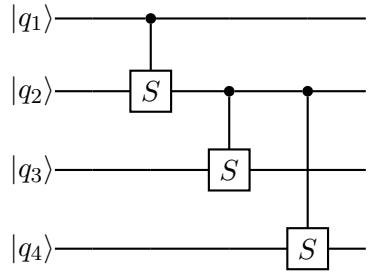
\begin{equation*}
        \begin{quantikz}
            \qw & \ket{q_1} & \qw & \ctrl{1} & \qw      & \qw      & \qw \\
            \qw & \ket{q_2} & \qw  & \gate{S} & \ctrl{1} & \ctrl{2} & \qw \\
            \qw & \ket{q_3} & \qw  & \qw      & \gate{S} & \qw      & \qw \\
            \qw & \ket{q_4} & \qw  & \qw      & \qw      & \gate{S} & \qw
        \end{quantikz}
    \end{equation*}    
\end{example}

\section{Characterizations of CSS Codes That Realize Logical Diagonal Gates Fault-Tolerantly}\label{sec:characterizations}

In this section, we characterize CSS codes that support the transversal implementation of a target logical diagonal gate.

\subsection{Physical Diagonal Gates Realizing Logical Diagonal Gates}

The following theorem establishes the structural conditions under which a CSS code realizes a target logical diagonal gate through a fixed physical diagonal gate. Our proof leverages techniques introduced in~\cite{MN2025} and independently in~\cite{CLMRS2026}.
\begin{theorem}\label{thm:diagonalgates}
    Let $(C_1, C_2)_{\tn{CSS}}$ be an $[[n,k]]$ CSS code. An $n$-qubit physical diagonal gate
    \begin{equation*}
        \physicalU = \diag \left( \exp{{\iota  \phaseP{x}} } : x \in \mathbb{F}_2^n \right)
    \end{equation*}
    realizes the $k$-qubit logical diagonal gate
    \begin{equation*}
        \logicalU = \diag \left( \exp{\iota \phaseL{a} } : a \in \mathbb{F}_2^k \right)
    \end{equation*}
    if and only if for all $x \in C_2$ and $a \in \mathbb{F}_2^k$, the following holds:
    \begin{equation}\label{eq:diagonalgates}
        \exp{\iota \phaseP{x \oplus y_a}} = \exp{\iota \phaseL{a}},
    \end{equation}
    where $y_a \in C_1 / C_2$ is the logical-$X$ coset representative corresponding to the logical state $\ket{a}$.
\end{theorem}
\begin{proof}
    The proof of sufficiency proceeds by direct substitution. For any $a \in \mathbb{F}_2^k$, we have
    \begin{align*}
        \physicalU \ket{a}_L &= \frac{1}{\sqrt{|C_2|}} \bracket{ \sum_{x \in C_2} \exp{\iota \phaseP{x \oplus y_a}} \ket{x \oplus y_a} }. 
    \end{align*}
    Applying the condition in~\eqref{eq:diagonalgates}, this simplifies to
    $$
    \physicalU \ket{a}_L \ = \ \frac{1}{\sqrt{|C_2|}}  \left( \sum_{x \in C_2} \exp{ \iota \phaseL{a}} \ket{x \oplus y_a} \right) \ 
        = \ \exp{\iota \phaseL{a}} \ket{a}_L \ 
        = \ \overlinelogicalU \ket{a}_L.
    $$
    We now establish necessity. Suppose that the physical gate $\physicalU$ realizes the target logical gate $\logicalU$. Then, for any $a \in \mathbb{F}_2^k$, we must have
    \begin{equation}\label{eq:necessityaction}
        \physicalU \ket{a}_L = \globalphase \overlinelogicalU \ket{a}_L,
    \end{equation}
    where $\globalphase$ denotes global phase. The left-hand side of~\eqref{eq:necessityaction} expands as
    \begin{equation}\label{eq:lhs_expansion}
        \physicalU \ket{a}_L = \frac{1}{\sqrt{|C_2|}}  \left( \sum_{x \in C_2} \exp{\iota  \phaseP{x \oplus y_a}} \ket{x \oplus y_a} \right),  
    \end{equation}
    while the right-hand side evaluates to
    \begin{align}
        \globalphase \overlinelogicalU \ket{a}_L &= \globalphase \exp{ \iota \phaseL{a}} \ket{a}_L \nonumber \\
        &= \globalphase \exp{ \iota \phaseL{a}} \frac{1}{\sqrt{|C_2|}} \left( \sum_{x \in C_2} \ket{x \oplus y_a} \right)  \label{eq:rhs_expansion}.
    \end{align}
    Equating~\eqref{eq:lhs_expansion} and~\eqref{eq:rhs_expansion} and canceling the normalization factor $1/\sqrt{|C_2|}$, we find that for all $a \in \mathbb{F}_2^k$, 
    \begin{equation*}
        \sum_{x \in C_2} \exp{\iota \phaseP{x \oplus y_a}} \ket{x \oplus y_a}  = \globalphase \exp{\iota \phaseL{a}} \left( \sum_{x \in C_2} \ket{x \oplus y_a} \right).
    \end{equation*}
    Because the states $\{ \ket{x \oplus y_a} : x \in C_2 \}$ form an orthonormal set in $\left(\mathbb{C}^2\right)^{\otimes n}$, equating the coefficients of the corresponding basis states implies that for all $x \in C_2$,
    \begin{equation*}
        \exp{\iota \phaseP{x \oplus y_a}} = \globalphase \exp{\iota  \phaseL{a}}.
    \end{equation*}
    In particular, setting $x = 0^n$ and $a = 0^k$ yields $\globalphase = 1$ (see  Remark~\ref{rmk:globalphase}), completing the proof.
\end{proof}

\subsection{Physical Transversal $Z$-Rotations That Preserve a CSS Code}\label{sec:phy_tran_css}

For the remainder of this paper, we focus exclusively on CSS codes that realize logical diagonal gates transversally.

The following theorem establishes that transversal $Z$-rotations preserving a CSS code space are necessarily dyadic.
\begin{theorem}[\cite{AT2016}, Theorem~1]
    If $\physicalU = \bigotimes_{i=1}^n R_Z(\theta_i)$ is a logical operator on a CSS code space, then each phase $\exp{\iota  \theta_i}$ must be a $2^{l_i}$-th root of unity for some non-negative integer $l_i$.
\end{theorem}

\begin{remark}\label{rmk:transversalzrotations}
    As a consequence of the preceding theorem, we hereafter restrict our attention to dyadic transversal $Z$-rotations of the form
    \begin{equation*}
        \physicalU = \bigotimes_{i=1}^n R_Z\left( \frac{\pi q_i}{2^{p_i}} \right),
    \end{equation*}
    where each $p_i$ is a non-negative integer and each $q_i$ is an odd integer.
\end{remark}

The following lemma demonstrates that every transversal $Z$-rotation can be parameterized by a non-negative integer and an integer vector.
\begin{lemma}\label{lem:z_rotation_parameterization}
    Any $n$-qubit dyadic transversal $Z$-rotation is characterized by a non-negative integer $p$ and an integer vector $w \in \mathbb{Z}^n$ containing at least one odd component, such that
    \begin{equation}\label{eq:z_rotation_diag}
        U = \diag \left(\exp{\iota \frac{\pi}{2^p} (w \cdot x)} : x\in \mathbb{F}_2^n \right).
    \end{equation}
    Moreover, the integer $p$ is uniquely determined, and the vector $w$ is unique modulo $2^{p+1}$.
\end{lemma}

The proof of this lemma is provided in Appendix~\ref{app:z_rotation_parameterization_proof}. Under this parameterization, the diagonal unitary $U$ resides precisely in the $(p+1)$-th level of the Clifford hierarchy, $\mathcal{C}^{(p+1)}$. 

Motivated by the characterization established in Lemma~\ref{lem:z_rotation_parameterization}, we introduce the following formal parameterization for transversal $Z$-rotations.
\begin{definition}\label{def:z_rotation_tuple}
    Any $n$-qubit dyadic transversal $Z$-rotation is uniquely identified by a non-negative integer $p$ and an integer vector $w \in \mathbb{Z}^n$ containing at least one odd component. We denote this operator by the tuple $(p, w)$, explicitly defined as
    \begin{equation*}
        U(p, w) := \diag\left( \exp{\iota \frac{\pi}{2^p} \left( w \cdot x\right)} : x \in \mathbb{F}_2^n \right).
    \end{equation*}
\end{definition}

\subsection{CSS Codes Realizing Logical Diagonal Gates Fault-Tolerantly}
In this subsection, we characterize CSS codes capable of realizing a targeted logical diagonal gate via transversal physical $Z$-rotations. Previously, Camps-Moreno et al.~\cite{CLMRS2026} established similar conditions under which transversal physical $Z$-rotations preserve the code space, determined their resulting logical actions, and identified the operator groups that induce non-trivial logical gates versus those that induce the logical identity. 
In contrast, the following lemma simultaneously fixes both the transversal physical operations and the target logical gate. 
\begin{lemma}\label{lem:fault_tolerant_logical_diag_gates}
    Let $(C_1, C_2)_{\tn{CSS}}$ be an $[[n,k]]$ CSS code. An $n$-qubit transversal physical $Z$-rotation $U{(p,w)}$ realizes the $k$-qubit logical diagonal gate
    \begin{equation*}
        \logicalU = \diag \left( \exp{\iota \lambda_{a} } : a \in \mathbb{F}_2^k \right)
    \end{equation*}
    if and only if for all $x \in C_2$ and $a \in \mathbb{F}_2^k$, the following condition holds:
    \begin{equation*}
        \exp{\iota \frac{\pi}{2^p} \left( w \cdot \left( x \oplus y_a \right) \right)} = \exp{\iota \lambda_{a}},
    \end{equation*}
    where $y_a \in C_1 / C_2$ is the logical-$X$ representative corresponding to the logical state $\ket{a}$.
\end{lemma}
\begin{proof}
    The result follows directly by substituting the parameterization of $U(p,w)$ from Definition~\ref{def:z_rotation_tuple} into the condition established in Theorem~\ref{thm:diagonalgates}.
\end{proof}

The preceding characterization implies that each phase exponent $\lambda_a$ must be of the form $\frac{\pi m_a}{2^p}$ for some integer $m_a$. The following theorem refines this characterization by explicitly incorporating the discrete phase constraint.

\begin{theorem}\label{thm:refined_fault_tolerant_logical_diag_gates}
    Let $(C_1, C_2)_{\tn{CSS}}$ be an $[[n,k]]$ CSS code, and let $f: \mathbb{F}_2^k \rightarrow \mathbb{Z}$ be an integer-valued function such that $f(0^k)=0\pmod{2^{\ell+1}}$ and $f(a_0)$ is odd for some $a_0 \in \mathbb{F}_2^k$. An $n$-qubit transversal physical $Z$-rotation $U(p,w)$ realizes the $k$-qubit logical diagonal gate
    \begin{equation*}
        \logicalU = \diag \left(\exp{\iota \frac{\pi}{2^{\ell}} f(a)} : a \in \mathbb{F}_2^k \right)
    \end{equation*}
    if and only if $p \geq \ell$ and for all $x \in C_2$ and $a \in \mathbb{F}_2^k$, the following modular equations hold:
    \begin{align}
        w \cdot x &= 0 \pmod{2^{p+1}},  \label{eq:char_crit_1}\\
        w \cdot y_a &= 2^{p-\ell} f(a) \pmod{2^{p+1}},  \label{eq:char_crit_2}\\
        w \cdot \bracket{x \ast y_a} &= 0 \pmod{2^p}, \label{eq:char_crit_3}
    \end{align}
    where $y_a \in C_1 / C_2$ is the logical-$X$ representative corresponding to the logical state $\ket{a}$.
\end{theorem}
\begin{proof}
    By Lemma~\ref{lem:fault_tolerant_logical_diag_gates}, the physical gate $U(p,w)$ realizes the logical gate $\logicalU$ if and only if for all $x \in C_2$ and $a \in \mathbb{F}_2^k$,
    \begin{equation}\label{eq:refined_phase_exponent}
         \exp{\iota \frac{\pi}{2^{p}} \left( w \cdot \left( x \oplus y_a \right) \right)} = \exp{\iota \frac{\pi}{2^\ell} f(a)}.
    \end{equation} 
    In particular, when this expression is evaluated at $a = a_0$, the condition that $f(a_0)$ is odd necessitates that $p \geq \ell$. Consequently, condition~\eqref{eq:refined_phase_exponent} is equivalent to
    \begin{equation}\label{eq:orig_char}
        w \cdot \bracket{x \oplus y_a} = 2^{p - \ell} f(a) \pmod{2^{p + 1}}.
    \end{equation} 
    For the remainder of the proof, we first establish necessity. The left-hand side of~\eqref{eq:orig_char} can be expanded as
    \begin{equation}\label{eq:weight_expansion}
        w \cdot x + w \cdot y_a - 2 \left( w \cdot (x \ast y_a) \right) = 2^{p - \ell} f(a) \pmod{2^{p+1}}.
    \end{equation}
    Since $f(0^k) = 0 \pmod{2^{p+1}}$, setting $a=0$ (which implies $y_a = 0$) in~\eqref{eq:weight_expansion} yields~\eqref{eq:char_crit_1}:
    \begin{equation}\label{eq:interim1}
    w \cdot x = 0 \pmod{2^{p+1}}.
    \end{equation}
    Similarly, setting $x=0$ in~\eqref{eq:weight_expansion} isolates~\eqref{eq:char_crit_2}:
    \begin{equation}\label{eq:interim2}
    w \cdot y_a = 2^{p - \ell} f(a) \pmod{2^{p+1}}.
    \end{equation}
    Substituting~\eqref{eq:interim1} and~\eqref{eq:interim2} back into~\eqref{eq:weight_expansion}, we find that $-2 \left( w \cdot (x \ast y_a) \right) = 0 \pmod{2^{p+1}}$, yielding~\eqref{eq:char_crit_3}.
    To prove sufficiency, we simply substitute the conditions~\eqref{eq:char_crit_1}--\eqref{eq:char_crit_3} back into the left-hand side of~\eqref{eq:weight_expansion}, which immediately recovers the original modular equation~\eqref{eq:orig_char}.
\end{proof}

\begin{remark}\label{rmk:inverse}
    From the above characterization, it is immediate that within a CSS code, a physical rotation $U(p,w)$ realizes the logical gate $\logicalU$ if and only if $U(p,-w)$ realizes $\logicalU^\dag$. This follows directly from the diagonal representation
    \begin{equation*}
        \logicalU^\dag = \diag \left(\exp{-\iota \frac{\pi}{2^{\ell}}f(a)} : a \in \mathbb{F}_2^k \right).
    \end{equation*}
\end{remark}
\begin{remark}\label{rmk:logical_op}
  The modular criteria in~\eqref{eq:char_crit_1} and~\eqref{eq:char_crit_3} merely guarantee that the physical rotation $U(p,w)$ acts as a logical operator. To see this, we apply the gate to an arbitrary logical basis state $\ket{a}$:
   \begin{align*}
       U(p,w) \ket{a}_L &= \frac{1}{\sqrt{|C_2|}} \sum_{x \in C_2} \exp{\iota \frac{\pi}{2^{p}}  w \cdot (x \oplus y_a)} \ket{x \oplus y_a} \\
       &= \frac{1}{\sqrt{|C_2|}} \sum_{x \in C_2} \exp{\iota \; \frac{\pi}{2^{p}}  \left( w \cdot x - 2 w \cdot (x \ast y_a) + w \cdot y_a \right)} \ket{x \oplus y_a} \\
       &= \exp{\iota \; \frac{\pi}{2^{p}} \; \left(w \cdot y_a \right)} \ket{a}_L,
   \end{align*}
   the final equality being obtained by an application of~\eqref{eq:char_crit_1} and~\eqref{eq:char_crit_3}. Thus, realizing a specific target logical gate $\logicalU$ depends entirely on satisfying the modular equation~\eqref{eq:char_crit_2} over the coset space.
\end{remark}

The following remark provides a general formula for the logical diagonal gates realizable via transversal physical $Z$-rotations within a CSS code. From this, we observe that such logical gates are limited to single-qubit $Z$-rotations and multi-controlled $Z$-rotations. Both this formula and the observation were independently derived by Camps-Moreno et al.~\cite{CLMRS2026}.
\begin{remark}\label{rmk:logicalgatestructure}
The function $f$ in Theorem~\ref{thm:refined_fault_tolerant_logical_diag_gates} cannot be chosen arbitrarily. This restriction arises because $f$ expands as
    \begin{equation*}
         2^{p-\ell} f(a) = w \cdot y_a = \sum_{i=1}^{k} {(-2)}^{i-1} \sum_{1 \le j_1 < \cdots < j_i \le k} \left( \prod_{q=1}^i a({j_q}) \right) w \cdot (y_{j_1} \ast \cdots \ast y_{j_i}) \pmod{2^{p+1}}.
    \end{equation*}
    Consequently, the logical gate $\logicalU$ decomposes into a product of single-qubit dyadic $Z$-rotations and multi-controlled dyadic $Z$-rotations, whose realization is governed entirely by the coset basis $\{y_1, \ldots, y_k\}$ and its $t$-fold Schur products (for $t \ge 2$): 
    \begin{align*}
        \logicalU &=  \left( \prod_{i=1}^{k} R_Z\left(\frac{\pi}{2^p} w \cdot y_i\right)_{e_i} \right) \cdot \left( \prod_{1 \le j_1 < j_2 \le k} C_{j_1, j_2}\left( R_Z\left( -\frac{\pi}{2^{p-1}} w \cdot (y_{j_1} \ast y_{j_2}) \right) \right) \right)   \\
        &\qquad \qquad \qquad \qquad \qquad \qquad \cdots \ 
        C_{1, \ldots, k}\left( R_Z\left( {(-1)}^{k-1} \frac{\pi}{2^{p-k+1}} w \cdot (y_{1} \ast \cdots \ast y_{k}) \right) \right).
    \end{align*}
\end{remark}

\bigskip

For certain choices of the function $f: \mathbb{F}_2^k \to \mathbb{Z}$, to satisfy the conditions~\eqref{eq:char_crit_1}--\eqref{eq:char_crit_3} in Theorem~\ref{thm:refined_fault_tolerant_logical_diag_gates}, a stronger condition than $p \ge \ell$ is needed. For example, take $k=3$ and consider the function $f(a) = 1$ if $a = 111$ and $f(a) = 0$ otherwise. The logical gate defined by this $f$ is the $2$-controlled $Z$-rotation $C^{(2)}\left(R_Z(\frac{\pi}{2^{\ell}})\right)$. In particular, for $\ell=0$, this is the $\ccz$ gate. Observe by Theorem~\ref{thm:refined_fault_tolerant_logical_diag_gates} that a transversal physical $Z$-rotation $U(p,w)$ realizes logical $\ccz$ only if $p \ge 0$. However, for $p \le 1$, transversal physical $Z$-rotations lie within the Clifford group, whereas $\ccz$ is a non-Clifford gate in the third level of the Clifford hierarchy. So, one intuitively expects that $U(p,w)$ can realize $\ccz$ only if $p \ge 2$. Indeed, Proposition~\ref{prop:addressable_multi_controlled} below shows that the inequality $p \ge \ell$ must be further strengthened for $U(p,w)$ to realize an addressable multi-controlled $Z$-rotation. 

Let $(C_1, C_2)_{\tn{CSS}}$ be an $[[n,k]]$ CSS code, and consider an addressable logical $(m-1)$-controlled-$R_Z\left(\frac{\pi}{2^\ell}\right)$ gate acting on the $k$-qubit code space, where $2 \leq m \leq k$. Parameterized by an order-$m$ target address tensor $A \in \mathbb{F}_2^{k \times \cdots \times k}$, the logical gate $C^{(m-1)}_A \left(R_Z\left(\frac{\pi}{2^{\ell}} \right) \right)$ is given by
\begin{align*}
   C^{(m-1)}_A \left(R_Z\left(\frac{\pi}{2^{\ell}} \right) \right) 
   &= \prod_{1 \leq i_1 < \cdots < i_m \leq k} C^{(m-1)}_{i_1, \ldots, i_m} \left( R_Z\left(\frac{\pi}{2^{\ell}} A\left(i_1, \ldots, i_m\right) \right) \right) \\[5pt]
   &=  \diag \left(\exp{\iota \frac{\pi}{2^{\ell}} \left(\sum_{1 \leq i_1 < \cdots < i_m \leq k} A\left(i_1, \ldots, i_m \right)  a\left(i_1\right) \cdots a\left(i_m\right) \right)} : a \in \mathbb{F}_2^k \right),
\end{align*}
where at least one upper-triangular tensor entry $A(i_1, \ldots, i_m)$ is odd. 
\begin{proposition}\label{prop:addressable_multi_controlled}
    An $n$-qubit transversal physical $Z$-rotation $U(p,w)$ realizes the logical gate $\addmczrot$ on $(C_1, C_2)_{\tn{CSS}}$ if and only if $p \geq \ell+m-1$ and for all $x \in C_2$ and $a \in \mathbb{F}_2^k$, the following modular condition holds:
    \begin{align*}
        w \cdot x &= 0 \pmod{2^{p+1}}, \\
        w \cdot y_a &= 2^{p-\ell} \; \left(\sum_{1 \leq i_1 < \cdots < i_m \leq k} A\left(i_1, \ldots, i_m \right)  a\left({i_1}\right) \cdots a\left(i_m\right) \right) \pmod{2^{p+1}}, \\[4pt]
        w \cdot \left( x \ast y_a \right) &= 0 \pmod{2^p}.
    \end{align*}
    where $y_a \in C_1 / C_2$ is the logical-$X$ representative corresponding to the logical state $\ket{a}$.
\end{proposition}
The proof of this proposition is deferred to Appendix~\ref{app:prop_addressable_multi_controlled}.

\section{Appending Framework: Constructing CSS Codes That Realize Logical Diagonal Gates Fault-Tolerantly}\label{app_framework}

In this section, we propose a systematic framework for extending a CSS code to realize an arbitrary number of target logical diagonal gates (single-qubit $Z$-rotations or multi-controlled $Z$-rotations) via transversal physical $Z$-rotations. Specifically, given an $[[n',k']]$ primary CSS code and an arbitrary but fixed sequence of logical diagonal gates $\custlogicalU{1}, \ldots, \custlogicalU{\numtargates}$ (for any $\numtargates \in \mathbb{N}$) as described in Theorem~\ref{thm:refined_fault_tolerant_logical_diag_gates}, our objective is to construct a ``derived'' CSS code that preserves the $k'$ logical qubits while realizing these logical gates transversally. Furthermore, we require this derived code to maintain a minimum distance comparable to that of the primary code. 

Consider a primary CSS code, denoted by $(C_1', C_2')_{\tn{CSS}}$, with parameters $[[n',k', \geq d']]$, where the minimum $X$- and $Z$-distances satisfy
\begin{align*}
    d_X' = \min_{x \in C_1' \setminus C_2'} w_H(x) &\geq d', \\
    d_Z' = \min_{z \in (C_2')^\perp \setminus (C_1')^\perp} w_H(z) &\geq d_{\min}((C_2')^\perp) \geq d'.
\end{align*}
Let $G_{C_2'}$ and $G_{C_1' / C_2'}$ denote the generator matrices for the code $C_2'$ and the coset space $C_1' / C_2'$, respectively. The generator matrix for $C_1'$, denoted by $G_{C_1'}$, can then be expressed as the block matrix
\begin{equation*}
    G_{C_1'} = \begin{bmatrix}
        G_{C_1' / C_2'} \\ G_{C_2'}
    \end{bmatrix}.
\end{equation*}
For each $i \in \mathbb{N}$, suppose the diagonal representation of $\custlogicalU{i}$ is given by
\begin{equation*}
    \custlogicalU{i} = \diag \left(\exp{\iota \frac{\pi}{2^{\ell_i}} f_i(a)} : a \in \mathbb{F}_2^{k'} \right),
\end{equation*}
where $\ell_i$ is a non-negative integer and $f_i : \mathbb{F}_2^{k'} \rightarrow \mathbb{Z}$ is a function as in Theorem~\ref{thm:refined_fault_tolerant_logical_diag_gates}.

The primary CSS code $(C_1', C_2')_{\tn{CSS}}$ need not inherently realize any of the target gates $\custlogicalU{i}$ transversally. To handle this, for each target logical gate $\custlogicalU{i}$, we append dedicated matrices to the generator matrices $G_{C_2'}$ and $G_{C_1'/C_2'}$ of the primary CSS code to obtain augmented generator matrices $G_{C_2}$ and $G_{C_1/C_2}$ that define our derived CSS code $(C_1, C_2)_{\tn{CSS}}$. Each appended matrix is structured to ensure that the modular equations~\eqref{eq:char_crit_1}--\eqref{eq:char_crit_3} required to realize the specific logical gate $\custlogicalU{i}$ in $(C_1, C_2)_{\tn{CSS}}$ are satisfied for some transversal physical $Z$-rotation $U(p_i,w_i)$. Moreover, any logical diagonal gate realized via transversal physical $Z$-rotations in the primary code remains transversally realizable within the derived code. 

To execute this plan, carefully chosen auxiliary matrices $G_1^{(1)} \in \mathbb{F}_2^{k' \times n_1}, \ldots, G_1^{(\numtargates)} \in \mathbb{F}_2^{k' \times n_{\numtargates}}$ are sequentially appended to $G_{C_1' / C_2'}$, where each $G_1^{(i)}$ is dedicated to realizing the logical gate $\custlogicalU{i}$. (The issue of precisely how these matrices $G_1^{(i)}$ are chosen will be dealt with later.) Thus, we have
\begin{equation*}
    G_{C_1 / C_2} = \begin{bmatrix}
            G_{C_1' / C_2'} & G_1^{(1)} & G_1^{(2)} & \cdots & G_1^{(\numtargates)}
        \end{bmatrix}.
\end{equation*}        
Appending the matrix $G_{C_2'}$ with zeros will result in a trivial $Z$-distance ($d_Z=1$) for the resulting CSS code. Instead, we augment $G_{C_2'}$ with another carefully chosen set of auxiliary matrices $G_2^{(1)}, \ldots, G_2^{(\numtargates)}$, each $G_2^{(i)}$ being a $k_2^{(i)} \times n_i$ matrix over $\mathbb{F}_2$. Note that $G_1^{(i)}$ and $G_2^{(i)}$ have the same number of columns so that one can be stacked atop the other. Thus,
\begin{equation*}
G_{C_2} = \begin{bmatrix}
            G_{C_2'} & 0 & 0 & \cdots & 0 \\
            0 & G_2^{(1)} & 0 & \cdots & 0 \\
            0 & 0 & G_2^{(2)} & \cdots & 0 \\
            \vdots & \vdots & \vdots & \ddots & \vdots \\
            0 & 0 & 0 & \cdots & G_2^{(\numtargates)}
        \end{bmatrix}
\end{equation*}
and
\begin{equation*}
 G_{C_1} \ = \ \begin{bmatrix}
            G_{C_1 / C_2} \\ G_{C_2} \end{bmatrix}
            \ = \ \begin{bmatrix}
            G_{C_1' / C_2'} & G_1^{(1)} & G_1^{(2)} & \cdots & G_1^{(\numtargates)} \\
            G_{C_2'} & 0 & 0 & \cdots & 0 \\
            0 & G_2^{(1)} & 0 & \cdots & 0 \\
            0 & 0 & G_2^{(2)} & \cdots & 0 \\
            \vdots & \vdots & \vdots & \ddots & \vdots \\
            0 & 0 & 0 & \cdots & G_2^{(\numtargates)}
        \end{bmatrix}.
\end{equation*}
The lemma below, proved in Appendix~\ref{app:lem_params}, establishes the parameters of the derived code ${(C_1, C_2)}_{\tn{CSS}}$.
\begin{lemma}\label{lem:params}
    ${(C_1, C_2)}_{\tn{CSS}}$ is an $[[ n' + \sum_{i=1}^{\numtargates} n_i, k', \min\{d_X, d_Z\} ]]$ code, where $d_X \geq d'$ and 
    \begin{equation*}
        d_Z \geq \min\{d_{\min}((C_2')^\perp), d_{\min}((C_2^{(1)})^\perp), \ldots, d_{\min}((C_2^{(\numtargates)})^\perp)\},
    \end{equation*}
    with $C_2^{(i)}$ being the code generated by the rows of $G_2^{(i)}$ for each $i \in [\numtargates]$.
\end{lemma}

We now establish the precise conditions on the matrices $G_1^{(i)}$ and $G_2^{(i)}$ required to realize $\custlogicalU{i}$ via a transversal $Z$-rotation $U(p_i,w_i)$ applied to the physical qubits of $(C_1, C_2)_{\tn{CSS}}$. This rotation is applied exclusively to the physical qubits corresponding to the block-column of $G_{C_1}$ that contains $G_1^{(i)}$ and $G_2^{(i)}$. Specifically, the vector $w_i$ takes the form
\begin{equation*}
    w_i = \left( \mathbf{0}_{n'}, \mathbf{0}_{n_1}, \ldots, \mathbf{0}_{n_{i-1}}, \widetilde{w}_i, \mathbf{0}_{n_{i+1}}, \ldots, \mathbf{0}_{n_{\numtargates}} \right).
\end{equation*}

Let $y_1', \ldots, y_{k'}'$ and $y_1^{(i)}, \ldots, y_{k'}^{(i)}$ denote the row vectors of the matrices $G_{C_1' / C_2'}$ and $G_1^{(i)}$, respectively.
The appending procedure merges these components so that the rows of $G_{C_1 / C_2}$ take the block form 
\begin{align*}
    y_j &= \left(y_j', y_j^{(1)}, \ldots, y_j^{(\numtargates)} \right) \text{ for } j \in [k'].
\end{align*}
Therefore, for any $a \in \mathbb{F}_2^{k'}$, the corresponding logical-$X$ representative is given by $y_a = \left( y_a', y_a^{(1)}, \ldots, y_a^{(\numtargates)} \right)$, where 
\begin{equation*}
    y_a' = \bigoplus_{j=1}^{k'} a(j) y_j' \quad \text{and} \quad y_a^{(i)} = \bigoplus_{j=1}^{k'} a(j) y_j^{(i)} \text{ for } i \in [\numtargates].
\end{equation*}
Moreover, because $C_2$ is constructed as a direct sum, for any $x \in C_2$, there exist $x' \in C_2'$ and $x_i \in C_2^{(i)} := \text{rowspace}(G_2^{(i)})$ such that $x=(x',x_1, \ldots, x_{\numtargates})$.

By Theorem~\ref{thm:refined_fault_tolerant_logical_diag_gates} and the preceding decomposition, the physical gate $U(p_i, w_i)$ realizes the logical gate $\custlogicalU{i}$ if and only if for all $x_i \in C_2^{(i)}$ and $a \in \mathbb{F}_2^{k'}$, the following modular constraints hold:
\begin{align}
    \widetilde{w}_i \cdot x_i &= 0 \pmod{2^{p_i+1}} \nonumber  \\
    \widetilde{w}_i \cdot y_a^{(i)} &= 2^{p_i - \ell_i} f_i(a) \pmod{2^{p_i+1}} \label{eq:app_generic_equations} \\
    \widetilde{w}_i \cdot \bracket{x_i \ast y_a^{(i)}} &= 0 \pmod{2^{p_i}} \nonumber. 
\end{align}
These conditions make it clear that realizing the logical gate $\custlogicalU{i}$ within the code ${(C_1, C_2)}_{\tn{CSS}}$ depends entirely on the design of the auxiliary matrices $G_1^{(i)}$ and $G_2^{(i)}$, as well as the applied $Z$-rotation $U(p_i, w_i)$. In subsequent sections, we will detail appending-matrix construction strategies along with the corresponding transversal physical $Z$-rotations tailored to satisfy these conditions for single-qubit and multi-controlled $Z$-rotations.


The structure of $x \in C_2$ and $y_a \in C_1/C_2$ has one other important consequence, again by virtue of Theorem~\ref{thm:refined_fault_tolerant_logical_diag_gates}: the physical $Z$-rotation $U(p,w')$ realizes a logical diagonal gate $U_L$ in the primary CSS code $(C_1', C_2')_{\tn{CSS}}$ if and only if $U(p,w)$ realizes the same gate $U_L$ in $(C_1, C_2)_{\tn{CSS}}$, when $w$ is set to be
$$
w = (w',\mathbf{0}_{n_1}, \ldots, \mathbf{0}_{n_M}).
$$
In particular, this means that if a logical diagonal gate $U_L$ is realized in the primary CSS code by some transversal physical $Z$-rotation $U_P$, then the same gate $U_L$ is realized in $(C_1, C_2)_{\tn{CSS}}$ by applying the same transversal $Z$-rotation $U_P$ to the $n'$ physical qubits that correspond to the primary code.

In summary, our appending framework provides a flexible means of extending a primary CSS code to support fault-tolerant implementations of desired logical dyadic $Z$-rotations, single-qubit as well as multi-controlled. For instance, one can realize both $T$ and $\text{CS}$ gates within the same code. As noted earlier, such a capability is particularly valuable when designing codes for quantum algorithms that require a diverse set of diagonal operations, as all gates are consolidated within a single code to maintain robust error correction. However, the cost we pay for this is the increased physical qubit overhead as the number of target logical gates grows. 

\section{Fault-tolerant Addressable Logical Single-Qubit $Z$-Rotations in CSS codes}\label{ft_sq}

For a non-negative integer $\ell$, we begin this section by characterizing CSS codes that transversally support an addressable logical $R_Z\left(\frac{\pi}{2^\ell}\right)$ gate for an arbitrary but fixed addressability configuration. The following proposition is an immediate consequence of Theorem~\ref{thm:refined_fault_tolerant_logical_diag_gates}.
\begin{proposition}\label{prop:target_addr_logical_sq}
    Let ${(C_1, C_2)}_{\tn{CSS}}$ be an ${[[n,k]]}$ CSS code. For any $A \in \mathbb{F}_2^k$, the code realizes the addressable logical gate $\addrsqzrot$
    via a transversal physical $Z$-rotation $U(p,w)$ if and only if $p \geq \ell$ and, for all $x \in C_2$ and $a \in \mathbb{F}_2^k$, the following modular conditions hold:
    \begin{align*}
        w \cdot x &= 0 \pmod{2^{p + 1}}, \qquad
        w \cdot y_a = 2^{p - \ell} w_H(a \ast A) \pmod{2^{p+1}},\\
        &\qquad \qquad \qquad w \cdot (x \ast y_a) = 0 \pmod{2^{p}}.
    \end{align*}
\end{proposition}

We now provide an alternative characterization of CSS codes capable of supporting the addressable logical gate $\addrsqzrot$, decomposing the modular condition on the coset space $C_1/C_2$ into equivalent constraints on its basis vectors. Furthermore, we use this characterization to construct the appending matrices introduced in Section~\ref{app_framework} to realize these gates. While Proposition~\ref{prop:target_addr_logical_sq} allows any $p \ge \ell$, for our purposes, it suffices to take $p=\ell$.

We begin with a technical lemma that reduces modular constraints over the entire code to equivalent conditions over its basis vectors.
\begin{lemma}\label{lem:schur_basis_vectors}    
    Let $C \subseteq \mathbb{F}_{2}^n$ be a binary code with basis $\{x_1, \ldots, x_k\}$, and let $A \in \mathbb{F}_2^k$ be a fixed vector. For any $a \in \mathbb{F}_2^k$, define the corresponding codeword $x_a = \bigoplus_{i=1}^k a(i) x_i$. For any integer vector $w \in \mathbb{Z}^n$, the modular condition
    \begin{equation}\label{eq:orig_mod_code}
        w \cdot x_a = w_H(a \ast A) \pmod{2^{\ell+1}}    
    \end{equation}    
    holds for all $a \in \mathbb{F}_2^k$ if and only if the following conditions are satisfied for all $i \in [k]$ and ordered indices $1 \leq i_1 < \cdots < i_{j} \le k$ with $j \in \{2, \ldots, \ell+1\}$:
    \begin{align} 
        w \cdot x_{i} &= A(i) \pmod{2^{\ell+1}}, \label{eq:basis_1} \\ 
        w \cdot (x_{i_1} \ast \cdots \ast x_{i_{j}}) &= 0 \pmod{2^{\ell-j+2}}. \label{eq:basis_2}  
    \end{align}
\end{lemma}
This lemma is proved in Appendix~\ref{app:lem_schur_basis_vectors}. Applying the lemma to Proposition~\ref{prop:target_addr_logical_sq} yields the following characterization.
\begin{proposition}\label{prop:target_addr_log_sq_basis}    
    Let ${(C_1, C_2)}_{\tn{CSS}}$ be an $[[n,k]]$ CSS code, where $\{y_1, \ldots, y_k\}$ is a basis of $C_1/C_2$. For any $A \in \mathbb{F}_2^k$, the code realizes the addressable logical gate $\addrsqzrot$ via a transversal physical $Z$-rotation $U(\ell,w)$ if and only if for all $x \in C_2$, $a \in \mathbb{F}_2^k$, $i \in [k]$, and ordered indices $1 \leq i_1 < \cdots < i_{j} \le k$ with $j \in \{2, \ldots, \ell+1\}$, the following modular equations hold:
    \begin{alignat*}{2}
        w \cdot x &= 0 \pmod{2^{\ell+1}},
        &\qquad w \cdot (x \ast y_a) &= 0 \pmod{2^{\ell}}, \\
        w \cdot y_i &= A(i) \pmod{2^{\ell+1}}, &\qquad 
        w \cdot (y_{i_1} \ast \cdots \ast y_{i_j}) &= 0 \pmod{2^{\ell-j+2}}.
    \end{alignat*}
\end{proposition}

\subsection{Appending Matrices for a Targeted Addressable Logical Gate $\addrsqzrot$}\label{app_targ_addr_log_sq}

Given a primary $[[n',k']]$ CSS code $(C_1', C_2')_{\tn{CSS}}$ and an arbitrary but fixed addressability vector $A \in \mathbb{F}_2^{k'}$, our objective is to extend the primary code to transversally realize the addressable logical gate $\addrsqzrot$. We achieve this via the appending framework in Section~\ref{app_framework}, by deriving appending matrices tailored specifically to $\addrsqzrot$. To construct these matrices, we leverage auxiliary CSS codes in which the physical transversal $Z$-rotation $U(\ell,\widetilde{w})$ realizes the logical transversal $\sqzrot$ (see Section~\ref{css_families}). 

To this end, let $(\widetilde{C}_1, \widetilde{C}_2)_{\tn{CSS}}$ be an $[[\widetilde{n}, \widetilde{k}, \geq \widetilde{d}]]$ CSS code in which the physical $Z$-rotation $U(\ell,\widetilde{w})$ realizes the logical transversal $\sqzrot$. We require $\widetilde{d}_Z \ge d_{\min} ( \widetilde{C}_2^\perp ) \ge \widetilde{d}$ so that all logical-$Z$ operators (including $Z$-stabilizers) have weight at least $\widetilde{d}$. Following standard notation, let $G_{\widetilde{C}_2}$ and $G_{\widetilde{C}_1 / \widetilde{C}_2}$ denote the generator matrices for the code $\widetilde{C}_2$ and the coset space $\widetilde{C}_1 / \widetilde{C}_2$, respectively. Let $\widetilde{y}_1, \ldots, \widetilde{y}_{\widetilde{k}}$ denote the rows of $G_{\widetilde{C}_1/\widetilde{C}_2}$. Applying Proposition~\ref{prop:target_addr_log_sq_basis} with $\widetilde{w}$, for all $\widetilde{x} \in \widetilde{C}_2$, $a \in \mathbb{F}_2^{\widetilde{k}}$, $i \in [\widetilde{k}]$, and ordered indices $1 \leq i_1 < \cdots < i_{j} \le \widetilde{k}$ with $j \in \{2, \ldots, \ell+1\}$, the following conditions are satisfied:
\begin{alignat}{2}
    \widetilde{w} \cdot \widetilde{x} &= 0 \pmod{2^{\ell+1}},
    &\qquad \widetilde{w} \cdot (\widetilde{x} \ast \widetilde{y}_a) &= 0 \pmod{2^{\ell}}, \nonumber \\
    \widetilde{w} \cdot \widetilde{y}_i &= 1 \pmod{2^{\ell+1}}, &\qquad 
    \widetilde{w} \cdot (\widetilde{y}_{i_1} \ast \cdots \ast \widetilde{y}_{i_j}) &= 0 \pmod{2^{\ell-j+2}},
    \label{eq:secondary_css_basis_mod}
\end{alignat}
where $\widetilde{y}_a = \bigoplus_{i=1}^{\widetilde{k}} a(i) \widetilde{y}_i$.

\subsubsection{Setting up the construction} To align with the notation established in Section~\ref{app_framework}, assume that $\numtargates=1$ and $\custlogicalU{1}=\addrsqzrot$. For notational brevity, we temporarily suppress the superscript $(1)$. Thus, $G_1$ and $G_2$ are the auxiliary matrices we need for our appending construction. 
The derived CSS code ${(C_1, C_2)}_{\tn{CSS}}$ is obtained from the generator matrices
\begin{equation*}
    G_{C_2} = \begin{bmatrix}
        G_{C_2'} & 0 \\
        0 & G_2
    \end{bmatrix}, \qquad G_{C_1 / C_2} = \begin{bmatrix}
        G_{C_1' / C_2'} & G_1
    \end{bmatrix}, \qquad G_{C_1} = \begin{bmatrix}
        G_{C_1 / C_2} \\ G_{C_2}
    \end{bmatrix}. 
\end{equation*}
Let $y_1', \ldots, y_{k'}'$ denote the rows of $G_{C_1' / C_2'}$, and let $y_1, \ldots, y_{k'}$ denote the rows of the appended matrix $G_1$. The rows of $G_{C_1 / C_2}$ are $(y_1', y_1), \ldots, (y_{k'}', y_{k'})$. For any $a \in \mathbb{F}_2^{k'}$, the corresponding coset representative is given by $(y_a', y_a)$, where $y_a' = \bigoplus_{i=1}^{k'} a(i) y_i'$ and $y_a = \bigoplus_{i=1}^{k'} a(i) y_i$. 

Recall from the general appending framework that we aim to realize the addressable gate $\addrsqzrot$ via a transversal $Z$-rotation $U(\ell,w)$, where the vector $w$ is supported solely on the appended coordinates. Let $n''$ denote the number of appended coordinates, or equivalently, the number of columns in $G_1$ (or $G_2$). By Proposition~\ref{prop:target_addr_log_sq_basis}, the physical rotation $U(\ell,w)$ with $w = (\mathbf{0}_{n'}, \widetilde{w})$ realizes $\addrsqzrot$ within the code ${(C_1, C_2)}_{\tn{CSS}}$ if the following conditions hold for all $x \in \text{rowspace}(G_2)$, $a \in \mathbb{F}_2^{k'}$, $i \in [k']$, and ordered indices $1 \leq i_1 < \cdots < i_{j} \le k'$ with $j \in \{2, \ldots, \ell+1\}$:
\begin{alignat}{2}
    \widetilde{w} \cdot x &= 0 \pmod{2^{\ell+1}}, &\qquad \widetilde{w} \cdot (x \ast y_a) &= 0 \pmod{2^\ell}, \notag \\
    \widetilde{w} \cdot y_i &= A(i) \pmod{2^{\ell+1}},  &\qquad \widetilde{w} \cdot (y_{i_1} \ast \cdots \ast y_{i_j}) &= 0 \pmod{2^{\ell-j+2}}.     
    \label{eq:app_mod}
\end{alignat}
A comparison of~\eqref{eq:app_mod} with \eqref{eq:secondary_css_basis_mod} motivates us to extract vectors from the code $(\widetilde{C}_1,\widetilde{C}_2)_{\tn{CSS}}$ to construct $G_1$ and $G_2$. We first describe how to do this for the case when $\lvert \supp(A)\rvert \le \widetilde{k}$, and subsequently show how this can be extended to the general case. 

\subsubsection{The bounded support case: $\lvert \supp(A)\rvert \le \widetilde{k}$} \label{sec:bounded_support} For each $i \in [k']$, the construction must satisfy $\widetilde{w} \cdot y_i = A(i) \pmod{2^{\ell+1}}$. This motivates the following assignment strategy: when $A(i)=0$, the row $y_i$ is set to the zero vector, and when $A(i)=1$, we populate the row with a vector $\widetilde{y}_j$. Crucially, the cross-weight condition $\widetilde{w} \cdot (y_p \ast y_q) = 0 \pmod{2^{\ell}}$ strictly prohibits reusing the same vector $\widetilde{y}_j$ for distinct indices $p \ne q$ with $A(p)=A(q)=1$. Consequently, each non-zero entry in $A$ demands a unique $\widetilde{y}_j$. The assumption $\lvert \supp(A) \rvert \le \widetilde{k}$ guarantees a sufficient supply of these distinct vectors, ensuring the assignment is realizable without repetition.

To simplify the presentation of the construction, we assume that $A(i)=1$ if and only if $i \in  [\mathsf{s}]$, where $\mathsf{s}:=\lvert \supp(A) \rvert \le \widetilde{k}$. We define the rows $y_1, \ldots, y_{k'}$ of $G_1$ as follows: we set $y_i = \widetilde{y}_i$ for target rows $i \in [\mathsf{s}]$, and set $y_i=0$ for all remaining rows $\mathsf{s}+1 \le i \le k'$.

Furthermore, setting $G_2=G_{\widetilde{C}_2}$ satisfies all required modular equations to realize the logical gate $\addrsqzrot$ by applying physical $U(\ell,w)$ with $w = (\mathbf{0}_{n'}, \widetilde{w})$. Since $d_{\min}(\widetilde{C}_2^\perp) \ge \widetilde{d}$, by Lemma~\ref{lem:params}, the distance of the final code $(C_1, C_2)_{\tn{CSS}}$ is at least $\min\{d_X, d_Z\}$, where $d_X \ge d'$ and $d_Z \ge \min \{d_Z', \widetilde{d}\}$.
\subsubsection{Removing the restriction on $|\supp(A)|$} \label{sec:unrestricted_support} 
A limitation of the construction described in Section~\ref{sec:bounded_support} is that it does not directly scale: if $\lvert\supp(A)\rvert > \widetilde{k}$, we exhaust the available supply of coset vectors $\widetilde{y}_j$, and reusing them violates the required modular constraints (specifically, $\widetilde{w} \cdot (y_p \ast y_q) = 0 \pmod{2^\ell}$). To overcome this, we decompose the logical gate $\addrsqzrot$ into $g := \lceil \lvert\supp(A)\rvert / \widetilde{k} \rceil$ addressable gates $\sqzrot_{A_1}, \ldots, \sqzrot_{A_g}$ such that each of the first $g-1$ address vectors $A_i$ ($i = 1,2,\ldots,g-1$) has support size equal to $\widetilde{k}$, and $|\supp(A_g)|\; \le \widetilde{k}$.
To be precise, again setting $\mathsf{s}=\lvert\supp(A)\rvert$, and letting $j_1,j_2,\ldots,j_{\mathsf{s}}$ denote the indices in $\supp(A)$ listed in increasing order, we define the address vectors for $1 \le i \le g-1$ by $A_i = \left( \mathbf{0}_{j_{(i-1)\widetilde{k}}}, A\left(j_{(i-1)\widetilde{k}}+1:j_{i\widetilde{k}}\right), \mathbf{0}_{k'-j_{i\widetilde{k}}} \right)$, and the final vector by $A_g = \left(\mathbf{0}_{j_{(g-1)\widetilde{k}}}, A\left(j_{(g-1)\widetilde{k}}+1:k'\right) \right)$, where $j_0 := 0$ and $A(\ell_1:\ell_2) := \bigl(A(\ell_1),A(\ell_1+1),\ldots,A(\ell_2)\bigr)$. In other words, $A_1$ is determined by the prefix of $A$ up to and including the $\widetilde{k}$-th $1$, $A_2$ by the next sub-block of $A$ up to the $2\widetilde{k}$-th $1$, and so on up to $A_{g-1}$, with $A_g$ capturing the remaining suffix. Altogether, we have $\addrsqzrot = \prod_{i=1}^g \sqzrot_{A_i}$. 

We now instantiate the general appending framework from Section~\ref{app_framework} with $\numtargates=g$, setting $\custlogicalU{i}=\sqzrot_{A_i}$ for each $i \in [g]$. Because $\lvert \supp(A_i) \rvert \leq \widetilde{k}$, we can leverage the strategy in Section~\ref{sec:bounded_support} to construct the corresponding appending matrices. Thus, for each $i \in [g]$, a transversal $Z$-rotation $U(\ell,w_i)$ realizes the logical gate $\sqzrot_{A_i}$, where $w_i$ is supported strictly on the appended coordinates corresponding to $\sqzrot_{A_i}$:
\begin{equation*}
    w_i = (\mathbf{0}_{n' + (i-1)\widetilde{n}}, \widetilde{w}, \mathbf{0}_{(g-i)\widetilde{n}}).
\end{equation*}
Consequently, the product of these physical rotations, $\prod_{i=1}^g U(\ell,w_i) = U\left(\ell, \sum_{i=1}^g w_i\right)$, realizes the target logical gate $\prod_{i=1}^g \sqzrot_{A_i} = \addrsqzrot$. Since $\sum_{i=1}^g w_i = (\mathbf{0}_{n'}, \widetilde{w}, \ldots, \widetilde{w})$, we conclude that the target logical gate $\addrsqzrot$ is realized within the CSS code ${(C_1,C_2)}_{\tn{CSS}}$ by applying the transversal $Z$-rotation $U(\ell,w)$, with $w = (\mathbf{0}_{n'}, \widetilde{w}, \ldots, \widetilde{w})$, to the physical qubits of the code. Finally, Lemma~\ref{lem:params} guarantees the resulting code ${(C_1,C_2)}_{\tn{CSS}}$ has distance at least $\min\{d_X, d_Z\}$, with $d_X \ge d'$ and $d_Z \ge \min\{d_Z', \widetilde{d}\}$.


\subsubsection{Summary}\label{sec:summary}
We now summarize the appending-matrix construction. Fix a non-negative integer $\ell$ and an $[[\widetilde{n}, \widetilde{k}, \ge \widetilde{d}]]$ auxiliary code wherein the physical $Z$-rotation $U(\ell,\widetilde{w})$ realizes the logical transversal $\sqzrot$. Given a primary CSS code $(C_1', C_2')_{\tn{CSS}}$ and a target address vector $A \in \mathbb{F}_2^{k'}$, we perform the following steps to obtain the derived code:
\begin{enumerate}    
    \item Determine the number of appending matrices $\numtargates=\lceil \lvert \supp(A) \rvert/\widetilde{k} \rceil$, and define the matrices $G_1^{(j)} \in \mathbb{F}_2^{k' \times \widetilde{n}}$ and $G_2^{(j)} \in \mathbb{F}_2^{k_2 \times \widetilde{n}}$ for $j \in [\numtargates]$, where $k_2 := \dim(\widetilde{C}_2)$.
    
    \item Let $\mathsf{s}=\lvert \supp(A) \rvert$, and let $1\le i_1 < \cdots < i_{\mathsf{s}} \le k'$ denote the ordered indices in $\supp(A)$. Setting $i_0=0$, we decompose $A$ into $\numtargates$ address vectors $A_1, \ldots, A_\numtargates$, where
    \begin{align*}
        A_j &= \left( \mathbf{0}_{i_{(j-1)\widetilde{k}}}, A\left(i_{(j-1)\widetilde{k}}+1:i_{j\widetilde{k}}\right), \mathbf{0}_{k'-i_{j\widetilde{k}}} \right), \quad \text{for }j \in [\numtargates-1] \text{ and} \\
        A_{\numtargates} &=  \left(\mathbf{0}_{i_{(\numtargates-1)\widetilde{k}}}, A\left( i_{(\numtargates-1)\widetilde{k}}+1:k'\right) \right).
    \end{align*}
    
    \item For each $j \in [\numtargates]$:
    \begin{enumerate}
        \item Let $\mathsf{s}_j = \lvert \supp(A_j) \rvert$, and let $p_1 < \cdots < p_{\mathsf{s}_j}$ denote the indices in $\supp(A_j)$. We define the rows $y_1^{(j)}, \ldots, y_{k'}^{(j)}$ of $G_1^{(j)}$ for each $h \in [k']$ as
        \begin{equation*}
            y_{h}^{(j)} = \begin{cases} 
            \widetilde{y}_q, & \text{if } h=p_q \text{ for some } q \in [\mathsf{s}_j], \\ 
                     \mathbf{0}_{\widetilde{n}}, & \text{otherwise.} 
        \end{cases} 
    \end{equation*}
        \item Set $G_2^{(j)} = G_{\widetilde{C}_2}$.
    \end{enumerate}
    \item The generator matrices of the final derived CSS code $(C_1, C_2)_{\tn{CSS}}$ are given by
    \begin{equation*}
        G_{C_2} = \begin{bmatrix}
            G_{C_2'} & 0 & 0 & \cdots & 0 \\
            0 & G_2^{(1)} & 0 & \cdots & 0 \\
            0 & 0 & G_2^{(2)} & \cdots & 0 \\
            \vdots & \vdots & \vdots & \ddots & \vdots \\
            0 & 0 & 0 & \cdots & G_2^{(\numtargates)}
        \end{bmatrix} \quad \text{and} \quad G_{C_1 / C_2} = \begin{bmatrix}
            G_{C_1' / C_2'} & G_1^{(1)} & G_1^{(2)} & \cdots & G_1^{(\numtargates)}
        \end{bmatrix}.
    \end{equation*}
\end{enumerate}

The above procedure yields a derived CSS code with parameters $[[n'+\widetilde{n}M, k', \ge \min\{d', \widetilde{d}\}]]$ in which the transversal physical $Z$-rotation $U(\ell,w)$ with $w=(\mathbf{0}_{n'}, \underbrace{\widetilde{w}, \ldots, \widetilde{w}}_{M \text{ times}})$ realizes the target logical gate $\sqzrot_{A}$.

The overall assignment strategy guarantees that no coset vector $\widetilde{y}_{j}$ is repeated within any single block $G_1^{(i)}$. The rationale behind partitioning $\addrsqzrot$ thus becomes apparent: if $y_p$ and $y_q$ are any two distinct non-zero rows of $\left[ G_1^{(1)} \; \cdots \;  G_1^{(\numtargates)} \right]$, then either $\widetilde{w} \cdot (y_p \ast y_q) = \widetilde{w} \cdot (\widetilde{y}_{j_1} \ast \widetilde{y}_{j_2})$ for some $j_1 \ne j_2$  (this happens if the two rows pass through the same block $G_1^{(i)}$) or $y_p \ast y_q = 0$. In either case, $\widetilde{w} \cdot (y_p \ast y_q) = 0 \pmod{2^\ell}$.


We next illustrate the mechanics of the appending construction with a concrete example.

\begin{example}\label{example:sq}
Let the primary CSS component of our construction be the order-$4$ quantum Hamming code~\cite{S1996}, denoted ${(C_1', C_2')}_{\tn{CSS}}$, which achieves parameters $[[n'=15, k'=7, d'= 3]]$. Let the vectors ${y}_1', \ldots, {y}_{7}'$ denote the rows of the coset generator matrix $G_{C_1' / C_2'}$. Given a target addressability vector $A \in \mathbb{F}_2^{7}$ defined by non-zero coordinates $A(1) = A(3) = A(6) = 1$, our goal is to construct a derived CSS code encoding $k'=7$ logical qubits that realizes the addressable logical gate $T_A$, where $T=R_Z\left(\frac\pi4\right)$.

Our second component is an auxiliary CSS code, $(\widetilde{C}_1, \widetilde{C}_2)_{\mathrm{CSS}}$ wherein the physical transversal $T^{\dagger}$ (i.e., $U(2,-\mathbf{1})$) realizes the logical transversal $T$. We fulfill this using a ${[[\widetilde{n}=14, \widetilde{k}=2, \widetilde{d} \geq 2]]}$ code derived by puncturing two coordinates of the Reed--Muller code $\mathrm{RM}(1,4)$. While the existence of this code is well known, we will detail its construction in the Section~\ref{css_families}. Let the vectors $\widetilde{y}_1, \widetilde{y}_2$ denote the rows of the coset space generator matrix $G_{\widetilde{C}_1 / \widetilde{C}_2}$. 

We now commence the construction of the appending matrices. Since $\lvert \supp(A) \rvert \geq \widetilde{k}$, we decompose $T_A$ into $g = \lceil \lvert \supp(A) \rvert / \widetilde{k} \rceil = 2$ gates $T_{A_1}$ and $T_{A_2}$, where
\begin{equation*}
    A_1 = (1,0,1,0,0,0,0) \quad \text{ and } \quad A_2 = (0,0,0,0,0,1,0).
\end{equation*}

Setting $\custlogicalU{1} = T_{A_1}$ and $\custlogicalU{2} = T_{A_2}$, the appending matrices $G_1^{(1)}$ and $G_1^{(2)}$ in $\mathbb{F}_2^{7 \times 14}$ are constructed as follows:
\begin{equation*}
    G_1^{(1)} = \begin{bmatrix}
        \widetilde{y}_1 \\ 0 \\ \widetilde{y}_2 \\ 0 \\ 0 \\ 0 \\ 0  
    \end{bmatrix} \quad \text{and} \quad G_1^{(2)} = \begin{bmatrix}
        0 \\ 0 \\ 0 \\ 0 \\ 0 \\ \widetilde{y}_1 \\ 0  
    \end{bmatrix}, 
\end{equation*}
while the accompanying matrices $G_2^{(1)}$ and $G_2^{(2)}$ are given by
\begin{equation*}
    G_2^{(1)} = G_2^{(2)} = G_{\widetilde{C}_2}.
\end{equation*}

Let ${(C_1, C_2)}_{\tn{CSS}}$ denote the resultant CSS code obtained by integrating these appending matrices into the primary code. The generator matrix of the subcode $C_2$ is given by
\begin{equation*}
    G_{C_2} = \begin{bmatrix}
        G_{C_2'} & 0 & 0 \\
        0 & G_2^{(1)} & 0 \\
        0 & 0 & G_2^{(2)}
    \end{bmatrix} = \begin{bmatrix}
        G_{C_2'} & 0 & 0 \\
        0 & G_{\widetilde{C}_2} & 0 \\
        0 & 0 & G_{\widetilde{C}_2}
    \end{bmatrix},
\end{equation*}
and the generator matrix for the coset space $C_1 / C_2$ is given by
\begin{equation*}
    G_{C_1/ C_2} = \begin{bmatrix}
        G_{C_1' / C_2'} & G_1^{(1)} & G_1^{(2)}
    \end{bmatrix} =   \left[\begin{array}{c|c|c} 
        {y}_1' & \widetilde{y}_1 & 0 \\ 
        {y}_2' & 0 & 0  \\ 
        {y}_3' & \widetilde{y}_2 & 0  \\ 
        {y}_4' & 0 & 0  \\
        {y}_5' & 0 & 0  \\ 
        {y}_6' & 0 & \widetilde{y}_1 \\
        y_7' & 0 & 0
    \end{array}\right].
\end{equation*}

This completes our appending construction. By invoking Lemma~\ref{lem:params}, we deduce that ${(C_1, C_2)}_{\tn{CSS}}$ is a ${[[43, 7, \geq 2]]}$ code. Applying the transversal $Z$-rotation $U(2, w_1)$ with $w_1=(\mathbf{0}_{15}, -\mathbf{1}_{14},\mathbf{0}_{14})$ realizes logical $T_{A_1}$, while $U(2, w_2)$ with $w_2=(\mathbf{0}_{15}, \mathbf{0}_{14}, -\mathbf{1}_{14})$ realizes logical $T_{A_2}$. Consequently, the product of these physical rotations, $U(2,w_1) U(2,w_2) = U(2,w)$ with combined vector $w=(\mathbf{0}_{15}, -\mathbf{1}_{14}, -\mathbf{1}_{14})$, successfully realizes the target logical gate $T_{A_1} T_{A_2} = T_{A}$.
\end{example}    


We end this section by pointing out that the methodology herein enables us to construct CSS codes that admit transversal realizations of all possible addressable logical $\sqzrot$ gates, for some fixed $\ell \ge 0$. This is based on the observation that to do this, it suffices to construct CSS codes that admit transversal realizations of $\sqzrot$ on each logical qubit independently. Indeed, suppose that for each $i$, a transversal $Z$-rotation $U(\ell,w_i)$ realizes the logical $\sqzrot$ gate on the $i$-th logical qubit. Then, for any target address vector $A$, the composite gate $\prod_{i} U(\ell,w_i)^{A(i)}$ realizes the addressable logical gate $\addrsqzrot$.  

The construction of such CSS codes is straightforward: we simply invoke the appending framework of Section~\ref{app_framework} with $\numtargates=k'$ and set $\custlogicalU{i}=\sqzrot_{e_i}$. The strategy described in Section~\ref{sec:bounded_support} tells us how to choose appending matrices for single-qubit $Z$-rotations with address vectors supported on a single coordinate.

\subsection{Relaxed Modular Equations to Realize Addressable Logical Single-Qubit $Z$-Rotations}\label{sec:sq_relaxed}

The system of modular equations in Proposition~\ref{prop:target_addr_log_sq_basis} can be further relaxed. Specifically, weakening the modulo-$2^{\ell+1}$ constraint $w \cdot y_i = A(i) \pmod{2^{\ell+1}}$ to make it a modulo-$2^\ell$ constraint still allows the code to transversally realize the logical gate $\addrsqzrot$, as we now demonstrate. 

Using Lemma~\ref{lem:relaxed} in Appendix~\ref{sec:relaxed}, we observe that any vector $w$ satisfying the relaxed conditions can be systematically refined into a vector $w'$ that satisfies the original modular equations, thereby realizing the target gate $\addrsqzrot$ via $U(\ell,w')$. Applying Lemma~\ref{lem:relaxed} to Proposition~\ref{prop:target_addr_log_sq_basis} yields the following alternative characterization:
\begin{proposition}\label{prop:target_addr_log_sq_basis_relaxed}        
    Let ${(C_1, C_2)}_{\tn{CSS}}$ be an $[[n,k]]$ CSS code, where $\{x_1, \ldots, x_{k_2}\}$ and $\{y_1, \ldots, y_k\}$ are bases for $C_2$ and $C_1/C_2$, respectively. For any $A \in \mathbb{F}_2^k$, the code realizes the addressable logical gate $\addrsqzrot$ via a transversal physical $Z$-rotation $U(\ell,w')$ for some $w' \in \mathbb{Z}^n$ if and only if there exists a vector $w \in \mathbb{Z}^n$ such that for all $a \in \mathbb{F}_2^k$, $i \in [k_2]$, ordered indices $1 \leq i_1 < \cdots < i_{j} \le k_2$ with $j \in \{2, \ldots, \ell+1\}$, $p \in [k]$, and ordered indices $1 \leq p_1 < \cdots < p_q \le k$ with $q \in \{2, \ldots, k\}$, the following conditions hold:
    \begin{alignat*}{3}
        w \cdot x_i &= 0 \pmod{2^{\ell}}, &\qquad w \cdot (x_{i_1} \ast \cdots \ast x_{i_j}) &= 0\pmod{2^{\ell-j+2}},
        &\qquad w \cdot (x \ast y_a) &= 0 \pmod{2^{\ell}}, \\
        w \cdot y_p &= A(p) \pmod{2^{\ell}}, &\qquad 
        w \cdot (y_{p_1} \ast \cdots \ast y_{p_q}) &= 0 \pmod{2^{\ell-q+2}}.
    \end{alignat*}
\end{proposition}
The vector $w$ in Proposition~\ref{prop:target_addr_log_sq_basis_relaxed} need not satisfy the original characterization in Proposition~\ref{prop:target_addr_log_sq_basis} required to realize the addressable logical gate $\addrsqzrot$ directly. Instead, by Lemma~\ref{lem:relaxed}, we can map the physical rotation $U(\ell, w)$ to a refined rotation $U(\ell,w')$, where the modified vector $w'=w-2^{\ell}w_1 - 2^{\ell} w_2$ exactly satisfies the original characterization in Proposition~\ref{prop:target_addr_log_sq_basis}. Here, following the proof of Lemma~\ref{lem:relaxed}, the correction vector $w_1$ is obtained when lifting the conditions on $w \cdot x_i$ from $0 \pmod{2^\ell}$ to $0 \pmod{2^{\ell+1}}$, and $w_2$ is obtained when lifting the conditions on $w \cdot y_p$ from $A(p) \pmod{2^\ell}$ to $A(p) \pmod{2^{\ell+1}}$.

Because $U(\ell,w')$ differs from $U(\ell,w)$ only by the application of physical Pauli $Z$ gates, the relaxed conditions can be interpreted as characterizing all CSS codes capable of realizing the addressable logical gate $\addrsqzrot$ via transversal physical $Z$-rotations, up to the application of Pauli $Z$ corrections. In particular, the process of lifting the modular conditions using Lemma~\ref{lem:relaxed} establishes the sufficiency of Proposition~\ref{prop:target_addr_log_sq_basis_relaxed}. Conversely, the necessity follows immediately, as any vector satisfying the strict characterization of Proposition~\ref{prop:target_addr_log_sq_basis} trivially satisfies the relaxed modular conditions.

\section{Fault-tolerant Addressable Logical $\czrot$ Gates in CSS Codes}\label{ft_cs}

In this section, we utilize Proposition~\ref{prop:addressable_multi_controlled} to derive an alternative characterization of CSS codes that transversally realize an addressable logical $\czrot$ gate, which motivates our appending-matrix construction for $1$-controlled $Z$-rotations. Although Proposition~\ref{prop:addressable_multi_controlled} allows any $p \ge \ell + 1$, we restrict our attention to $p = \ell + 1$. 

While the following characterization and appending-matrix construction extend to general multi-controlled $Z$-rotations, the associated technical details can obscure the core ideas. Therefore, we focus here on $1$-controlled $Z$-rotations and defer a self-contained treatment of the general case to Appendix~\ref{ft_mc_zrot}.   

Lemma~\ref{lem:schur_basis_vectors_for_mc} in Appendix~\ref{ft_mc_zrot} extends Lemma~\ref{lem:schur_basis_vectors} from binary vectors $A \in \mathbb{F}_2^k$ to binary tensors $A \in \mathbb{F}_2^{k \times \cdots \times k}$. Applying it to Proposition~\ref{prop:addressable_multi_controlled} for $1$-controlled $Z$-rotations directly yields the following result.
\begin{proposition}\label{prop:target_addr_log_cs_basis}    
    Let ${(C_1, C_2)}_{\tn{CSS}}$ be an $[[n,k]]$ CSS code, where $\{y_1, \ldots, y_k\}$ is a basis of $C_1/C_2$. For any matrix $A \in \mathbb{F}_2^{k \times k}$, the code realizes the addressable logical gate $\addczrot$ via a transversal physical $Z$-rotation $U(\ell+1,w)$ if and only if for all $x \in C_2$, $a \in \mathbb{F}_2^k$, and ordered indices $1 \le i_1 < \cdots < i_j \le k$ with $j \in [\ell+2]$, the following conditions hold:
    \begin{align*}
        &w \cdot x = 0 \pmod{2^{\ell+2}},  \qquad
        w \cdot {(x \ast y_a)} = 0 \pmod{2^{\ell+1}}, \\
        &w \cdot \left( y_{i_1} \ast \cdots \ast y_{i_{j}}\right) = \begin{cases} 
            -A(i_1, i_2) \pmod{2^{\ell+1}}, & \text{if } j=2, \\ 
            0 \pmod{2^{\ell-j+3}}, & \text{otherwise.} 
        \end{cases}
    \end{align*}
\end{proposition}

\subsection{Appending Matrices for a Targeted Addressable Logical gate $\addczrot$}\label{app_targ_addr_log_cs}

Given a primary $[[n',k']]$ CSS code $(C_1', C_2')_{\tn{CSS}}$ and a target address matrix $A \in \mathbb{F}_2^{k' \times k'}$, we utilize the appending framework from Section~\ref{app_framework} to construct an extended code that transversally realizes the addressable logical gate $\addczrot$. We derive the required appending matrices using the auxiliary $[[\widetilde{n}, \widetilde{k}, \widetilde{d}]]$ CSS code $(\widetilde{C}_1, \widetilde{C}_2)_{\tn{CSS}}$ from Section~\ref{app_targ_addr_log_sq}, wherein the physical transversal $Z$-rotation $U(\ell+1, \widetilde{w})$ realizes the logical transversal $R_Z\left( \frac{\pi}{2^{\ell+1}} \right)$, and $\widetilde{d}_Z \ge d_{\min}(\widetilde{C}_2^\perp) \ge \widetilde{d}$.  

Recall from Section~\ref{app_targ_addr_log_sq} that if $\widetilde{y}_1, \ldots, \widetilde{y}_{\widetilde{k}}$ are the rows of the coset generator matrix $G_{\widetilde{C}_1/\widetilde{C}_2}$, the following hold for all $\widetilde{x} \in \widetilde{C}_2$, $a \in \mathbb{F}_2^{\widetilde{k}}$, $i \in [\widetilde{k}]$, and ordered indices $1 \leq i_1 < \cdots < i_{j} \le \widetilde{k}$ with $j \in \{2, \ldots, \ell+2\}$:
\begin{alignat}{2}
    \widetilde{w} \cdot \widetilde{x} &= 0 \pmod{2^{\ell+2}},
    &\qquad \widetilde{w} \cdot (\widetilde{x} \ast \widetilde{y}_a) &= 0 \pmod{2^{\ell+1}}, \nonumber \\
    \widetilde{w} \cdot \widetilde{y}_i &= 1 \pmod{2^{\ell+2}}, &\qquad 
    \widetilde{w} \cdot (\widetilde{y}_{i_1} \ast \cdots \ast \widetilde{y}_{i_j}) &= 0 \pmod{2^{\ell-j+3}}.
    \label{eq:secondary_css_basis_mod_cs}
\end{alignat}

Our construction requires the following properties of the rows of $G_{\widetilde{C}_1 / \widetilde{C}_2}$ established in Proposition~\ref{prop:schur_sum_basis_mc_secc_ode}: 
\begin{enumerate}
    \item \begin{equation}\label{eq:sum} 
        \widetilde{w} \cdot \left( \bigoplus_{i=1}^{\widetilde{k}} \widetilde{y}_i\right) = \widetilde{k} \pmod{2^{\ell+2}}.
    \end{equation}
    \item For any index $p \in [\widetilde{k}]$,
    \begin{equation}
        \widetilde{w} \cdot \left( \left(\bigoplus_{i=1}^{\widetilde{k}} \widetilde{y}_i \right) \ast \widetilde{y}_p \right) = 1 \pmod{2^{\ell+1}}.\label{eq:schur_prop_sc_code_cs_1}
    \end{equation}
    \item For any integer $q \ge 2$ and ordered sequence of indices $1 \le p_1 < \cdots < p_q \le \widetilde{k}$,
    \begin{equation}    
        \widetilde{w} \cdot \left( \left(\bigoplus_{i=1}^{\widetilde{k}} \widetilde{y}_i \right) \ast \left( \widetilde{y}_{p_1} \ast \cdots \ast \widetilde{y}_{p_q} \right) \right) = 0 \pmod{2^{\ell-q+2}}.\label{eq:schur_prop_sc_code_cs_2}
    \end{equation}
\end{enumerate}

\subsubsection{Reduction to a simpler special case}
We first decompose the target gate $\addczrot$ by sequentially ordering its control qubits. Let $c_1 < \cdots < c_\kappa$ denote the ordered indices in $[k'-1]$ corresponding to active control qubits---that is, indices $i$ for which there exists at least one target qubit $t \in \{i+1, \ldots, k'\}$ such that $A(i, t)=1$. We factor $\addczrot$ into $\kappa$ addressed gates $C_{A_1}\left( R_Z\left( \frac{\pi}{2^\ell} \right) \right), \ldots, C_{A_\kappa}\left( R_Z\left( \frac{\pi}{2^\ell} \right) \right)$. Each matrix $A_j$ isolates the interactions controlled by $c_j$ by inheriting the $c_j$-th row slice of $A$ and vanishing elsewhere:
\begin{equation*}
    A_j(p,q):= \begin{cases}
        A(c_j,q), & \text{if }  p=c_j \text{ and } q \in \{c_j+1, \ldots, k'\}, \\
        0, & \text{otherwise.}
    \end{cases}
\end{equation*}

If the appending-matrix construction is established for each factor $C_{A_j}\left( R_Z\left( \frac{\pi}{2^\ell} \right) \right)$, the appending framework guarantees that each gate $C_{A_j}\left( R_Z\left( \frac{\pi}{2^\ell} \right) \right)$ is realized by some transversal $Z$-rotation $U(\ell+1,w_j)$ in the final extended code. The composition $\prod_{j=1}^\kappa U(\ell+1,w_j)$ then realizes the full target gate $\addczrot$. Consequently, without loss of generality, we assume that the first row is the only non-zero row of $A$---meaning the logical $C\left( R_Z\left( \frac{\pi}{2^\ell} \right) \right)$ gates act exclusively with the first logical qubit as the control qubit and a subset of the remaining qubits as targets.

Furthermore, we restrict the construction to target matrices satisfying $\lvert \supp(A(1,2:k')) \rvert \le \widetilde{k}$, where $A(p, q:r):=(A(p,q), A(p,q+1), \ldots, A(p,r))$ for $p \in [k']$ and $1 \le q \le r \le k'$. If the target support exceeds $\widetilde{k}$, the flexibility of the appending framework allows us to partition the addressable gate $\addczrot$ into $g=\lceil \lvert \supp(A(1,2:k')) \rvert / \widetilde{k} \rceil$ factors, each acting from control qubit $1$ onto at most $\widetilde{k}$ targets in $\{2, \ldots, k'\}$. Constructing appending matrices independently for each factor and composing their corresponding physical $Z$-rotations realizes the target gate $\addczrot$ within the final extended code. 

\subsubsection{The appending-matrix construction in the special case of $|\supp(A(1,2:k'))| \ \le \widetilde{k}$}

To simplify the presentation of the construction, we further assume that $A(i,j)=1$ if and only if $i=1$ and $j \in \{2, 3, \ldots, \mathsf{s}+1 \}$, where $\mathsf{s}:=\lvert \supp(A(1,2:k')) \rvert \le \widetilde{k}$. This effectively designates $\{2, \ldots, \mathsf{s}+1\}$ as the set of all target qubits. Again, recall that we are given a primary $[[n',k']]$ CSS code $(C_1', C_2')_{\tn{CSS}}$, and that we will obtain appending matrices $G_1$ and $G_2$ using an auxiliary $[[\widetilde{n}, \widetilde{k}, \widetilde{d}]]$ CSS code $(\widetilde{C}_1, \widetilde{C}_2)_{\tn{CSS}}$ for which the modular equations in~\eqref{eq:secondary_css_basis_mod_cs} hold. We define the rows $y_1, \ldots, y_{k'}$ of $G_1$ as follows: we set $y_1 = \bigoplus_{p=1}^{\widetilde{k}} \widetilde{y}_p$ for the control qubit row, assign $y_j = \widetilde{y}_{j-1}$ for target qubits $2 \le j \le \mathsf{s}+1$, and set $y_j=0$ for all remaining rows $\mathsf{s}+1 < j \le k'$.

Next, we set $G_2 = G_{\widetilde{C}_2}$, and let ${(\widehat{C}_1, \widehat{C}_2)}_{\tn{CSS}}$ denote the CSS code obtained by applying the appending-matrix construction to $(C_1', C_2')_{\tn{CSS}}$. 

\subsubsection{$U(\ell+1,w)$ is a logical operator for $(\widehat{C}_1, \widehat{C}_2)_{\tn{CSS}}$}

Consider the transversal $Z$-rotation $U(\ell+1, w)$ with $w=(\mathbf{0}_{n'}, -\widetilde{w})$. By construction, $G_2 = G_{\widetilde{C}_2}$ and the rows of $G_1$ are drawn from $\widetilde{C}_1/\widetilde{C}_2$. Consequently, for any $x=(x',\widetilde{x}) \in \widehat{C}_2$ and $y_a=(y', \widetilde{y}) \in \widehat{C}_1/ \widehat{C}_2$ (with $a \in \mathbb{F}_2^{k'}$), we have $\widetilde{x} \in \widetilde{C}_2$ and $\widetilde{y} \in \widetilde{C}_1/ \widetilde{C}_2$. By~\eqref{eq:secondary_css_basis_mod_cs}, this guarantees
\begin{equation*}
    w \cdot x = 0 \pmod{2^{\ell+2}} \quad \text{ and } \quad w \cdot (x \ast y_a) = 0 \pmod{2^{\ell + 1}}.
\end{equation*}

Thus, Remark~\ref{rmk:logical_op} establishes $U(\ell+1,w)$ as a logical operator on $(\widehat{C}_1, \widehat{C}_2)_{\tn{CSS}}$. 

\subsubsection{$U(\ell+1,w)$ realizes $\addczrot$ up to a logical phase} For any logical state $\ket{a}$ with $a \in \mathbb{F}_2^{k'}$, the action of $U(\ell+1,w)$ is given by
\begin{align*}
    U(\ell+1,w) &\ket{a}_L = \exp{\iota \frac{\pi}{2^{\ell+1}} \left(w \cdot \left( \bigoplus_{i=1}^{k'} a(i) y_i\right) \right)} \ket{a}_L \\
    &= \exp{\iota \frac{\pi}{2^{\ell+1}} \left( \sum_{i=1}^{k'} {(-2)}^{i-1} \left( \sum_{1 \leq j_1 < \cdots < j_i \leq k'} \left( \prod_{p=1}^i a({j_p}) \right) w \cdot \left( y_{j_1} \ast \cdots \ast y_{j_i}\right) \right) \right)} \ket{a}_L.
\end{align*}
Separating the linear and quadratic terms, we can express this action as
\begin{align*}
    U(\ell+1,w)\ket{a}_L &= \exp\Biggl\{ \iota \frac{\pi}{2^{\ell+1}} \Biggl( \sum_{i=1}^{k'} a(i) w \cdot y_i -2 \sum_{1 \le i < j \le k'} a(i)a(j) w \cdot (y_i \ast y_j) \\ 
    &\qquad\qquad+
    \sum_{i=3}^{k'}
    {(-2)}^{i-1} \left( \sum_{1 \leq j_1 < \cdots < j_i \leq k'} \left( \prod_{p=1}^i a({j_p}) \right) w \cdot \left( y_{j_1} \ast \cdots \ast y_{j_i}\right) \right)\Biggr) \Biggr\} \ket{a}_L.
\end{align*}

We next show that the $t$-fold Schur products for $t > 2$ vanish in the phase exponent, while the $2$-fold products exactly realize the target gate $\addczrot$. The $1$-fold terms (the linear terms), however, introduce an extraneous logical phase, which we will eliminate in the next subsection. 

To see why the $t$-fold Schur products for $t > 2$ vanish, first consider any ordered sequence of $t>2$ target qubit indices $2 \le p_1 < \cdots < p_t \le \mathsf{s}+1$. Condition~\eqref{eq:secondary_css_basis_mod_cs} implies
\begin{equation*}
    w \cdot (y_{p_1} \ast \cdots \ast y_{p_t}) = -\widetilde{w} \cdot (\widetilde{y}_{p_1-1} \ast \cdots \ast \widetilde{y}_{p_t-1}) =  0 \pmod{2^{\ell-t+3}}.
\end{equation*}
On the other hand, if the Schur product combines the control qubit index with $t \ge 2$ target qubit indices,~\eqref{eq:schur_prop_sc_code_cs_2} yields
\begin{equation*}
    w \cdot (y_{1} \ast y_{p_1} \ast \cdots \ast y_{p_t}) = -\widetilde{w} \cdot \left( \left( \bigoplus_{i=1}^{\widetilde{k}} \widetilde{y}_i \right) \ast \widetilde{y}_{p_1-1} \ast \cdots \ast \widetilde{y}_{p_t-1} \right) = 0 \pmod{2^{\ell-t+2}}.
\end{equation*}
Because only the first $\mathsf{s}+1$ rows of $G_1$ are non-zero, the preceding modular equations guarantee that all $t$-fold Schur products for $t > 2$ vanish in the phase exponent. Consequently, the action of the physical rotation $U(\ell+1,w)$ reduces to the linear and quadratic terms:
\begin{equation*}
    U(\ell+1,w)\ket{a}_L =\exp\Biggl\{ \iota \frac{\pi}{2^{\ell+1}} \Biggl( \sum_{i=1}^{k'} a(i) w \cdot y_i -2 \sum_{1 \le i < j \le k'} a(i)a(j) w \cdot (y_i \ast y_j) \Biggr) \Biggr\} \ket{a}_L.
\end{equation*}

We now evaluate the $2$-fold Schur products of non-zero rows of $G_1$ in the phase exponent. For any target qubit indices $2 \le p < q \le \mathsf{s}+1$,~\eqref{eq:secondary_css_basis_mod_cs} gives
\begin{equation*}
    w \cdot (y_{p} \ast y_{q}) = -\widetilde{w} \cdot ( \widetilde{y}_{p-1} \ast \widetilde{y}_{q-1}) = 0 \pmod{2^{\ell+1}}.
\end{equation*}
If the product instead involves the control qubit and any target qubit $2 \le p \le \mathsf{s} + 1$,~\eqref{eq:schur_prop_sc_code_cs_1} yields
\begin{equation*}
    w \cdot (y_1 \ast y_{p}) = -\widetilde{w} \cdot \left( \left( \bigoplus_{i=1}^{\widetilde{k}} \widetilde{y}_i \right) \ast \widetilde{y}_{p-1} \right) = -1 \pmod{2^{\ell+1}}.
\end{equation*}
Substituting these evaluations into the quadratic terms, and recalling that $A(i,j)=1$ if and only if for $i=1$ and $j \in \{2, \ldots, \mathsf{s}+1\}$, we see that the logical action of the transversal physical rotation $U(\ell+1,w)$ further reduces to
\begin{align*}
    U(\ell+1,w)\ket{a}_L &= \exp{\iota \frac{\pi}{2^{\ell+1}} \left( \sum_{i=1}^{k'} a(i) w \cdot y_i + 2\sum_{1 \le i < j  \le k'} a(i) a(j) A(i,j) \right)}  \ket{a}_L \\
    &= \exp{\iota \frac{\pi}{2^{\ell+1}} \left( \sum_{i=1}^{k'} a(i) w \cdot y_i \right)} \addczrot \ket{a}_L.
\end{align*}
Thus, we have successfully realized the target gate $\addczrot$ up to a logical phase. 

\subsubsection{Eliminating the residual logical phase}
By~\eqref{eq:secondary_css_basis_mod_cs} and~\eqref{eq:sum}, we have $-\widetilde{w} \cdot \widetilde{y}_i=-1 \pmod{2^{\ell+2}}$ and $-\widetilde{w} \cdot \left( \bigoplus_{i=1}^{\widetilde{k}} \widetilde{y}_i \right) = -\widetilde{k} \pmod{2^{\ell+2}}$. Substituting these values allows us to rewrite the logical action as
\begin{equation*}
    U(\ell+1,w)\ket{a}_L = \exp{\iota \frac{\pi}{2^{\ell+1}} \left( -\widetilde{k} \sum_{i=1}^{k'} a(i) A_1(i) - \sum_{i=1}^{k'}a(i) A_2(i) \right)} \addczrot \ket{a}_L,
\end{equation*}
where $A_1 = (1, \mathbf{0}_{k'-1})$ and $A_2 = (0, \mathbf{1}_{\mathsf{s}}, \mathbf{0}_{k'-\mathsf{s}-1})$. We can thus explicitly identify the unwanted residual phase as the action of the logical $Z$-rotations $\left( R_Z\left(\frac{\pi}{2^{\ell+1}}\right)_{A_1}^\dag \right)^{\widetilde{k}}$ and $R_Z\left(\frac{\pi}{2^{\ell+1}}\right)_{A_2}^\dag$.

To cancel this unwanted logical phase, we append auxiliary matrices to $(\widehat{C}_1, \widehat{C}_2)_{\tn{CSS}}$ that realize the inverse logical $Z$-rotations $\left(R_Z\left(\frac{\pi}{2^{\ell+1}}\right)_{A_1}\right)^{\widetilde{k}}$ and $R_Z\left(\frac{\pi}{2^{\ell+1}}\right)_{A_2}$. This yields the final derived CSS code $(C_1, C_2)_{\tn{CSS}}$, which transversally realizes the exact target gate $\addczrot$. Note that the appending matrix required for the rotation $\left(R_Z\left(\frac{\pi}{2^{\ell+1}}\right)_{A_1}\right)^{\widetilde{k}}$ is the same as that of the rotation $R_Z\left(\frac{\pi}{2^{\ell+1}}\right)_{A_1}$. This is because if a transversal physical $Z$-rotation $\physicalU$ realizes $R_Z\left(\frac{\pi}{2^{\ell+1}}\right)_{A_1}$, then $\physicalU^{\widetilde{k}}$ realizes $\left(R_Z\left(\frac{\pi}{2^{\ell+1}}\right)_{A_1}\right)^{\widetilde{k}}$. 

We note here that although not explicitly invoked, Proposition~\ref{prop:target_addr_log_cs_basis}  motivates our construction. We define the control qubit indexed row as the sum $\bigoplus_{i=1}^{\widetilde{k}} \widetilde{y}_i$ and the target qubit indexed rows as the vectors $\widetilde{y}_j$ specifically to satisfy the modular equation $w \cdot (y_{i_1} \ast y_{i_2}) = -A(i_1, i_2) \pmod{2^{\ell+1}}$. The intermediate code $(\widehat{C}_1, \widehat{C}_2)_{\tn{CSS}}$ satisfies all requirements of the proposition except $w \cdot y_{i_1} = 0 \pmod{2^{\ell+2}}$. Appending auxiliary matrices to this intermediate code to cancel the residual logical phase precisely resolves this modular equation, ensuring the final derived code fully satisfies Proposition~\ref{prop:target_addr_log_cs_basis}.

Finally, observe that we can use the same code $(\widetilde{C}_1,\widetilde{C}_2)_{\tn{CSS}}$ to construct the appending matrices for $R_Z\left(\frac{\pi}{2^{\ell+1}}\right)_{A_1}$ and $R_Z\left(\frac{\pi}{2^{\ell+1}}\right)_{A_2}$. This ensures that the target rotations trivially satisfy the bounded support condition of Section~\ref{sec:bounded_support}. Thus, we have a total appended blocklength of exactly $3\widetilde{n}$ in $(C_1, C_2)_{\tn{CSS}}$. 

\subsubsection{Summary}
We now summarize the appending-matrix construction for $1$-controlled $Z$-rotations. Fix a non-negative integer $\ell$ and an $[[\widetilde{n}, \widetilde{k}, \ge \widetilde{d}]]$ auxiliary code wherein the physical $Z$-rotation $U(\ell+1,\widetilde{w})$ realizes the logical transversal $R_Z\left(\frac{\pi}{2^{\ell+1}}\right)$. Given a primary CSS code $(C_1', C_2')_{\tn{CSS}}$ and a target address matrix $A \in \mathbb{F}_2^{k' \times k'}$, we perform the following steps to obtain the derived code:
\begin{enumerate}  
    \item First, we determine the number of appending matrices:
    \begin{enumerate}
        \item Let $1 \le c_1 < \cdots < c_{\kappa} \le k'-1$ denote the ordered indices $c_i$ for which $\supp(A(c_i,c_i+1:k'))$ is non-empty.
        \item For each $i \in [\kappa]$, let $M_{i} = \lceil \lvert \supp(A(c_i, c_i+1:k'))\rvert / \widetilde{k} \rceil$ be the number of appending matrices required for the control qubit $c_i$.
        \item The total number of appending matrices is $\numtargates = \sum_{i=1}^\kappa M_{i}$. We define the corresponding matrices $G_1^{(i,j)} \in \mathbb{F}_2^{k' \times \widetilde{n}}$ and $G_2^{(i, j)} \in \mathbb{F}_2^{k_2 \times \widetilde{n}}$ for $i \in [\kappa]$ and $j \in [\numtargates_i]$, where $k_2 := \dim(\widetilde{C}_2)$.
    \end{enumerate}
    
    \item For each $i \in [\kappa]$:
    \begin{enumerate}
        \item Let $\mathsf{s}_i=\lvert \supp(A(c_i, c_i+1:k')) \rvert$, and let $ c_i < j_1 < \cdots < j_{\mathsf{s}_i} \le k'$ denote the ordered indices in $\supp(A(c_i, c_i+1:k'))$. Setting $j_0:=c_i$, we decompose $A(c_i, \cdot)$ into $\numtargates_i$ address vectors $A_{i,1}, \ldots, A_{i,\numtargates_i}$, where:
        \begin{align*}
            A_{i,p} &= \left( \mathbf{0}_{j_{(p-1)\widetilde{k}}}, A\left(c_i, j_{(p-1)\widetilde{k}}+1:j_{p\widetilde{k}}\right), \mathbf{0}_{k'-j_{p\widetilde{k}}} \right), \quad \text{for }p \in [\numtargates_i-1] \text{ and} \\
            A_{i, \numtargates_i} &=  \left(\mathbf{0}_{j_{(\numtargates_i-1)\widetilde{k}}}, A\left(c_i, j_{(\numtargates_i-1)\widetilde{k}}+1:k'\right) \right).
        \end{align*}
        \item For each $j \in [\numtargates_i]$:
        \begin{enumerate}
            \item Let $\mathsf{s}_{i,j} = \lvert \supp(A_{i, j}) \rvert$, and let $c_i < p_1 < \cdots < p_{\mathsf{s}_{i,j}} \le k'$ denote the indices in $\supp(A_{i,j})$. We define the rows $y_1^{(i,j)}, \ldots, y_{k'}^{(i,j)}$ of $\widehat{G}_1^{(i,j)}$ for each $h \in [k']$ as:
        \begin{equation*}
            y_{h}^{(i,j)} = \begin{cases} 
            \bigoplus_{p=1}^{\widetilde{k}} \widetilde{y}_p, & \text{if } h=c_i,\\
            \widetilde{y}_q, & \text{if } h=p_q \text{ for some } q \in [\mathsf{s}_{i,j}], \\ 
                     \mathbf{0}_{\widetilde{n}}, & \text{otherwise.} 
        \end{cases} 
    \end{equation*}
        \item Set $\widehat{G}_2^{(i,j)} = G_{\widetilde{C}_2}$.
        
        \item To cancel the residual phase induced by the addressable logical gates $\left(\sqzrot_{A_1}\right)^{\widetilde{k}}$ and $\sqzrot_{A_2}$, we append additional auxiliary matrices to $\widehat{G}_1^{(i,j)}$ and $\widehat{G}_2^{(i,j)}$, as summarized in Section~\ref{sec:summary}. The corresponding address vectors are $A_1=(\mathbf{0}_{c_i-1}, 1, \mathbf{0}_{k'-c_i})$ and $A_2$, where $A_2(q)=1$ if and only if $q=p_r$ for $r \in \mathsf{s}_{i,j}$. This step yields the final appending matrices $G_1^{(i,j)}$ and $G_2^{(i,j)}$.
        \end{enumerate}
    \end{enumerate}
    \item The generator matrices of the final derived CSS code $(C_1, C_2)_{\tn{CSS}}$ are given by
    \begin{equation*}
        G_{C_2} = \begin{bmatrix}
            G_{C_2'} & 0 & \cdots & 0 \\
            0 & G_2^{(1,1)} & \cdots & 0 \\
            \vdots & \vdots &  \ddots & \vdots \\
            0 & 0 & \cdots & G_2^{(\kappa, M_{\kappa})}
        \end{bmatrix} \quad \text{and} \quad G_{C_1 / C_2} = \begin{bmatrix}
            G_{C_1' / C_2'} & G_1^{(1,1)} & \cdots & G_1^{(\kappa, \numtargates_\kappa)}
        \end{bmatrix}.
    \end{equation*}
\end{enumerate}
By Lemma~\ref{lem:params}, this procedure yields a derived CSS code with parameters $[[n'+3\widetilde{n}\numtargates, k', \ge \min\{d', \widetilde{d}\}]]$. Within this code, the target logical gate $\addczrot$ is implemented via the transversal physical $Z$-rotation 
\begin{equation*}
    \prod_{i=1}^\kappa \prod_{j=1}^{M_i} U\left(\ell+1,w_{1}^{(i,j)}\right)\left(U\left(\ell+1,w_{2}^{(i,j)} \right)\right)^{\widetilde{k}} U\left(\ell+1,w_{3}^{(i,j)} \right)
\end{equation*}
where the vectors $w_{1}^{(i,j)}$, $w_{2}^{(i,j)}$, and $w_{3}^{(i,j)}$ are given by
\begin{align*}
    w_1^{(i,j)} &= \left(\mathbf{0}_{n' + 3\widetilde{n}\left(\left(\sum_{p=1}^{i-1}M_p\right) + j-1\right)}, -\widetilde{w}, \mathbf{0}_{\widetilde{n}}, \mathbf{0}_{\widetilde{n}}, \mathbf{0}_{3\widetilde{n}\left( M_i - j + \left(\sum_{p=i+1}^\kappa M_p\right)\right)} \right), \\
    w_2^{(i,j)} &= \left(\mathbf{0}_{n' + 3\widetilde{n}\left(\left(\sum_{p=1}^{i-1}M_p\right) + j-1\right)}, \mathbf{0}_{\widetilde{n}}, \widetilde{w}, \mathbf{0}_{\widetilde{n}}, \mathbf{0}_{3\widetilde{n}\left( M_i - j + \left(\sum_{p=i+1}^\kappa M_p\right)\right)} \right), \\
    w_3^{(i,j)} &= \left(\mathbf{0}_{n' + 3\widetilde{n}\left(\left(\sum_{p=1}^{i-1}M_p\right) + j-1\right)}, \mathbf{0}_{\widetilde{n}}, \mathbf{0}_{\widetilde{n}}, \widetilde{w}, \mathbf{0}_{3\widetilde{n}\left( M_i - j + \left(\sum_{p=i+1}^\kappa M_p\right)\right)} \right).
\end{align*}

As in Section~\ref{ft_sq}, we conclude this section by constructing CSS codes that transversally realize any addressable logical $\czrot$ gate. A CSS code can realize any addressable logical $\czrot$ gate if and only if it can realize $\czrot$ between any pair of logical qubits. For the construction, we invoke the appending framework with $\numtargates=\binom{k'}{2}$, and for each pair of logical qubits $1 \le i < j \le k'$, define a target logical gate $\custlogicalU{i,j}=C_{i,j}\left( R_Z\left( \frac{\pi}{2^\ell} \right)\right)$. This amounts to instantiating the above construction with the target address matrices $A_{i,j}$ whose sole non-zero entry is $A_{i,j}(i,j)=1$.

\section{CSS Families Realizing Addressable Logical Single-Qubit and $1$-Controlled $Z$-Rotations}\label{css_families}


In this section, we demonstrate the utility of the appending framework by constructing CSS code families that realize addressable logical single-qubit and intra-block $1$-controlled $Z$-rotations via transversal physical $Z$-rotations. We subsequently compare these families against existing constructions in the literature. 

\subsection{Systematic Construction of CSS Families via the Appending Framework}\label{gen_family_framework}

This subsection outlines a systematic approach for constructing CSS code families that transversally realize addressable logical $Z$-rotations via the appending framework. 

\subsubsection{CSS Families Realizing a Sequence of Targeted Addressable Logical Gates}\label{target_families_framework}
Let $(A_i)_{i \in \mathbb{N}}$ be an arbitrary but fixed sequence of target address tensors, where $A_i \in \mathbb{F}_2^{k_i'}$ for single-qubit $Z$-rotations and $A_i \in \mathbb{F}_2^{k_i' \times k_i'}$ for $1$-controlled $Z$-rotations. Our objective is to construct a CSS code family $((C_{1,i}, C_{2,i})_{\tn{CSS}})_{i \in \mathbb{N}}$ that encodes $k_i'$ logical qubits and transversally realizes the corresponding addressable logical gate: $\sqzrot_{A_i}$ or $C_{A_i}\left(R_Z\left( \frac{\pi}{2^\ell}\right)\right)$. The generalization to $(m-1)$-controlled $Z$-rotations, $C_{A_i}^{(m-1)}\left(R_Z\left( \frac{\pi}{2^\ell}\right)\right)$, is analogous and is detailed in Section~\ref{css_families_mcz} of Appendix~\ref{ft_mc_zrot}.

Achieving this via the appending framework requires an auxiliary CSS code family capable of transversally realizing the required logical transversal $Z$-rotations. Specifically, the auxiliary family must realize the logical transversal $\sqzrot$ for the single-qubit construction, and the logical transversal $R_Z \left( \frac{\pi}{2^{\ell+1}}\right)$ for the $1$-controlled-$Z$ construction. 

Since the primary CSS code family is unconstrained, we instantiate it with an asymptotically good family of CSS codes $((C_{1,i}', C_{2,i}')_{\tn{CSS}})_{i \in \mathbb{N}}$ from~\cite{ALT2001}. As the blocklength $n_i'$ grows arbitrarily large, the parameters $[[n_i', k_i', d_i']]$ of this family scale linearly, maintaining $k_i' \ge R n_i'$ and $d_i' \ge \min\{ d_{\min}(C_{1,i}') , d_{\min}({(C_{2,i}')}^\perp)\} \ge \delta n_i' = \Omega(n_i')$, for fixed constants $R, \delta \in (0,1)$. 

Our choice of the auxiliary code family $((\widetilde{C}_{1,r}, \widetilde{C}_{2,r})_{\tn{CSS}})_{r \in \mathbb{N}}$ is determined entirely by the target addressable logical $Z$-rotations. This auxiliary family has parameters $[[\widetilde{n}_r, \widetilde{k}_r, \widetilde{d}_r]]$, with $\widetilde{k}_r = \widetilde{n}_r^{\alpha + o(1)}$ and $\widetilde{d}_r \ge \min\{ d_{\min}(\widetilde{C}_{1,r}), d_{\min}({(\widetilde{C}_{2,r})}^\perp)\} = \widetilde{n}_r^{\beta+o(1)}$, for fixed constants $\alpha, \beta \in (0,1)$. We detail the specific auxiliary families achieving state-of-the-art parameters in the next subsection.

\emph{Code Parameters of the Final Derived CSS Code:} For each $i \in \mathbb{N}$, we construct appending matrices for the primary code $(C_{1,i}', C_{2,i}')_{\tn{CSS}}$ using the auxiliary code $(\widetilde{C}_{1,r}, \widetilde{C}_{2,r})_{\tn{CSS}}$. By Lemma~\ref{lem:params}, the resulting CSS code $(C_{1,i,r}, C_{2,i,r})_{\tn{CSS}}$ achieves the parameters $[[n_{i,r}, k_{i,r}, d_{i,r}]]$, given by
\begin{align*}
    n_{i,r}&= n_{i}' + \widetilde{n}_r M_{i,r}, \\
    k_{i,r}&=k_i', \\
    d_{i,r} &\ge \min\{d_i', \widetilde{d}_r\}, 
\end{align*}
where $M_{i,r}$ denotes the total number of auxiliary blocks appended to the primary code. 

The value of $M_{i,r}$ depends on the target gate. In the single-qubit case, it is bounded by
\begin{equation*}
    M_{i,r}=\left\lceil \frac{\lvert \supp(A_i) \rvert}{\widetilde{k}_r} \right\rceil \le \frac{\lvert \supp(A_i) \rvert}{\widetilde{k}_r} + 1.    
\end{equation*}
For the $1$-controlled-$Z$ case, letting $\kappa_i = \lvert \{j \in [k_i'-1] : A_i(j,j+1:k_i') \ne 0\} \rvert$ denote the number of active control qubits, we have
\begin{align*}
    M_{i,r} &= 3 \sum_{\substack{j \in [k_i'-1] : \\ A_i(j,j+1:k_i') \ne 0}} \left\lceil \frac{\lvert \supp(A_i(j,j+1:k_i')) \rvert}{\widetilde{k}_r} \right\rceil \\
    &\le 3\sum_{\substack{j \in [k_i'-1] : \\ A_i(j,j+1:k_i') \ne 0}} \left( \frac{\lvert \supp(A_i(j,j+1:k_i')) \rvert}{\widetilde{k}_r} + 1 \right) = 3 \left( \frac{\lvert \supp(A_i) \rvert}{\widetilde{k}_r} + \kappa_i \right).
\end{align*}

To control the scaling of the support size of the target address tensors relative to the blocklength of the primary code, we assume $ \lvert \supp(A_i) \rvert \le (n_i')^{\epsilon}$ for a constant $\epsilon \in (0,1]$ in the single-qubit case, and a constant $\epsilon \in (0,2]$ in the $1$-controlled-$Z$ case. For the latter, we additionally bound the number of active control qubits by $\kappa_i \le (n_i')^{\epsilon'}$ for some constant $\epsilon' \in (0,1]$. 

To express the parameters of the resulting family $(C_{1,i,r}, C_{2,i,r})_{\tn{CSS}}$ entirely in terms of the primary blocklength $n_i'$, we fix a constant $\eta > 0$ and assume that for each $i \in \mathbb{N}$ , there exists a suitable index $r_i \in \mathbb{N}$ such that $(n_i')^\eta = \Theta(\widetilde{n}_{r_i})$. While this property is not universally satisfied by all auxiliary families, the specific constructions introduced in the next section meet this constraint. Omitting the index $r_i$ from the subscript for brevity, the physical blocklength $n_i$ of $(C_{1,i}, C_{2,i})_{\tn{CSS}}$ scales as follows: for the single-qubit case,
\begin{equation*}
    n_{i} = \mathcal{O}\left(n_i' + (n_i')^{\eta(1 - \alpha) + \epsilon + o(1)} + (n_i')^\eta\right) = \mathcal{O}\left({(n_i')}^{\max\{1, \eta, \eta(1-\alpha)+\epsilon+ o(1)\}}\right),
\end{equation*}
and for the $1$-controlled-$Z$ case,
\begin{equation*}
    n_{i} = \mathcal{O}\left(n_i' + 3(n_i')^{\eta(1 - \alpha) + \epsilon + o(1)} + 3(n_i')^{\eta+\epsilon'}\right) = \mathcal{O}\left({(n_i')}^{\max\{1, \eta + \epsilon', \eta(1-\alpha)+\epsilon+ o(1)\}}\right).
\end{equation*}
Finally, the logical dimension $k_i$ and minimum distance $d_i$ for both cases scale as
\begin{align*}
    k_i &\ge n_i' R = \Omega(n_i'), \\
    d_i &\ge \min\left\{ n_i' \delta, (n_i')^{\eta \beta + o(1)}\right\} = \Omega\left( {(n_i')}^{\min\{1, \eta \beta + o(1)\}} \right).
\end{align*}

Since we consider auxiliary code families that exhibit sub-linear growth in their logical dimension (i.e., $0 < \alpha < 1$), we address this regime first.


\emph{Construction~1, Sub-linear Target Support Sizes $(0 < \epsilon < 1)$:}\label{sublinear_target}
For both the single-qubit and $1$-controlled-$Z$ cases, consider the regime where $0 < \epsilon, \alpha < 1$, corresponding to sub-linear growth in the target support sizes. Consequently, for $1$-controlled $Z$-rotations, the number of active control qubits must also grow sub-linearly, implying $0 < \epsilon' < 1$. By choosing any constant $0 < \eta < \min\left\{1,\frac{1-\epsilon}{1-\alpha}\right\}$ for the single-qubit case, and $0< \eta < \min\left\{1-\epsilon',\frac{1-\epsilon}{1-\alpha}\right\}$ for the $1$-controlled-$Z$ case, the overall code parameters scale as $[[\Theta(n_i'), \Theta(n_i'), \Omega((n_i')^{\eta \beta+ o(1)})]]$ (since $\eta, \beta \le 1$). Therefore, this regime yields a code family with an asymptotically constant rate and a polynomial minimum distance.

\emph{Construction~2, At Least Linear Target Support Sizes ($\epsilon=1$ for the single-qubit case, $\epsilon \in [1,2]$ for the $1$-controlled-$Z$ case):}
In the regime where the target support size grows at least linearly (i.e., $\epsilon \ge 1$), achieving an asymptotically constant rate is impossible within this framework. Instead, the overall code parameters scale as $[[\Theta({(n_i')}^{\zeta + o(1)}), \Omega(n_i'), \Omega( {(n_i')}^{\min\{1, \eta \beta + o(1)\}} )]]$, where $\zeta=\max\{\eta, \eta(1-\alpha) + \epsilon\}$ for the single-qubit case, and $\zeta=\max\{\eta + \epsilon', \eta(1-\alpha) + \epsilon\}$ for the $1$-controlled-$Z$ case.

The dominant term in the minimum distance depends on the choice of $\eta$. If $\eta \in (0, 1/\beta]$, the auxiliary code controls the distance, yielding parameters $[[\Theta({(n_i')}^{\zeta+o(1)}), \Omega(n_i'), \Omega( {(n_i')}^{\eta \beta + o(1)} )]]$. If $\eta \in (1/\beta, \infty)$, the primary code controls the distance, yielding parameters $[[\Theta((n_i')^{\zeta+o(1)}), \Omega(n_i'), \Omega(n_i' )]]$. 

Rescaling these parameters to express them in terms of the final overall blocklength $n_i$, we obtain our final scaling bounds. For $\eta \in (0, 1/\beta]$, the code parameters scale as $[[\Theta(n_i), \Omega({(n_i)}^{1/\zeta + o(1)}), \Omega( {(n_i)}^{(\eta \beta)/\zeta + o(1)} )]]$. For $\eta \in (1/\beta, \infty)$, they scale as $[[\Theta(n_i), \Omega({(n_i)}^{1/\zeta + o(1)}), \Omega({(n_i)}^{1/\zeta + o(1)})]]$. Consequently, for any positive $\eta$, this construction yields polynomial scaling for both logical dimension and minimum distance.

We now turn to the auxiliary code families that exhibit linear growth in their logical dimension (i.e., $\alpha=1$ and the $o(1)$ term in the exponent is zero).

\emph{Construction~3, Linear Target Support Sizes $(0 < \epsilon \le 1)$ and Sub-linear Active Control Qubits $(0 < \epsilon' < 1)$:} By choosing any fixed constant $0 < \eta \le 1$ for the single-qubit case, and $0 < \eta \le 1-\epsilon'$ for the $1$-controlled-$Z$ case, the overall code parameters scale as $[[\Theta(n_i'), \Theta(n_i'), \Omega((n_i')^{\eta \beta + o(1)})]]$ (since $\eta, \beta \le 1$). Therefore, even in this linear target support sizes regime, we obtain a family of codes with an asymptotically constant rate. However, as we will discuss in the next subsection, satisfying these specific constraints on $\alpha$ and the $o(1)$ term is possible only for auxiliary families realizing the logical transversal $S$ and $T$ gates. For the auxiliary families realizing the logical transversal $T$ gate, the minimum distance grows sub-polynomially; consequently, the minimum distance of the final derived family also exhibits sub-polynomial growth.

\begin{remark}\label{rmk:linear}
If the target support size grows super-linearly---or if the number of active control qubits grows at least linearly for the $1$-controlled-$Z$ case---while the necessary constraints on $\alpha$ and the $o(1)$ term remain satisfied, we can directly apply the scaling analysis from Construction~2. The asymptotic code parameters are identical, with the substitution that $\eta(1-\alpha)=0$ since $\alpha=1$.
\end{remark}

The parameters of our constructions are summarized in Table~\ref{constructions_table}.
\renewcommand{\arraystretch}{1}
\begin{table*}[!h]\label{Constructions}
\begin{center}
\begin{tabular}{| p{5.9cm} | p{8.2cm} |}
    \hline
\textbf{Constructions} & \textbf{Single-qubit $Z$-rotations}   \\
 \hline
\makecell[l]{Construction~1 \\ (Sub-linear target support sizes, \\
$0 < \epsilon, \epsilon', \alpha < 1$)}  & \makecell[l]{$[[\Theta(n_i'), \Theta(n_i'), \Omega((n_i')^{\eta \beta+ o(1)})]]$, \\ for single-qubit: $0 < \eta < \min\left\{1,\frac{1-\epsilon}{1-\alpha}\right\}$, \\ for $1$-controlled-$Z$: $0 < \eta < \min\left\{1-\epsilon',\frac{1-\epsilon}{1-\alpha}\right\}$} \\
\hline
\makecell[l]{Construction~2 \\ (At least linear target support sizes, \\
$\epsilon = 1$ for single-qubit, \\ $1 \le \epsilon \le 2$ for $1$-controlled-$Z$, \\ and
$0 < \alpha \le 1$)}  & \makecell[l]{For $\eta > 0$: \\ $[[n_i, \Omega(n_i^{\frac{1}{\zeta} + o(1)}), \Omega( n_i^{\frac{\min\{1,\eta \beta +o(1)\}}{\zeta}} )]]$, \\ for single-qubit: $\zeta=\max\{\eta, \eta(1-\alpha) + \epsilon\}$, \\ for $1$-controlled-$Z$: $\zeta=\max\{\eta +\epsilon', \eta(1-\alpha) + \epsilon\}$} \\
\hline
\makecell[l]{Construction~3 \\ (Linear target support sizes and \\ sub-linear active control qubits, \\
$0 < \epsilon \le 1, 0 < \epsilon' < 1, \widetilde{k}_r =\Omega(\widetilde{n}_r)$)}  & \makecell[l]{$[[\Theta(n_i'), \Theta(n_i'), \Omega((n_i')^{\eta \beta + o(1)})]]$, \\ for single-qubit: $0 < \eta \le 1$, \\ for $1$-controlled-$Z$: $0 < \eta \le 1-\epsilon'$} \\
\hline
\end{tabular}
\vspace{0.5cm}
\caption{Asymptotic parameters for the CSS code family constructions.}\label{constructions_table}
\end{center}
\end{table*}

\subsubsection{CSS Families Realizing All Addressable Logical Gates}\label{css_families_all_gates}
Thus far, we have focused on constructing families of codes that realize specific target addressable logical gates. We now broaden our scope to construct code families capable of realizing all addressable logical gates. As detailed at the conclusion of Sections~\ref{ft_sq} and~\ref{ft_cs}, the resulting CSS code $(C_{1,i,r}, C_{2,i,r})_{\tn{CSS}}$ achieves the parameters $[[n_{i,r}, k_{i,r}, d_{i,r}]]$, given by
\begin{equation*}
    n_{i,r}=n_i' + \widetilde{n}_r M_{i,r} = \Theta(n_i' + (n_i')^{\eta} M_{i,r} ), \quad k_{i,r}=\Omega(n_i'), \quad d_{i,r} = \Omega\left( {(n_i')}^{\min\{1, \eta \beta + o(1)\}} \right), 
\end{equation*}
where $M_{i,r} = k_i'=\Theta(n_i')$ for the single-qubit case, and $M_{i,r} = \binom{k_i'}{2} = \Theta((n_i')^2)$ for the $1$-controlled-$Z$ case. 

Consequently, achieving an asymptotically constant rate is impossible in this regime. Instead, the overall code parameters scale as $[[\Theta({(n_i')}^{\zeta}), \Omega(n_i'), \Omega\left( {(n_i')}^{\min\{1, \eta \beta + o(1)\}} \right)]]$, where $\zeta=1+\eta$ for the single-qubit case, and $\zeta=2+\eta$ for the $1$-controlled-$Z$ case. 

As before, the dominant term in the minimum distance depends on the choice of $\eta$. If $\eta \in (0, 1/\beta]$, the auxiliary code controls the distance, yielding parameters $[[\Theta((n_i')^{\zeta}), \Omega(n_i'), \Omega( {(n_i')}^{\eta \beta + o(1)} )]]$. If $\eta \in (1/\beta, \infty)$, the primary code controls the distance, yielding parameters $[[\Theta({(n_i')}^{\zeta}), \Omega(n_i'), \Omega( n_i' )]]$. 

Rescaling these parameters to express them in terms of the final overall blocklength $n_i$, we obtain our final scaling bounds. For $\eta \in (0, 1/\beta]$, the code parameters scale as $[[\Theta(n_i), \Omega({n_i}^{1/\zeta}), \Omega( {n_i}^{(\eta \beta)/\zeta + o(1)} )]]$. For $\eta \in (1/\beta, \infty)$, they scale as $[[\Theta(n_i), \Omega({n_i}^{1/\zeta}), \Omega({n_i}^{1/\zeta})]]$. Consequently, for any positive $\eta$, we achieve polynomial scaling for both the logical dimension and minimum distance.


\subsection{CSS Code Families Realizing Logical Transversal 
Single-Qubit $Z$-Rotations}\label{sec_families}

In this subsection, we present both explicit and non-explicit auxiliary CSS code families from~\cite{W2026} that realize logical transversal $\sqzrot$ via physical transversal $\sqzrot^\dag$. Furthermore, we demonstrate that these auxiliary code families satisfy the hypotheses required in the preceding subsection. This verification ensures they can be integrated into the appending framework, enabling the construction of families realizing addressable logical $Z$-rotations in the sections that follow. 

\subsubsection{$\ell=1$}
As detailed in~\cite{MN2025}, puncturing $\widehat{t} = \left\lfloor \frac{\min\{ \widehat{k}, \widehat{d}, \widehat{d}^\perp \}}{2} \right\rfloor$ coordinates of a $2^{\ell+1}$-divisible classical $[\widehat{n},\widehat{k},\widehat{d}]$ code with dual distance $\widehat{d}^\perp$ yields an auxiliary CSS code $(\widehat{C}_1, \widehat{C}_2)_{\tn{CSS}}$ with parameters $[[ \widehat{n} - \widehat{t}, \widehat{t}, \min\{ \widehat{d}_X, \widehat{d}_Z \}]]$. The $X$- and $Z$-distances satisfy $\widehat{d}_X \ge d_{\min}(\widehat{C}_1) \ge \widehat{d}-\widehat{t}$ and $\widehat{d}_Z \ge d_{\min}(\widehat{C}_2^\perp) \ge \widehat{d}^\perp - \widehat{t} $, and the physical transversal $\sqzrot^\dag$ realizes the logical transversal $\sqzrot$ within this CSS code. 

For the $S$ gate $(\ell=1)$, this construction requires a doubly-even ($4$-divisible) classical code. Applying this to the asymptotically good family of self-dual doubly-even codes from~\cite{SMT1972,PW2007} yields an asymptotically good CSS family $((\widetilde{C}_{1,r}, \widetilde{C}_{2,r})_{\tn{CSS}})_{r \in \mathbb{N}}$ with parameters $[[\widetilde{n}_r, \widetilde{k}_r = \Theta(\widetilde{n}_r), \widetilde{d}_r=\Theta( \widetilde{n}_r)]]$, in which the physical transversal $S^\dag$ realizes the logical transversal $S$. However, because the underlying doubly even codes in~\cite{SMT1972,PW2007} are non-explicit, this resulting CSS family is non-explicit.

To construct an explicit asymptotically good family, we rely on the relaxed characterization in Proposition~\ref{prop:target_addr_log_sq_basis_relaxed} (with $\ell=1$) to transversally realize the logical transversal $S$. This relaxed characterization establishes that, rather than requiring doubly-even codes, it suffices to use self-orthogonal codes in the puncturing construction above. Explicit asymptotically good binary self-orthogonal codes with asymptotically good duals are known to exist~\cite{ALT2001}. Consequently, we establish the following:

\emph{Family 1:} There exists explicit asymptotically good CSS family $[[\widetilde{n}_r, \widetilde{k}_r=\Omega(\widetilde{n}_r),  \widetilde{d}_r]]$ with $\widetilde{d}_r \ge \min\{ d_{\min}(\widetilde{C}_{1,i}), d_{\min}({(\widetilde{C}_{2,r})}^\perp)\} = \Omega(\widetilde{n}_r)$, wherein the physical transversal $S^\dag$ realizes the logical transversal $S$.

\subsubsection{$\ell \ge 2$, Sub-linear Logical Dimension:}
In the standard puncturing construction, puncturing $\widehat{t}$ coordinates reduces the minimum distance by at most $\widehat{t}$, yielding the naive lower bounds $d_{\min}(\widehat{C}_1) \ge \widehat{d}-\widehat{t}$ and $d_{\min}((\widehat{C}_2)^\perp) \ge \widehat{d}^\perp - \widehat{t}$. However, when puncturing decreasing monomial codes on a down set, these lower bounds can be improved substantially. This was initially observed in~\cite{W2026}, which utilized this property to construct explicit and non-explicit CSS code families wherein the physical transversal $T^{\dag}$ realizes the logical transversal $T$. As further noted in that work, these arguments extend naturally to arbitrary single-qubit $Z$-rotations ($\ell \ge 2$) yielding the following code families:

\emph{Family 2:} For any fixed integer $\ell \ge 2$ and constant $\alpha_{1,\ell} \in \left(0,H_2\left(\frac{1}{\ell}\right)\right)$, there exist explicit CSS code families wherein the physical transversal $\sqzrot^\dag$ realizes the logical transversal $\sqzrot$. These codes achieve parameters $[[\widetilde{n}_r, \Omega(\widetilde{n}_r^{\alpha_{1,\ell} + o(1)}), \Omega(\widetilde{n}_r^{\beta_{1,\ell}(\alpha_{1,\ell}) + o(1)})]]$, where the distance exponent is given by
\begin{equation*}
    \beta_{1,\ell}(\alpha_{1,\ell}) = \begin{cases}
        \frac{1}{\ell+1}, & \text{if } 0 < \alpha_{1,\ell} \le H_2\left( \frac{1}{2(\ell+1)} \right), \\
        \frac{1-\ell q^{*}(\alpha_{1,\ell})}{\ell+1}, & \text{if } H_2\left( \frac{1}{2(\ell+1)} \right) \le \alpha_{1,\ell} < H_2\left( \frac{1}{\ell} \right).
    \end{cases}
\end{equation*}
Here, $H_2$ denotes the binary entropy function with inverse $H_2^{-1}:(0,1) \rightarrow (0,1/2)$, and $q^{*}(\alpha_{1,\ell})$ is the unique positive solution to  
\begin{equation*}
    1- \ell q= (1+q) H_2 \left( \frac{(\ell+1) H_2^{-1}(\alpha_{1,\ell})}{1+q}\right).
\end{equation*}
The quantity $q^{*}(\alpha_{1,\ell})$ is a strictly increasing function of $ \alpha_{1,\ell} \in \Big[ H_2\left( \frac{1}{2(\ell+1)} \right), H_2\left( \frac{1}{\ell}  \right) \Big)$ ranging from $q^{*}\left(H_2\left( \frac{1}{2(\ell+1)} \right)\right)=0$ to $q^{*}\left(H_2\left( \frac{1}{\ell} \right)\right)=\frac{1}{\ell}$.  

\emph{Family 3:} For any fixed integer $\ell \ge 2$ and constant $\alpha_{2,\ell} \in \left(0,H_2\left(\frac{1}{\ell}\right)\right)$, there exist non-explicit CSS code families wherein the physical transversal $\sqzrot^\dag$ realizes the logical transversal $\sqzrot$. These codes achieve parameters $[[\widetilde{n}_r, \Omega(\widetilde{n}_r^{\alpha_{2,\ell} + o(1)}), \Omega(\widetilde{n}_r^{\beta_{2,\ell}(\alpha_{2,\ell}) + o(1)})]]$, where the distance exponent is given by
\begin{equation*}
    \beta_{2,\ell}(\alpha_{2,\ell}) = \begin{cases}
        \frac{1}{\ell}, & \text{if } 0 < \alpha_{2,\ell} \le \frac{1}{\ell}, \\
        \frac{1}{\ell} - \frac{\alpha_{2,\ell} - 1/\ell}{\ell(\ell+1)\left(H_2\left(\frac{1}{(\ell+1)}\right) - 1/\ell\right)}, & \text{if } \frac{1}{\ell} \le \alpha_{2,\ell} \le H_2\left( \frac{1}{\ell+1} \right), \\
        1 - \ell H_2^{-1}(\alpha_{2,\ell}), & \text{if } H_2\left( \frac{1}{\ell+1} \right) \le \alpha_{2,\ell} < H_2\left( \frac{1}{\ell} \right).
    \end{cases}
\end{equation*}

In both Families~2 and~3, the minimum distance of the auxiliary code is at least $\min\{ d_{\min}(\widetilde{C}_{1,r}), d_{\min}(\widetilde{C}_{2,r}^\perp) \} = \Omega(\widetilde{n}_r^{\beta_{j,\ell}(\alpha_{j,\ell})+o(1)})$, where $j \in \{1,2\}$ corresponds to the respective family. 

\subsubsection{$\ell=2$, Linear Logical Dimension}
While CSS Families~2 and~3 achieve only a sub-linear logical dimension, a linear logical dimension is crucial for optimizing the parameters of certain constructions within our framework. Using monomial codes,~\cite{W2026} established explicit CSS code families that achieve a linear logical dimension wherein the physical transversal $T^{\dagger}$ realizes the logical transversal $T$. However, the minimum distance of this family grows  sub-polynomially with the blocklength.

\emph{Family 4:} For any constant $R_3 \in (0,1/3)$, there exists an explicit family of CSS codes wherein the physical transversal $T^\dag$ realizes the logical transversal $T$. These codes achieve parameters $[[ \widetilde{n}_r, \widetilde{k}_r, \widetilde{d}_r]]$, where the logical dimension scales as $\widetilde{k}_r = (R_3+o(1))\widetilde{n}_r=\Omega(\widetilde{n}_r)$, and the minimum distance is bounded by
\begin{equation*}
    \widetilde{d}_r \ge \min\{ d_{\min}(\widetilde{C}_{1,r}), d_{\min}(\widetilde{C}_{2,r}^\perp) \} = \Omega(2^{(c_{R_3}+o(1)) \sqrt{\log(\widetilde{n}_r)}})
\end{equation*}
Here, $c_{R_3}=-\Phi^{-1}\left(\frac{2R_3}{1+R_3}\right)>0$, where $\Phi$ denotes the cumulative distribution function of the standard normal distribution. We note that these monomial code-based constructions do not readily generalize to yield a linear logical dimension for $\sqzrot$ gates when $\ell > 2$. The parameters of the auxiliary code families are summarized in Table~\ref{aux_families}:
\renewcommand{\arraystretch}{1}
\begin{table*}[!h]\label{aux_families}
\begin{center}
\begin{tabular}{| p{6.5cm} | p{5cm} | p{3cm} |}
    \hline
\textbf{Families} & \textbf{Logical gates} & \textbf{Explicit/Non-explicit}   \\
 \hline
\makecell[l]{Family~1 \\ $[[\widetilde{n}_r, \Theta(\widetilde{n}_r), \Theta(\widetilde{n}_r)]]$}  & {Logical transversal $S$}  & {Explicit} \\
\hline
\makecell[l]{Family~2 \\ $\ell \ge 2$, $[[\widetilde{n}_r, \Omega(\widetilde{n}_r^{\alpha_{1,\ell}+o(1)}), \Omega(\widetilde{n}_r^{\beta_{1,\ell}+o(1)})]]$}  & Logical transversal $\sqzrot$ & Explicit \\
\hline
\makecell[l]{Family~3 \\ $\ell \ge 2$, $[[\widetilde{n}_r, \Omega(\widetilde{n}_r^{\alpha_{2,\ell}+o(1)}), \Omega(\widetilde{n}_r^{\beta_{2,\ell}+o(1)})]]$}  & Logical transversal $\sqzrot$ & Non-explicit \\
\hline
\makecell[l]{Family~4 \\ $[[\widetilde{n}_r,\Omega(\widetilde{n}_r), 2^{(c_{R_3}+o(1))\sqrt{\log(\widetilde{n}_r)}}]]$}  & Logical transversal $T$  & Explicit\\
\hline
\end{tabular}
\vspace{0.5cm}
\caption{Auxiliary CSS code families}\label{aux_families}
\end{center}
\end{table*}

We conclude this subsection by noting that the asymptotically good CSS codes from~\cite{SJ2026}, which transversally realize the logical transversal $\sqzrot$, cannot be utilized as auxiliary codes in our framework. This limitation arises from their distance bounding properties. While these families achieve $d_{\min}(\widetilde{C}_{2,r}^\perp \setminus \widetilde{C}_{1,r}^\perp) = \Omega(\widetilde{n}_r)$, they fail to satisfy the dual distance requirement $d_{\min}(\widetilde{C}_{2,r}^\perp) = \Omega(\widetilde{n}_r)$ required by our construction. Specifically, there exist $Z$-stabilizers in $\widetilde{C}_{1,r}^\perp$ with asymptotically constant weights, which prevents the overall dual distance from scaling as $\Omega(\widetilde{n}_r)$. Moreover, despite possessing some constant-weight stabilizers, these families do not qualify as quantum LDPC codes, as they lack a complete generating set for $\widetilde{C}_{1,r}^\perp$ with asymptotically constant weights.

\subsubsection{Polynomial Growth of Auxiliary Code Blocklengths Relative to the Primary Code}

Families~1--4 satisfy all the required constraints on auxiliary codes established in Subsection~\ref{gen_family_framework}. The only remaining condition to verify is that, for any fixed constant $\eta > 0$ and each $i \in \mathbb{N}$, there exists an index $r_i$ such that $(n_i')^{\eta}=\Theta(\widetilde{n}_{r_i})$. 

The sequence of primary blocklengths $(n_i')_{i \in \mathbb{N}}$ is strictly increasing and diverges to $\infty$. Assume that the blocklengths of the auxiliary families satisfy $\lim_{r \to \infty} \frac{\widetilde{n}_r}{\widetilde{n}_{r-1}}=\gamma$ for some constant $1 \le \gamma < \infty$. Because the sequence $(\widetilde{n}_r)_{r \in \mathbb{N}}$ is similarly increasing and unbounded, for each $i \in \mathbb{N}$, we can identify an index $r_i$ such that
\begin{equation*}
    \widetilde{n}_{r_i-1} \le {(n_i')}^{\eta} \le \widetilde{n}_{r_i}.
\end{equation*}

Families~1--4 all  satisfy this property. As a concrete example, consider Family~2, where $\widetilde{n}_r=2^r-k$ with $k=o(2^r)$. In this case, we obtain $\lim_{r \to \infty} \frac{\widetilde{n}_r}{\widetilde{n}_{r-1}}=2$. Similarly, for Family~4, where $\widetilde{n}_r=2^r-k$ with $k= R_3 2^r + o(2^r)$ for some constant $R_3 \in (0,1/3)$, factoring out the dominant terms yields $\lim_{r \to \infty} \frac{\widetilde{n}_r}{\widetilde{n}_{r-1}}=2$. Analogous calculations hold for Families~1 and~3. 

\subsection{CSS Code Families Realizing Addressable Logical Single-Qubit $Z$-Rotations}

We now construct code families by incorporating the auxiliary codes of  Section~\ref{sec_families} into the general framework formulated in Section~\ref{gen_family_framework}. For any fixed integer $\ell \ge 1$, we focus on constructing two types of code families: one realizing a targeted sequence of addressable logical gates $(\sqzrot_{A_i})_{i \in \mathbb{N}}$ with $A_i \in \mathbb{F}_2^{k_i'}$, and another realizing any addressable logical $\sqzrot$ gate. 

\subsubsection{Addressable Logical $S$ Gates $(\ell=1)$}

We utilize Family~1 to construct both types of code families realizing addressable logical $S$ gates. Since $\alpha=\beta=1$ and the $o(1)$ terms in the exponents of both the distance and logical dimension of Family~1 are equal to zero, choosing $\eta=1$ in Construction~3 of Section~\ref{gen_family_framework}, yields an asymptotically good family with parameters $[[\Theta(n_i'), \Theta(n_i'), \Theta(n_i')]]$ that transversally realizes the targeted addressable logical $S_{A_i}$ gates. 

To realize all addressable logical $S$ gates, we apply the framework from Section~\ref{css_families_all_gates}. Because $\alpha=\beta=1$ and the $o(1)$ terms vanish, for any $\eta > 0$, we construct a code family with parameters  $[[\Theta(n_i), \Omega({(n_i)}^{1/(1+\eta)}), \Omega({(n_i)}^{\min\{1,\eta\}/(1+\eta)})]]$ that transversally realizes any addressable logical $S$ gate. 

\subsubsection{Addressable Targeted Sequence of Logical $\sqzrot$ gates $(\ell \ge 2)$ for Sub-linear Target Support Sizes}
We utilize Families~2 and~3 to construct these code families. For any fixed integer $\ell \ge 2$, since $0 < \epsilon, \alpha_{1,\ell}, \alpha_{2,\ell} < 1$, we apply Construction~1 for sub-linear target support sizes established in Section~\ref{gen_family_framework}. Specifically, for any constants $0 < \alpha_{1,\ell}, \alpha_{2,\ell} < H_2(\frac{1}{\ell})$, this framework yields two classes of code families that transversally realize addressable logical $\sqzrot_{A_i}$ gates:

First, by choosing any constant $0 < \eta < \min\left\{1, \frac{1-\epsilon}{1-\alpha_{1,\ell}}\right\}$, we obtain an explicit code family with parameters $[[\Theta(n_i'),\Theta(n_i'), \Omega({(n_i')}^{\eta \beta_{1,\ell}(\alpha_{1,\ell}) +o(1) })]]$.

Second, by choosing any constant $0 < \eta < \min\left\{1, \frac{1-\epsilon}{1-\alpha_{2,\ell}}\right\}$, we obtain a non-explicit code family with parameters $[[\Theta(n_i'),\Theta(n_i'), \Omega({(n_i')}^{\eta \beta_{2,\ell}(\alpha_{2,\ell}) +o(1) })]]$, yielding improved distance scaling at the expense of explicitness.

\subsubsection{Addressable Targeted Sequence of Logical $\sqzrot$ gates $(\ell \ge 2)$ for Linear Target Support Sizes}
We utilize Families~2 and~3 to construct these code families. For any fixed integer $\ell \ge 2$, since $\epsilon=1$ for the linear scaling regime, we apply Construction~2 of Section~\ref{gen_family_framework}. Specifically, for any constants $0 < \alpha_{1,\ell}, \alpha_{2,\ell} < H_2\left(\frac{1}{\ell}\right)$, this framework yields two classes of code families that transversally realize the targeted addressable logical $\sqzrot_{A_i}$ gates. 

First, utilizing Family~2, we obtain an explicit code family governed by the scaling factor $\zeta_{1} = \max\{\eta, \eta (1-\alpha_{1,\ell}) + 1\}$. For any $\eta > 0$, the family achieves parameters $[[\Theta(n_i),\Omega({n_i}^{1/\zeta_1 + o(1)}), \Omega({n_i}^{\min\{1,\eta \beta_{1,\ell}(\alpha_{1,\ell})\}/\zeta_1 +o(1)})]]$.

Second, utilizing Family~3, we obtain a non-explicit code family governed by the scaling factor $\zeta_2 = \max\{\eta, \eta (1-\alpha_{2,\ell}) + 1\}$. For any $\eta > 0$, the family achieves parameters $[[\Theta(n_i),\Omega({n_i}^{1/\zeta_2 + o(1)}), \Omega({n_i}^{\min\{1,\eta \beta_{2,\ell}(\alpha_{2,\ell})\}/\zeta_2 +o(1)})]]$.

\subsubsection{Addressable Targeted Sequence of Logical $T$ Gates for Linear Target Support Sizes Achieving Linear Logical Dimension}\label{logt_sep_cae}
We utilize Family~4 to construct these code families. Because $\epsilon=\alpha_3=1$ and the $o(1)$ term in the logical dimension exponent is zero, we apply the linear scaling regime established in Section~\ref{gen_family_framework}. Furthermore, because the minimum distance of Family~4 grows sub-polynomially---specifically, $\widetilde{d}_{r_i}=\Omega\left(2^{(c_{R_3}+o(1)) \sqrt{\eta \log(n_i')}}\right)$, given the relation ${(n_i')}^{\eta}=\Theta(\widetilde{n}_{r_i})$---it never surpasses the linear distance $d_i'=\Omega(n_i')$ of the primary code when $0 < \eta \le 1$. Consequently, the overall minimum distance is dictated entirely by the auxiliary code, yielding an explicit code family that transversally realizes addressable logical $T_{A_i}$ gates with parameters $[[\Theta(n_i'), \Theta(n_i'), \Omega\left(2^{(c_{R_3}+o(1)) \sqrt{\eta \log(n_i')}}\right)]]$ for any constant $0 < \eta \le 1$.

Consequently, for targeted logical $T$ gates, this construction yields a linear logical dimension at the expense of sub-polynomial minimum distance scaling.

\subsubsection{All Addressable Logical $\sqzrot$ Gates ($\ell \ge 2$)}

We utilize Families~2 and~3 to construct these code families. For any fixed integer $\ell \ge 2$, we apply the framework detailed in Section~\ref{css_families_all_gates}. Specifically, for any constants $0 < \alpha_{1,\ell}, \alpha_{2,\ell} < H_2\left(\frac{1}{\ell}\right)$ and a scaling factor $\zeta = 1 + \eta$, this framework yields two code families that transversally realize any addressable logical $\sqzrot$ gate:

First, utilizing Family~2, we obtain explicit code families. For any $\eta > 0$, the family achieves parameters $[[\Theta(n_i),\Omega({(n_i)}^{1/\zeta + o(1)}), \Omega({(n_i)}^{\min\{1,\eta \beta_{1,\ell}(\alpha_{1,\ell})\}/\zeta +o(1)})]]$.

Second, utilizing Family~3, we obtain non-explicit code families. For any $\eta > 0$, the family achieves parameters $[[\Theta(n_i),\Omega({(n_i)}^{1/\zeta + o(1)}), \Omega({(n_i)}^{\min\{1,\eta \beta_{2,\ell}(\alpha_{2,\ell})\}/\zeta +o(1)})]]$.

The constructed CSS families realizing addressable logical single-qubit $Z$-rotations are summarized in the following table:
\renewcommand{\arraystretch}{1}
\begin{table}[!h]
\begin{center}
\begin{tabular}{| p{8cm} | p{7.6cm} |}
    \hline
\textbf{Families} & \textbf{Logical gates}  \\
 \hline
\makecell[l]{Explicit: $[[n, \Theta(n), \Theta(n)]]$}  & \makecell[l]{Targeted addressable logical $S_{A_i}$ gates, \\ $\lvert \supp(A_i) \rvert \le {n}^{\epsilon}$, for $\epsilon \in (0,1]$} \\
\hline
\makecell[l]{Explicit:\\$[[n, \Omega(n^{1/(1+\eta)}), \Omega(n^{\min\{1,\eta\}/(1+\eta)})]]$, for $\eta > 0$}  & \makecell[l]{Any addressable logical $S$ gate} \\
\hline
\makecell[l]{ $\ell \ge 2$:\\
$[[n, \Theta(n), \Omega(n^{\eta \beta_{i,\ell}(\alpha_{i,\ell}) + o(1)})]]$, \\ for $0 < \eta < \min\left\{1, \frac{1-\epsilon}{1-\alpha_{i,\ell}}\right\}$, $0 < \alpha_{i,\ell} < H_2\left(\frac{1}{\ell}\right)$, \\ explicit for $i=1$, non-explicit for $i=2$}  & \makecell[l]{Targeted addressable logical $\sqzrot_{A_i}$ gates, \\ for $\epsilon \in (0,1)$} \\
\hline
\makecell[l]{ $\ell \ge 2$:\\
$[[n, \Omega(n^{1/\zeta_i+o(1)}), \Omega(n^{\min\left\{1,\eta \beta_{i,\ell}(\alpha_{i,\ell}) \right\}/\zeta_i + o(1)})]]$, \\ $\zeta_i=\max\{\eta, \eta (1-\alpha_{i,\ell}) + \epsilon\}$, \\ for $\eta > 0$, $0 < \alpha_{i,\ell} < H_2\left(\frac{1}{\ell}\right)$, \\ explicit for $i=1$, non-explicit for $i=2$}  & \makecell[l]{Targeted addressable logical $\sqzrot_{A_i}$ gates, \\ for $\epsilon \in (0,1]$} \\
\hline
\makecell[l]{ Explicit:\\
$[[n, \Theta(n), \Omega(2^{(c_{R_3}+o(1)) \sqrt{\eta \log(n)}})]]$, \\ for $\eta \in (0, 1]$, $R_3 \in (0,1/3)$}  & \makecell[l]{Targeted addressable logical $T$ gates, \\ for $\epsilon \in (0,1]$} \\
\hline
\makecell[l]{ $\ell \ge 2$:\\
$[[n, \Omega(n^{1/\zeta+o(1)}), \Omega(n^{\min\left\{1,\eta \beta_{i,\ell}(\alpha_{i,\ell}) \right\}/\zeta + o(1)})]]$, \\ $\zeta=1+\eta$, for $\eta > 0$, $0 < \alpha_{i,\ell} < H_2\left(\frac{1}{\ell}\right)$, \\ explicit for $i=1$, non-explicit for $i=2$}  & \makecell[l]{Any addressable logical $\sqzrot$ gate} \\
\hline
\end{tabular}
\vspace{0.5cm}
\caption{CSS families realizing addressable logical single-qubit $Z$-rotations}\label{sq_families}
\end{center}
\end{table}

We conclude this subsection by noting that, as established in Remark~6.8 of~\cite{SJ2026}, there exist asymptotically good code families that transversally realize any addressable logical $\sqzrot$ gate, for any fixed integer $\ell \ge 1$. While the families constructed in this section do not achieve these asymptotically good parameters, they serve as a proof-of-principle demonstrating the broad utility of the appending framework. 

Although the flexible nature of the appending framework does not always yield asymptotically good parameters for single-qubit $Z$-rotations, in the next section, we leverage this framework to construct novel code families realizing intra-block addressable logical $1$-controlled $Z$-rotations.

\subsection{CSS Code Families Realizing Addressable Logical $1$-Controlled $Z$-Rotations}\label{sec_css_families_1_controlled}

We once again incorporate the auxiliary code families from Section~\ref{sec_families} into the general framework established in Section~\ref{gen_family_framework}, this time applying it to $1$-controlled $Z$-rotations. For any fixed integer $\ell \ge 0$, we focus on constructing two types of code families: one realizing a targeted sequence of addressable logical gates $\left(C_{A_i}\left(\sqzrot\right)\right)_{i \in \mathbb{N}}$ with $A_i \in \mathbb{F}_2^{k_i' \times k_i'}$, and another realizing any addressable logical $\czrot$ gate.

\subsubsection{Addressable Logical $\cz$ Gates $(\ell=0)$}

We utilize Family~1 to construct both types of code families realizing addressable logical $\cz$ gates. Since $\alpha=\beta=1$ and the $o(1)$ terms in the exponents of both the distance and logical dimension are equal to zero, for linearly growing target support sizes and sub-linearly growing target control-qubit sizes, choosing any constant $0 < \eta \le 1-\epsilon'$ in Construction~3 of Section~\ref{gen_family_framework} yields an explicit code family with parameters $[[\Theta(n_i'), \Theta(n_i'), \Omega({(n_i')}^{\eta})]]$ that transversally realizes the addressable logical $C_{A_i}(Z)$ gates.

On the other hand, if either the target support sizes grow super-linearly or the target control-qubit sizes grow at least linearly, we apply Construction~2 of Section~\ref{gen_family_framework}, as detailed in Remark~\ref{rmk:linear}. Governed by the scaling factor $\zeta=\max\{\eta+\epsilon', \epsilon\}$, we obtain explicit code family that transversally realizes the targeted addressable logical $C_{A_i}(Z)$ gates. For any $\eta > 0$, the parameters of this code family scale as $[[\Theta(n_i), \Omega({n_i}^{1/\zeta}), \Omega({(n_i)}^{\min\{1,\eta\}/\zeta})]]$. 

To realize all addressable logical $\cz$ gates, we apply the framework from Section~\ref{css_families_all_gates}. Because $\alpha=\beta=1$ and the $o(1)$ terms vanish, for any $\eta > 0$, we construct a code family with parameters $[[\Theta(n_i), \Omega({(n_i)}^{1/(2+\eta)}), \Omega({(n_i)}^{\min\{1,\eta\}/(2+\eta)})]]$ that transversally realizes any addressable logical $\cz$ gate.

\subsubsection{Addressable Targeted Sequence of Logical $\czrot$ gates $(\ell \ge 1)$ for Sub-linear Target Support Sizes}
We utilize Families~2 and~3 realizing logical transversal $R_Z\left( \frac{\pi}{2^{\ell+1}}\right)$ to construct these code families. For any fixed integer $\ell \ge 1$, since $0 < \epsilon, \epsilon', \alpha_{1,\ell+1}, \alpha_{2,\ell+1} < 1$, we apply Construction~1 for sub-linear target support sizes established in Section~\ref{gen_family_framework}. Specifically, for any constants $0 < \alpha_{1,\ell+1}, \alpha_{2,\ell+1} < H_2(\frac{1}{\ell+1})$, this framework yields two classes of code families that transversally realize the addressable logical $C_{A_i}\left(\sqzrot\right)$ gates:

First, by choosing any constant $0 < \eta < \min\left\{1-\epsilon', \frac{1-\epsilon}{1-\alpha_{1,\ell+1}}\right\}$, we obtain an explicit code family with parameters $[[\Theta(n_i'),\Theta(n_i'), \Omega({(n_i')}^{\eta \beta_{1,\ell+1}(\alpha_{1,\ell+1}) +o(1) })]]$.

Second, by choosing any constant $0 < \eta < \min\left\{1-\epsilon', \frac{1-\epsilon}{1-\alpha_{2,\ell+1}}\right\}$, we obtain a non-explicit code family with parameters $[[\Theta(n_i'),\Theta(n_i'), \Omega({(n_i')}^{\eta \beta_{2,\ell+1}(\alpha_{2,\ell+1}) +o(1) })]]$, yielding improved distance scaling at the expense of explicitness.

\subsubsection{Addressable Targeted Sequence of Logical $\czrot$ gates $(\ell \ge 2)$ for At Least Linear Target Support Sizes}
We utilize Families~2 and~3 to construct these code families. For any fixed integer $\ell \ge 1$, since $\epsilon \ge 1$ in this regime, we apply Construction~2 of Section~\ref{gen_family_framework}. Specifically, for any constants $0 < \alpha_{1,\ell+1}, \alpha_{2,\ell+1} < H_2\left(\frac{1}{\ell+1}\right)$, this framework yields two classes of code families that transversally realize the addressable logical $C_{A_i}\left(\sqzrot\right)$ gates.

First, utilizing Family~2, we obtain an explicit code family governed by the scaling factor $\zeta_{1} = \max\{\eta+\epsilon', \eta (1-\alpha_{1,\ell+1}) + \epsilon\}$. For any $\eta >0$, the family achieves parameters $[[\Theta(n_i),\Omega({(n_i)}^{1/\zeta_1 + o(1)}), \Omega({(n_i)}^{ \min\{1, \eta \beta_{1,\ell+1}(\alpha_{1,\ell+1})\}/\zeta_1 +o(1)})]]$.

Second, utilizing Family~3, we obtain a non-explicit code family governed by the scaling factor $\zeta_2 = \max\{\eta+\epsilon', \eta (1-\alpha_{2,\ell+1}) + \epsilon\}$. For any $\eta > 0$, the family achieves parameters $[[\Theta(n_i),\Omega({(n_i)}^{1/\zeta_2 + o(1)}), \Omega({(n_i)}^{(\min\{1,\eta \beta_{2,\ell+1}(\alpha_{2,\ell+1})\}/\zeta_2 +o(1)})]]$.

\subsubsection{Addressable Targeted Sequence of Logical $\cs$ Gates for At Least Linear Target Support Sizes}

The analysis closely parallels that of Section~\ref{logt_sep_cae}.
We utilize Family~4 to construct these code families. Because $\epsilon=\alpha_3=1$ and the $o(1)$ term in the dimension exponent vanishes, if $\epsilon=1$ and $\epsilon' < 1$, we apply Construction~3 of Section~\ref{gen_family_framework} to obtain an explicit code family that transversally realizes the addressable logical $C_{A_i}(S)$ gates with parameters $[[\Theta(n_i'), \Theta(n_i'), \Omega\left(2^{(c_{R_3}+o(1)) \sqrt{\eta \log(n_i')}}\right)]]$ for any constant $0 < \eta \le 1-\epsilon'$.

If either $\epsilon > 1$ or $\epsilon' \ge 1$, then, following Remark~\ref{rmk:linear}, we apply Construction~2 of Section~\ref{gen_family_framework} to obtain an explicit code family governed by the scaling factor $\zeta = \max\{\eta + \epsilon', \epsilon\}$. The code parameters of this family scale as $[[\Theta(n_i), \Omega(n_i^{1/\zeta}), \Omega\left(2^{(c_{R_3}/\zeta+o(1)) \sqrt{\frac{\eta}{\zeta} \log(n_i)}}\right)]]$. Consequently, for targeted logical $\cs$ gates, this construction yields improved logical dimension scaling at the expense of sub-polynomial minimum distance scaling.

\subsubsection{All Addressable Logical $\czrot$ Gates ($\ell \ge 1$)}

We utilize Families~2 and~3 to construct these code families. For any fixed integer $\ell \ge 1$, we apply the framework detailed in Section~\ref{css_families_all_gates}. Specifically, for any constants $0 < \alpha_{1,\ell+1}, \alpha_{2,\ell+1} < H_2\left(\frac{1}{\ell+1}\right)$ and a scaling factor $\zeta = 2 + \eta$, this framework yields two code families that transversally realize any addressable logical $\czrot$ gate:

First, utilizing Family~2, we obtain an explicit code family. For any $\eta >0$, the family achieves parameters $[[\Theta(n_i),\Omega({(n_i)}^{1/\zeta + o(1)}), \Omega({(n_i)}^{\min\{1,\eta \beta_{1,\ell+1}(\alpha_{1,\ell+1})\}/\zeta +o(1)})]]$. 

Second, utilizing Family~3, we obtain a non-explicit code family. For any $\eta >0$, the family achieves parameters $[[\Theta(n_i),\Omega({(n_i)}^{1/\zeta + o(1)}), \Omega({(n_i)}^{\min\{1,\eta \beta_{2,\ell+1}(\alpha_{2,\ell+1})\}/\zeta +o(1)})]]$.

The constructed CSS families realizing addressable logical $1$-controlled $Z$-rotations are summarized in the following table:
\renewcommand{\arraystretch}{1}
\begin{table*}[!h]
\begin{center}
\begin{tabular}{| p{8.3cm} | p{7.4cm} |}
    \hline
\textbf{Families} & \textbf{Logical gates}  \\
 \hline
\makecell[l]{Explicit: $[[n, \Theta(n), \Omega(n^{\eta})]]$, \\ for $0 < \eta \le 1-\epsilon'$}  & \makecell[l]{Targeted addressable logical $\cz_{A_i}$ gates, \\
for $\epsilon \in (0,1]$ and $\epsilon' \in (0,1)$} \\
\hline
\makecell[l]{Explicit: $[[n, \Omega(n^{1/\zeta}), \Omega(n^{\min\{1,\eta\}/\zeta})]]$, \\ for $\eta> 0$ and $\zeta=\max\{\eta+\epsilon', \epsilon\}$}  & \makecell[l]{Targeted addressable logical $\cz_{A_i}$ gates, \\ for $\epsilon \in (0,2]$ and $\epsilon' \in (0,1]$} \\
\hline
\makecell[l]{Explicit: $[[n, \Omega(n^{1/(2+\eta)}), \Omega(n^{\min\{1,\eta\}/(2+\eta)})]]$, \\ for $\eta >$0} & \makecell[l]{Any addressable logical $\cz$ gate} \\
\hline
\makecell[l]{ $\ell \ge 1$:\\
$[[n, \Theta(n), \Omega(n^{\eta \beta_{i,\ell+1}(\alpha_{i,\ell+1}) + o(1)})]]$, \\ for $0 < \eta < \min\{1, \frac{1-\epsilon}{1-\alpha_{i,\ell+1}}\}$, \\ $0 < \alpha_{i,\ell+1} < H_2\left(\frac{1}{\ell+1}\right)$, \\ explicit for $i=1$, non-explicit for $i=2$}  & \makecell[l]{Targeted addressable logical $C_{A_i}\left(\sqzrot\right)$ \\ gates, for $\epsilon, \epsilon' \in (0,1)$} \\
\hline
\makecell[l]{ $\ell \ge 1$:\\
$[[n, \Omega(n^{1/\zeta_i+o(1)}), \Omega(n^{\min\left\{1,\eta \beta_{i,\ell+1}(\alpha_{i,\ell+1}) \right\}/\zeta_i + o(1)})]]$, \\ $\zeta_i=\max\{\eta+\epsilon', \eta (1-\alpha_{i,\ell+1}) + \epsilon\}$, \\ for $\eta > 0$, $0 < \alpha_{i,\ell+1} < H_2\left(\frac{1}{\ell+1}\right)$, \\ explicit for $i=1$, non-explicit for $i=2$}  & \makecell[l]{Targeted addressable logical $C_{A_i}\left(\sqzrot\right)$ \\ gates, for $\epsilon \in (0,2]$ and $\epsilon' \in (0,1]$} \\
\hline
\makecell[l]{ Explicit:\\
$[[n, \Theta(n), \Omega(2^{(c_{R_3}+o(1)) \sqrt{\eta \log(n)}})]]$, \\ for $0 < \eta \le 1-\epsilon'$, $R_3 \in (0,1/3)$}  & \makecell[l]{Targeted addressable logical $C_{A_i}(S)$ gates, \\ for $\epsilon \in (0,1]$ and $\epsilon' \in (0,1)$} \\
\hline
\makecell[l]{ Explicit:\\
$[[n, \Omega(n^{1/\zeta}), \Omega(2^{((c_{R_3}/\zeta)+o(1)) \sqrt{(\eta/\zeta) \log(n)}})]]$, \\ $\zeta=\max\{\eta+\epsilon', \epsilon\}$, for $\eta > 0$, $R_3 \in (0,1/3)$}  & \makecell[l]{Targeted addressable logical $C_{A_i}(S)$ gates, \\ for $\epsilon \in (0,2]$ and $\epsilon' \in (0,1]$} \\
\hline
\makecell[l]{ $\ell \ge 1$:\\
$[[n, \Omega(n^{1/\zeta+o(1)}), \Omega(n^{\min\left\{1,\eta \beta_{i,\ell+1}(\alpha_{i,\ell+1}) \right\}/\zeta + o(1)})]]$, \\ $\zeta=2+\eta$, for $\eta > 0$, $0 < \alpha_{i,\ell+1} < H_2\left(\frac{1}{\ell+1}\right)$, \\ explicit for $i=1$, non-explicit for $i=2$}  & \makecell[l]{Any addressable logical $\czrot$ gate} \\
\hline
\end{tabular}
\vspace{0.5cm}
\caption{CSS families realizing addressable logical $1$-controlled $Z$-rotations}\label{cz_families}
\end{center}
\end{table*}

We now highlight main contributions of this subsection and Appendix~\ref{css_families_mcz}. Recent works~\cite{HVAR20251,HVAR2025,G2025} achieve asymptotically good parameters for addressable logical $\ccz$ gate, but they rely on transversal physical $\ccz$ gates applied across multiple code blocks, thereby realizing addressable logical inter-block $\ccz$ gates. To realize intra-block logical $\ccz$ gates within those frameworks, one must resort to non-transversal constant-depth physical circuits consisting of $\ccz, \cz$ and $Z$ gates. At the expense of asymptotically good parameter scaling, we construct code families that achieve a polynomial logical dimension and minimum distance, and which realize addressable logical intra-block multi-controlled $Z$-rotations entirely via transversal physical single-qubit $Z$-rotations.

Furthermore, the flexibility of the appending framework naturally extends beyond individual gate types to accommodate a fixed, constant number of addressable logical gates within a single code family. For instance, suppose we wish to construct a code family that realizes both addressable logical $R_Z\left(\frac{\pi}{8}\right)$ and $\cs$ gates. By utilizing the respective auxiliary codes for each gate type in the appending construction, we obtain a single CSS code family that achieves a polynomial logical dimension and minimum distance while transversally realizing both addressable logical $R_Z\left(\frac{\pi}{8}\right)$ and $\cs$ gates. More generally, given any constant number of desired logical single-qubit and multi-controlled $Z$-rotations (see Appendix~\ref{css_families_mcz}), this framework can construct a CSS code family that achieves polynomial parameter scaling while ensuring all addressable target logical gates are realized transversally.

\bigskip
We remark that the CSS families constructed in this section are not low-density parity-check (LDPC), as the weights of all non-trivial stabilizers in CSS codes derived from punctured classical divisible codes scale with the blocklength. If both the primary and auxiliary codes possess LDPC $X$-stabilizer generators, our construction yields a derived CSS family with an LDPC $X$-stabilizer check matrix. However, because we guarantee the $Z$-distance by bounding the minimum weight of the entire logical-$Z$ space (which includes all non-trivial $Z$-stabilizers), we must sacrifice the LDPC property for the $Z$-stabilizer check matrix. Indeed, constructing LDPC auxiliary codes that transversally realize the logical transversal $\sqzrot$ is an open problem. 

\section{Conclusion}

To summarize, in this paper, we characterized CSS codes that realize target logical diagonal gates via transversal physical $Z$-rotations. We then introduced the appending framework, which extends a primary CSS code using appending matrices to transversally realize an arbitrary number of such target gates. Because CSS codes can only realize logical single-qubit $Z$-rotations and multi-controlled $Z$-rotations, we detailed appending-matrix constructions for both these addressable cases. Finally, we constructed CSS code families that realize addressable logical single-qubit and multi-controlled $Z$-rotations. 

We lay out some directions for future work here. As noted in Section~\ref{sec:introduction}, the construction of CSS codes that transversally realize the full logical Clifford group relies on an inefficient combination of $\ell$ independent code blocks, where each block individually realizes this group. An important direction for future work then is to construct single-block codes with shorter physical blocklengths that achieve this. Moreover, improving the asymptotic parameters of such families remains an open problem---specifically, determining whether there exist such families whose rate or minimum distance scales better than the square root of the physical blocklength.

While the appending framework can be utilized to flexibly realize any desired set of logical gates, the resulting blocklengths are generally not minimal because each appended support is exclusively dedicated to a specific target gate. To achieve shorter blocklengths, one could bypass appending matrices and directly construct CSS codes that satisfy the characterization conditions, either analytically or via computational search. However, satisfying these simultaneous constraints for multiple target gates within a tight blocklength is inherently difficult. Consequently, length-optimized codes typically sacrifice flexibility and only support a restricted set of target gates.

To maintain a robust $Z$-distance in the appending framework, the minimum weight of the entire logical-$Z$ space, including all non-trivial $Z$-stabilizers, must remain high. While this requirement is acceptable for the shorter blocklengths used in practical quantum computing, it constitutes the primary obstacle to constructing LDPC CSS families that realize target logical diagonal gates~\cite{GK2026}. Therefore, a direction for future research is developing alternative constructions that circumvent the need to bound the weight of all logical-$Z$ operators. 

While the appending framework provides a flexible methodology for constructing code families, we acknowledge that it may not yield the best asymptotic parameter scaling, as seen in the case of addressable logical single-qubit $Z$-rotations (see~\cite{SJ2026, LWC2026,AM2026}). Nevertheless, its ability to transversally realize addressable logical intra-block multi-controlled $Z$-rotations constitutes a novel capability in the literature. Optimizing the parameters of these families realizing multi-controlled $Z$-rotations through alternative techniques therefore remains an open problem for future research.

Finally, no existing CSS codes in the literature utilize the relaxed characterizations derived in Sections~\ref{sec:sq_relaxed} and~\ref{sec:mc_relaxed}. Exploiting these conditions could yield improved codes for realizing target logical diagonal gates. Currently, we only leverage these relaxations for $S$ and $\cz$ gates, where punctured self-orthogonal codes can directly replace punctured doubly-even codes. 

\appendix

\section{Addressable Logical Multi-Controlled $Z$-Rotations in CSS Codes}\label{ft_mc_zrot}

For integers $m \geq 2$ and $\ell \ge 0$, we provide an alternative characterization of CSS codes that transversally support the addressable logical gate $\addmczrot$. This characterization motivates our appending-matrix construction for multi-controlled $Z$-rotations. While Proposition~\ref{prop:addressable_multi_controlled} permits any $p \ge \ell + m-1$, we restrict our focus to $p = \ell + m-1$.

Analogous to Lemma~\ref{lem:schur_basis_vectors}, the following technical lemma reduces modular constraints on the entire code to equivalent conditions on its basis vectors.
\begin{lemma}\label{lem:schur_basis_vectors_for_mc}    
    Let $C \subseteq \mathbb{F}_{2}^n$ be a binary code with basis $\{x_1, \ldots, x_k\}$, and let $A \in \mathbb{F}_2^{k \times \cdots \times k}$ be a fixed order-$m$ tensor. For any $a \in \mathbb{F}_2^k$, define the corresponding codeword $x_a = \bigoplus_{i=1}^k a(i) x_i$. For any integer vector $w \in \mathbb{Z}^n$, the modular condition
    \begin{equation}\label{eq:orig_mod_code_mc}
        w \cdot x_a = 2^{m-1} \sum_{1 \leq i_1 < \cdots < i_m \leq k} A(i_1, \ldots, i_m) a(i_1) \cdots a(i_m) \pmod{2^{\ell+m}}    
    \end{equation}    
    holds for all $a \in \mathbb{F}_2^k$ if and only if for all ordered indices $1 \le i_1 < \cdots < i_{j} \le k$ with $j \in [\ell+m]$ , the following modular equation is satisfied:
    \begin{equation} 
    w \cdot \left( x_{i_1} \ast \cdots \ast x_{i_{j}}\right) = \begin{cases} 
        {(-1)}^{m-1} A(i_1, \ldots,i_m) \pmod{2^{\ell+1}}, & \text{if } j=m, \\ 
        0 \pmod{2^{\ell+m-j+1}}, & \text{otherwise.} 
    \end{cases} \label{eq:basis_mc}
    \end{equation}
\end{lemma}
The proof of this lemma closely follows that of Lemma~\ref{lem:schur_basis_vectors} and is provided in Subsection~\ref{sec:technical_decomp_mod_mc}. 

Applying the Lemma to Proposition~\ref{prop:addressable_multi_controlled} yields the following proposition.
\begin{proposition}\label{prop:target_addr_log_mc_basis}        
    Let ${(C_1, C_2)}_{\tn{CSS}}$ be an $[[n,k]]$ CSS code, where $\{y_1, \ldots, y_k\}$ is a basis of $C_1/C_2$. For any order-$m$ target addressability tensor $A \in \mathbb{F}_2^{k \times \cdots \times k}$, the code realizes the addressable logical gate $\addmczrot$ via a transversal physical $Z$-rotation $U(\ell+m-1,w)$ if and only if for all $x \in C_2$, $a \in \mathbb{F}_2^k$, and ordered indices $1 \le i_1 < \cdots < i_{j} \le k$ with $j \in [\ell+m]$, the following modular equations hold:
    \begin{align*}
        &w \cdot x = 0 \pmod{2^{\ell+m}},  \qquad
        w \cdot {(x \ast y_a)} = 0 \pmod{2^{\ell+m-1}}, \\
        &w \cdot \left( y_{i_1} \ast \cdots \ast y_{i_{j}}\right) = \begin{cases} 
            {(-1)}^{m-1} A(i_1, \ldots,i_m) \pmod{2^{\ell+1}}, & \text{if } j=m, \\ 
            0 \pmod{2^{\ell+m-j+1}}, & \text{otherwise.} 
        \end{cases}
    \end{align*}
\end{proposition}

\subsection{Appending Matrices for a Targeted Addressable Logical Gate $\addmczrot$}\label{app_targ_addr_log_mc}

Our construction of appending matrices for logical $(m-1)$-controlled $Z$-rotations proceeds inductively. We assume appending matrices have already been derived for all logical single-qubit $Z$-rotations and $(\widetilde{m}-1)$-controlled $Z$-rotations, where $\widetilde{m} < m$. The single-qubit constructions established in Subsection~\ref{app_targ_addr_log_sq} serve as the base case.

Given a primary $[[n',k']]$ CSS code $(C_1', C_2')_{\tn{CSS}}$ and an order-$m$ target address tensor $A \in \mathbb{F}_2^{k' \times \cdots \times k'}$, we employ the appending framework from Section~\ref{app_framework} to construct an extended code that transversally realizes the addressable logical gate $\addmczrot$. We construct the required appending matrices using the auxiliary $[[\widetilde{n}, \widetilde{k}, \widetilde{d}]]$ CSS code $(\widetilde{C}_1, \widetilde{C}_2)_{\tn{CSS}}$ from Section~\ref{app_targ_addr_log_sq}, in which the physical transversal $U(\ell+m-1,\widetilde{w})$ realizes the logical transversal $R_Z\left( \frac{\pi}{2^{\ell+m-1}} \right)$, and $\widetilde{d}_Z \ge d_{\min}(\widetilde{C}_2^\perp) \ge \widetilde{d}$.  

Recall from Section~\ref{app_targ_addr_log_sq} that if $\widetilde{y}_1, \ldots, \widetilde{y}_{\widetilde{k}}$ are the rows of the coset generator matrix $G_{\widetilde{C}_1/\widetilde{C}_2}$, the following hold for all $\widetilde{x} \in \widetilde{C}_2$, $a \in \mathbb{F}_2^{\widetilde{k}}$, $i \in [\widetilde{k}]$, and ordered indices $1 \leq i_1 < \cdots < i_{j} \le \widetilde{k}$ with $j \in \{2, \ldots, \ell+m\}$:
\begin{alignat}{2}
    \widetilde{w} \cdot \widetilde{x} &= 0 \pmod{2^{\ell+m}},
    &\qquad \widetilde{w} \cdot (\widetilde{x} \ast \widetilde{y}_a) &= 0 \pmod{2^{\ell+m-1}}, \nonumber \\
    \widetilde{w} \cdot \widetilde{y}_i &= 1 \pmod{2^{\ell+m}}, &\qquad 
    \widetilde{w} \cdot (\widetilde{y}_{i_1} \ast \cdots \ast \widetilde{y}_{i_j}) &= 0 \pmod{2^{\ell+m-j+1}}.
    \label{eq:secondary_css_basis_mod_mc}
\end{alignat}

The following proposition, which is crucial for our subsequent construction, analyzes the Schur product of $\bigoplus_{i=1}^{\widetilde{k}} \widetilde{y}_i$ with the rows of the coset generator matrix $G_{\widetilde{C}_1/\widetilde{C}_2}$.
\begin{proposition}\label{prop:schur_sum_basis_mc_secc_ode}
    Let $\widetilde{y}_1, \ldots, \widetilde{y}_{\widetilde{k}}$ denote the rows of $G_{\widetilde{C}_1/\widetilde{C}_2}$. The following modular conditions hold:
    \begin{enumerate}
        \item 
        \begin{equation}\label{eq:sum_1}
            \widetilde{w} \cdot \left( \bigoplus_{i=1}^{\widetilde{k}} \widetilde{y}_i \right) = \widetilde{k} \pmod{2^{\ell+m}}.
        \end{equation}
        \item For any index $p \in [\widetilde{k}]$,
        \begin{equation}
            \widetilde{w} \cdot \left( \left(\bigoplus_{i=1}^{\widetilde{k}} \widetilde{y}_i \right) \ast \widetilde{y}_p \right) = 1 \pmod{2^{\ell+m-1}}.\label{eq:schur_prop_sc_code_mc_1}
        \end{equation}
        \item For any ordered sequence of indices $1 \le p_1 < \cdots < p_q \le \widetilde{k}$ with $q \ge 2$,
        \begin{equation}
            \widetilde{w} \cdot \left( \left(\bigoplus_{i=1}^{\widetilde{k}} \widetilde{y}_i \right) \ast \left( \widetilde{y}_{p_1} \ast \cdots \ast \widetilde{y}_{p_q} \right) \right) = 0 \pmod{2^{\ell+m-q}}.\label{eq:schur_prop_sc_code_mc_2}
        \end{equation}
    \end{enumerate}
\end{proposition}
The proof of this proposition is provided in Subsection~\ref{sec:schur_sum_basis_mc_secc_ode}.

\subsubsection{Reduction to a simpler special case}
We first decompose the target gate $\addmczrot$ by lexicographically ordering its control qubit tuples. Let $(1,2, \ldots, m-1) \leq \left(c_1^{(1)}, \ldots, c_{m-1}^{(1)}\right) < \cdots < \left(c_1^{(\kappa)}, \ldots, c_{m-1}^{(\kappa)}\right) \le (k'-m+1,\ldots,k'-1)$ denote all the active control qubit tuples---that is, tuples $(i_1, \ldots, i_{m-1})$ for which there exists at least one target qubit $t \in \{i_{m-1}+1, \ldots, k'\}$ such that $A(i_1, \ldots, i_{m-1}, t)=1$. We factor $\addmczrot$ into $\kappa$ addressed gates
$C_{A_1}^{(m-1)}(\sqzrot), \ldots, C_{A_\kappa}^{(m-1)}(\sqzrot)$. Each tensor $A_j$ isolates the interactions controlled by the qubit tuple $\left(c_1^{(j)}, \ldots, c_{m-1}^{(j)}\right)$ by inheriting the following tuple slice of $A$ and vanishing elsewhere:
\begin{equation*}
    A_j\left(i_1, \ldots, i_{m-1},i_m\right):=
    \begin{cases}
        A\left(c_1^{(j)}, \ldots, c_{m-1}^{(j)},i_m\right), & \text{if } (i_1, \ldots, i_{m-1})=(c_1^{(j)}, \ldots, c_{m-1}^{(j)}) \\ &\quad \quad \quad \text{ and } i_m \in \{c_{m-1}^{(j)}+1, \ldots, k'\}, \\
        0, & \text{otherwise.}
    \end{cases}
\end{equation*}

If the appending-matrix construction is established for each factor $C_{A_j}^{(m-1)}\left(\sqzrot\right)$, the appending framework guarantees that each gate $C_{A_j}^{(m-1)}\left(\sqzrot\right)$ is realized by some transversal $Z$-rotation $U(\ell+m-1,w_j)$ in the final extended code. The composition $\prod_{j=1}^\kappa U(\ell+m-1,w_j)$ then realizes the full target gate $\addmczrot$. Consequently, without loss of generality, we assume that the only active control qubit tuple is $(1,\ldots,m-1)$---meaning, the logical $C^{(m-1)}\left( R_Z\left( \frac{\pi}{2^\ell} \right) \right)$ gates act exclusively with the first qubit tuple $(1,2,\ldots,m-1)$ as the control qubits and a subset of the remaining qubits as targets.

Furthermore, we restrict the construction to target matrices satisfying $\lvert \supp(A(1,\ldots, m-1,m:k')) \rvert \le \widetilde{k}$, where 
\begin{equation*}
    A(1, \ldots,m-1, m:k'):=(A(1, \ldots,m-1, m),A(1, \ldots,m-1, m+1), \ldots, A(1, \ldots,m-1, k')).
\end{equation*}
If the target support exceeds $\widetilde{k}$, the flexibility of the appending framework allows us to partition the addressable gate $\addczrot$ into $g=\lceil \lvert \supp(A(1,\ldots,m-1,m:k')) \rvert / \widetilde{k} \rceil$ factors, each acting from the control qubit tuple $(1, \ldots, m-1)$ onto at most $\widetilde{k}$ targets in $\{m, \ldots, k'\}$. Constructing appending matrices independently for each factor and composing their corresponding physical $Z$-rotations realizes the target gate $\addmczrot$ within the final extended code. 

\subsubsection{The appending-matrix construction in the special case of $|\supp(A(1,\ldots, m-1,m:k'))| \ \le \widetilde{k}$} 

To simplify the presentation of the construction, we further assume that $A(i_1, \ldots, i_{m-1},i_m)=1$ if and only if $(i_1, \ldots, i_{m-1})=(1,\ldots,m-1)$ and $i_m \in \{m, m+1, \ldots, \mathsf{s}+m-1 \}$, where $\mathsf{s}:=\lvert \supp(A(1,\ldots, m-1,m:k')) \rvert \le \widetilde{k}$. This effectively designates $\{m, \ldots, \mathsf{s}+m-1\}$ as the set of all target qubits. Again, recall that we are given a primary $[[n',k']]$ CSS code $(C_1', C_2')_{\tn{CSS}}$, and that we will obtain appending matrices $G_1$ and $G_2$ using an auxiliary $[[\widetilde{n}, \widetilde{k}, \widetilde{d}]]$ CSS code $(\widetilde{C}_1, \widetilde{C}_2)_{\tn{CSS}}$ for which the modular equations in~\eqref{eq:secondary_css_basis_mod_mc} hold. We define the rows $y_1, \ldots, y_{k'}$ of $G_1$ as follows: we set $y_j = \bigoplus_{p=1}^{\widetilde{k}} \widetilde{y}_p$ for the control qubit rows $j \in [m-1]$, assign $y_j = \widetilde{y}_{j-m+1}$ for target qubits $m \le j \le \mathsf{s}+m-1$, and set $y_j=0$ for all remaining rows $\mathsf{s}+m-1 < j \le k'$. 

Next, we set $G_2 = G_{\widetilde{C}_2}$, and let ${(\widehat{C}_1, \widehat{C}_2)}_{\tn{CSS}}$ denote the CSS code obtained by applying the appending-matrix construction to $(C_1', C_2')_{\tn{CSS}}$. 

\subsubsection{$U(\ell+m-1,w)$ is a logical operator for $(\widehat{C}_1, \widehat{C}_2)_{\tn{CSS}}$}
Consider the transversal $Z$-rotation $U(\ell+m-1, w)$ with $w=(\mathbf{0}_{n'}, {(-1)}^{m-1} \widetilde{w})$. By construction, $G_2 = G_{\widetilde{C}_2}$ and the rows of $G_1$ are drawn from $\widetilde{C}_1/\widetilde{C}_2$. Consequently, for any $x=(x',\widetilde{x}) \in C_2$ and $y_a=(y', \widetilde{y}) \in C_1/ C_2$ (with $a \in \mathbb{F}_2^{k'}$), we have $\widetilde{x} \in \widetilde{C}_2$ and $\widetilde{y} \in \widetilde{C}_1/ \widetilde{C}_2$. By~\eqref{eq:secondary_css_basis_mod_mc}, this guarantees
\begin{equation*}
    w \cdot x = 0 \pmod{2^{\ell+m}} \quad \text{ and } \quad w \cdot (x \ast y_a) = 0 \pmod{2^{\ell + m -1}}.
\end{equation*}
Remark~\ref{rmk:logical_op} then establishes $U(\ell+m-1,w)$ as a logical operator on $(\widehat{C}_1, \widehat{C}_2)_{\tn{CSS}}$. 

\subsubsection{$U(\ell+m-1,w)$ realizes $\addmczrot$ up to a logical phase}
For any logical state $\ket{a}$ with $a \in \mathbb{F}_2^{k'}$, the action of $U(\ell+m-1,w)$ is given by
\begin{align*}
    U&(\ell+m-1,w)\ket{a}_L = \exp{\iota \frac{\pi}{2^{\ell+m-1}} \left(w \cdot \left( \bigoplus_{i=1}^{k'} a(i) y_i\right) \right)} \ket{a}_L \\
    &= \exp{\iota \frac{\pi}{2^{\ell+m-1}} \left( \sum_{i=1}^{k'} {(-2)}^{i-1} \left( \sum_{1 \leq j_1 < \cdots < j_i \leq k'} \left( \prod_{p=1}^i a({j_p}) \right) w \cdot \left( y_{j_1} \ast \cdots \ast y_{j_i}\right) \right)\right)} \ket{a}_L.
\end{align*}

We next show that the $t$-fold Schur products for $t > m$ vanish in the phase exponent, while the $m$-fold products exactly realize the target gate $\addmczrot$. The $t$-fold products for $t < m$, however, introduce an extraneous logical phase, which we will eliminate in the next subsection. 

To see why the $t$-fold Schur products for $t > m$ vanish, first consider any ordered sequence of $t>m$ target qubit indices $m \le p_1 < \cdots < p_t \le \mathsf{s}+m-1$. Condition~\eqref{eq:secondary_css_basis_mod_mc} implies
\begin{equation*}
    w \cdot (y_{p_1} \ast \cdots \ast y_{p_t}) = {(-1)}^{m-1} \widetilde{w} \cdot (\widetilde{y}_{p_1-m+1} \ast \cdots \ast \widetilde{y}_{p_t-m+1}) =  0 \pmod{2^{\ell+m-t+1}}.
\end{equation*}
On the other hand, if the Schur product combines the control qubit indices $1 \le p_1 < \cdots < p_q \le m-1$ with $q \ge 1$ and target qubit indices $m \le r_1 < \cdots < r_t \le \mathsf{s}+m-1$ with $t \ge m$ (and $q+t > m$),~\eqref{eq:schur_prop_sc_code_mc_2} yields
\begin{align*}
    w \cdot (y_{p_1} \ast \cdots \ast y_{p_q} \ast y_{r_1} \ast \cdots \ast y_{r_t}) &= {(-1)}^{m-1} \widetilde{w} \cdot \left( \left( \bigoplus_{i=1}^{\widetilde{k}} \widetilde{y}_i \right) \ast \widetilde{y}_{r_1-m+1} \ast \cdots \ast \widetilde{y}_{r_t-m+1} \right) \\
    &= 0 \pmod{2^{\ell+m-t}} \\
    &= 0 \pmod{2^{\ell+m-q-t+1}}.
\end{align*}
Because only the first $\mathsf{s}+m-1$ rows of $G_1$ are non-zero, the preceding modular equations guarantee that all $t$-fold Schur products for $t > m$ vanish in the phase exponent. Consequently, the action of the physical rotation $U(\ell+m-1,w)$ reduces to
\begin{align*}
    U(&\ell+m-1,w)\ket{a}_L \\
    &= \exp{\iota \frac{\pi}{2^{\ell+m-1}} \left( \sum_{i=1}^{m} {(-2)}^{i-1} \left( \sum_{1 \leq j_1 < \cdots < j_i \leq k'} \left( \prod_{p=1}^i a({j_p}) \right) w \cdot \left( y_{j_1} \ast \cdots \ast y_{j_i}\right) \right)\right)} \ket{a}_L.
\end{align*}

We now evaluate the $m$-fold Schur products of non-zero rows of $G_1$ in the phase exponent. For any control qubit indices $1 \leq p_1 < \cdots < p_q \leq m-1$ and target qubit indices $m \le r_1 < \cdots < r_t \le \mathsf{s}+m-1$ with $q+t=m$. We distinguish three cases: $q=0$, $1 \le q < m-1$ and $q=m-1$. If $q=0$, then by~\eqref{eq:secondary_css_basis_mod_mc}, these $m$-fold Schur products vanish:
\begin{equation*}
    w \cdot (y_{r_1} \ast \cdots \ast y_{r_m}) = {(-1)}^{m-1} \widetilde{w} \cdot (\widetilde{y}_{r_1-m+1} \ast \cdots \ast \widetilde{y}_{r_m-m+1}) = 0 \pmod{2^{\ell+1}}.
\end{equation*}
If $1 \le q < m-1$, then $t \ge 2$. Consequently, by~\eqref{eq:schur_prop_sc_code_mc_2}, these products similarly vanish:
\begin{align*}
    w \cdot (y_{p_1} \ast \cdots \ast y_{p_q} \ast y_{{r_1}} \ast \cdots \ast y_{{r_t}}) &= {(-1)}^{m-1} \widetilde{w} \cdot \left( \left( \bigoplus_{i=1}^{\widetilde{k}} \widetilde{y}_i \right) \ast \widetilde{y}_{r_1-m+1} \ast \cdots \ast \widetilde{y}_{r_t-m+1} \right) \\
    &= 0 \pmod{2^{\ell+m-t}} \\
    &= 0 \pmod{2^{\ell+1}}.
\end{align*}
If $q=m-1$, all control qubits are included, and~\eqref{eq:schur_prop_sc_code_mc_1} yields
\begin{align*}
    w \cdot (y_{p_1} \ast \cdots \ast y_{p_{m-1}} \ast y_{{r_1}}) &= {(-1)}^{m-1} \widetilde{w} \cdot \left( \left( \bigoplus_{i=1}^{\widetilde{k}} \widetilde{y}_i \right) \ast \widetilde{y}_{r_1-m+1} \right) \\
    &= {(-1)}^{m-1} \pmod{2^{\ell+m-1}} = {(-1)}^{m-1} \pmod{2^{\ell+1}}.
\end{align*}
Substituting these evaluations into the $m$-fold product terms, and recalling that $A(i_1,\ldots,i_{m-1},i_m)=1$ if and only if $(i_1, \ldots, i_{m-1})=(1,\ldots,m-1)$ and $i_m \in \{m, m+1, \ldots, \mathsf{s}+m-1 \}$, we see that the logical action of the transversal physical rotation $U(\ell+m-1,w)$ further reduces to
\begin{align*}
    &U(\ell+m-1,w)\ket{a}_L \\
    &= \exp\Biggl\{\iota \frac{\pi}{2^{\ell+m-1}} \Biggl( \left( \sum_{i=1}^{m-1} {(-2)}^{i-1} \left( \sum_{1 \leq j_1 < \cdots < j_i \leq k'} \left( \prod_{p=1}^i a({j_p}) \right) w \cdot \left( y_{j_1} \ast \cdots \ast y_{j_i}\right) \right)\right) + \\
    &\qquad \qquad \qquad \qquad \qquad \qquad \qquad \qquad \left(2^{m-1}\sum_{1 \le j_1 < \cdots < j_m \le k'} A(j_1, \ldots, j_m) a(j_1) \cdots a(j_m) \right) \Biggr)\Biggr\} \ket{a}_L \\
    &=  \exp{\iota \frac{\pi}{2^{\ell+m-1}} \left( \sum_{i=1}^{m-1} {(-2)}^{i-1} \left( \sum_{1 \leq j_1 < \cdots < j_i \leq k'} \left( \prod_{p=1}^i a({j_p}) \right) w \cdot \left( y_{j_1} \ast \cdots \ast y_{j_i}\right) \right)\right)} \\
    &\qquad \qquad \qquad \qquad \qquad \qquad \qquad \qquad \qquad \qquad \qquad \qquad \qquad \qquad \qquad  \addmczrot \ket{a}_L.
\end{align*}
Thus, we have successfully realized the target gate $\addmczrot$ up to a logical phase. 

\subsubsection{Eliminating the residual logical phase}
The residual logical phase arises entirely from the $\widetilde{m}$-fold Schur products for $2 \le \widetilde{m} < m$, which correspond to addressable logical $(\widetilde{m}-1)$-controlled $Z$-rotations, and the linear terms, which correspond to addressable logical single-qubit $Z$-rotations. Let $\logicalU'$ denote the composite logical gate induced by this extraneous phase. By our inductive construction hypothesis on the number of control qubits, we can systematically construct appending matrices for both these $\widetilde{m}$-controlled and single-qubit $Z$-rotations to transversally realize the inverse gate, $(\logicalU')^\dag$. Incorporating the appending matrices for $(\logicalU')^\dag$ cancels the residual phase, yielding the final extended CSS code that transversally realizes the target gate $\addmczrot$.  

To make this cancellation concrete for the $t$-fold Schur products, we explicitly evaluate the residual logical phases for $t=1$ and $t=2$. For the linear terms ($t=1$), by~\eqref{eq:secondary_css_basis_mod_mc} and~\eqref{eq:sum_1}, we have $\widetilde{w} \cdot \widetilde{y}_i=1 \pmod{2^{\ell+m}}$ and $\widetilde{w} \cdot \left( \bigoplus_{i=1}^{\widetilde{k}} \widetilde{y}_i \right) = \widetilde{k} \pmod{2^{\ell+m}}$. Consequently, the linear terms induce the addressable single-qubit $Z$-rotations $\left(R_Z\left(\frac{\pi}{2^{\ell+m-1}}\right)_{A_1}^{(-1)^{m-1}}\right)^{\widetilde{k}}$ and $R_Z\left(\frac{\pi}{2^{\ell+m-1}}\right)_{A_2}^{(-1)^{m-1}}$, where $A_1 = (\mathbf{1}_{m-1}, \mathbf{0}_{k'-m+1})$ and $A_2 = (\mathbf{0}_{m-1}, \mathbf{1}_{\mathsf{s}}, \mathbf{0}_{k'-\mathsf{s}-m+1})$. 

We systematically eliminate the phase induced by these linear terms by appending auxiliary matrices realizing $R_Z\left(\frac{\pi}{2^{\ell+m-1}}\right)_{A_1}$ and $R_Z\left(\frac{\pi}{2^{\ell+m-1}}\right)_{A_2}$. Depending on the parity of $m$, we may specifically need to realize the inverses of this gate; crucially, Remark~\ref{rmk:inverse} guarantees that this inversion is straightforward: if $U(p,w)$ realizes a logical gate $\logicalU$, then $U(p,-w)$ realizes $\logicalU^\dag$. 

Similarly, evaluating the $2$-fold product terms using $\widetilde{w} \cdot \left( \left( \bigoplus_{i=1}^{\widetilde{k}} \widetilde{y}_i \right) \ast \widetilde{y}_j \right) = 1 \pmod{2^{\ell+m-1}}$ and $\widetilde{w} \cdot \left( \bigoplus_{i=1}^{\widetilde{k}} \widetilde{y}_i \right) = \widetilde{k} \pmod{2^{\ell+m-1}}$ identifies the residual addressable $1$-controlled $Z$-rotations as $\left(C_{A_1}\left(R_Z\left(\frac{\pi}{2^{\ell+m-2}}\right)\right)^{(-1)^m}\right)^{\widetilde{k}}$ and $C_{A_2}\left(R_Z\left(\frac{\pi}{2^{\ell+m-2}}\right)\right)^{(-1)^m}$, where the non-zero entries of $A_1$ and $A_2$ in $\mathbb{F}_2^{k' \times k'}$ are
\begin{equation*}
    A_1(i,j)=1 \text{ for } 1 \leq i < j \le m-1 \quad \text{ and } \quad A_2(i,j)=1 \text{ for } i \in [m-1], \ m \le j \le \mathsf{s}+m-1.
\end{equation*}
Appending auxiliary matrices that realize the inverse of these gates cancels the phase induced by the $2$-fold products. Iterating this cancellation procedure through all remaining $t$-fold products ($3 \le t < m$) completely eliminates the extraneous logical phases. This completes our construction, yielding the final extended CSS code $(C_1, C_2)_{\tn{CSS}}$ that transversally realizes the target gate $\addmczrot$. 

Suppose we require a total of $N_{m-1}$ appending blocks to realize $\addmczrot$ in the special case. From Sections~\ref{app_targ_addr_log_sq} and~\ref{app_targ_addr_log_cs}, we have established the base cases: $N_0=1$ for single-qubit $Z$-rotations, and $N_1=3$ for $1$-controlled $Z$-rotations (one block to realize $\addczrot$ and two to cancel the extraneous $(\sqzrot_{A_1}^\dag)^{\widetilde{k}}$ and $\sqzrot_{A_2}^\dag$). 

For $(m-1)$-controlled $Z$-rotations with $m \ge 3$, the required number of appending blocks $N_{m-1}$ is determined by the following recursive formula:
\begin{equation*}
    N_{m-1} = \left( 
    \sum_{j=0}^{m-2} \left(\binom{m-1}{j} + \binom{m-2}{j}\right) N_{j} \right) + 1.
\end{equation*}

Here, each term $\left( \binom{m-1}{j} + \binom{m-2}{j} \right) N_{j}$ counts the number of extraneous appending blocks required to cancel the logical phases produced by unwanted $j$-controlled $Z$-rotations. Specifically, the $\binom{m-1}{j}$ term arises when we select $j$ control qubits from the original $m-1$ control rows while retaining the original target row. The $\binom{m-2}{j}$ term arises when we select both the $j$ control qubits and the target qubit entirely from the original $m-1$ control rows. Note that in this formula, single-qubit $Z$-rotations are treated as $0$-controlled $Z$-rotations, corresponding to the $j=0$ term.

Therefore, the overall blocklength of the final derived CSS code in the general case is given by:
\begin{equation*}
    n' + \sum_{j=1}^{\kappa} \left( \left\lceil \frac{\lvert \supp(A(c_1^{(j)}, \ldots, c_{m-1}^{(j)}, c_{m-1}^{(j)}+1:k')) \rvert}{\widetilde{k}} \right\rceil N_{m-1} \widetilde{n}\right)
\end{equation*}
For any fixed gate $\addmczrot$, the required number of appending blocks in the special case is a constant, independent of the overall code parameters.
\bigskip

As in Section~\ref{ft_cs}, we can construct CSS codes that transversally realize any addressable logical $\mczrot$ gate. A CSS code can realize any addressable logical $\mczrot$ gate if and only if it can realize $\mczrot$ across any logical qubit tuple $(i_1, \ldots, i_m)$, where $1 \le i_1 < \cdots < i_m \le k'$. For the construction, we invoke the appending framework with $\numtargates=\binom{k'}{m}$, and for every qubit tuple, define a target logical gate $\custlogicalU{i_1, \ldots, i_m}=C_{i_1, \ldots, i_m}\left( R_Z\left( \frac{\pi}{2^\ell} \right)\right)$. This amounts to instantiating the above construction with the target address tensors $A_{i_1, \ldots, i_m}$ whose sole non-zero entry is $A_{i_1, \ldots, i_m}(i_1, \ldots, i_m)=1$.

\subsection{CSS Families Realizing Addressable Logical Multi-Controlled $Z$-Rotations}\label{css_families_mcz}

In this section, closely following the methodology of Section~\ref{css_families}, we generalize our constructions from $1$-controlled to multi-controlled $Z$-rotations, yielding CSS code families that realize addressable logical multi-controlled $Z$-rotations via transversal physical $Z$-rotations. 

\subsubsection{Systematic Construction of CSS Families via the Appending Framework}\label{gen_family_framework_mcz}

This section outlines a systematic methodology for constructing CSS code families that transversally realize addressable logical multi-controlled $Z$-rotations via the appending framework. We omit exhaustive technical details here, as the core ideas follow directly from the analysis established in Subsection~\ref{gen_family_framework}.

\emph{CSS Families Realizing a Sequence of Targeted Addressable Logical Gates}
Let $(A_i)_{i \in \mathbb{N}}$ be an arbitrary but fixed sequence of target address tensors, where $A_i \in \mathbb{F}_2^{k_i' \times \cdots \times k_i'}$ is an order-$m$ tensor for $(m-1)$-controlled $Z$-rotations. Our objective is to construct a CSS code family $((C_{1,i}, C_{2,i})_{\tn{CSS}})_{i \in \mathbb{N}}$ that encodes $k_i'$ logical qubits and transversally realizes the corresponding addressable logical gate, $C_{A_i}^{(m-1)}\left( \sqzrot\right)$.

Achieving this via the appending framework requires auxiliary CSS code families from Subsection~\ref{sec_families} capable of transversally realizing the logical transversal $R_Z \left( \frac{\pi}{2^{\ell+m-1}}\right)$. We consider the asymptotically good primary CSS code family $((C_{1,i}', C_{2,i}')_{\tn{CSS}})_{i \in \mathbb{N}}$ with parameters $[[n_i', k_i', d_i']]$ introduced in Subsection~\ref{target_families_framework}. 

\emph{Code Parameters of the Final Derived CSS Code:} For each $i \in \mathbb{N}$, we construct appending matrices for the primary code $(C_{1,i}', C_{2,i}')_{\tn{CSS}}$ using the auxiliary code $(\widetilde{C}_{1,r}, \widetilde{C}_{2,r})_{\tn{CSS}}$. By Lemma~\ref{lem:params}, the resulting CSS code $(C_{1,i,r}, C_{2,i,r})_{\tn{CSS}}$ achieves the parameters $[[n_{i,r}, k_{i,r}, d_{i,r}]]$, given by
\begin{align*}
    n_{i,r} = n_{i}' + \widetilde{n}_r M_{i,r}, \quad
    k_{i,r}=k_i', \quad
    d_{i,r} \ge \min\{d_i', \widetilde{d}_r\}, 
\end{align*}
where $M_{i,r}$ denotes the total number of auxiliary blocks appended to the primary code.

Letting $\kappa_i$ denote the number of active $(m-1)$ control qubit tuples, the value of $M_{i,r}$ can be bounded as:
\begin{align*}
    M_{i,r} &= N_{m-1} \sum_{\substack{1 \le j_1 < \cdots < j_{m-1} \le k'-1:\\A_i(j_1, \ldots, j_{m-1}, j_{m-1}+1:k_i') \ne 0}}  \left\lceil \frac{\lvert \supp(A_i(j_1, \ldots, j_{m-1}, j_{m-1}+1:k_i')) \rvert}{\widetilde{k}_r} \right\rceil \\
    &\le N_{m-1} \sum_{\substack{1 \le j_1 < \cdots < j_{m-1} \le k'-1:\\A_i(j_1, \ldots, j_{m-1}, j_{m-1}+1:k_i') \ne 0}} \left( \frac{\lvert \supp(A_i(j_1, \ldots, j_{m-1}, j_{m-1}+1:k_i')) \rvert}{\widetilde{k}_r} + 1 \right) \\
    &= N_{m-1} \left( \frac{\lvert \supp(A_i) \rvert}{\widetilde{k}_r} + \kappa_i \right). \\
\end{align*}

To control the scaling of the target support size relative to the primary code blocklength, we assume $ \lvert \supp(A_i) \rvert \le (n_i')^{\epsilon}$ for a constant $\epsilon \in (0,m]$. Additionally, we bound the number of active control qubit tuples by $\kappa_i \le (n_i')^{\epsilon'}$ for some constant $\epsilon' \in (0,m-1]$. Finally, following Section~\ref{gen_family_framework}, for any fixed constant $\eta > 0$, assume that for each $i \in \mathbb{N}$, there exists a suitable index $r_i \in \mathbb{N}$ such that $(n_i')^\eta = \Theta(\widetilde{n}_{r_i})$. Omitting the index $r_i$ from the subscript for brevity, the physical blocklength $n_i$ of $(C_{1,i}, C_{2,i})_{\tn{CSS}}$ scales as follows:
\begin{align*}
    n_{i} &= \mathcal{O}\left(n_i' + N_{m-1} (n_i')^{\eta(1 - \alpha) + \epsilon + o(1)} + N_{m-1}(n_i')^{\eta+\epsilon'}\right) = \mathcal{O}\left({(n_i')}^{\max\{1, \eta + \epsilon', \eta(1-\alpha)+\epsilon+ o(1)\}}\right), \\
    &\qquad \qquad k_i \ge n_i' R = \Omega(n_i'), \quad d_i \ge \min\{ n_i' \delta, (n_i')^{\eta \beta + o(1)}\} = \Omega\left( {(n_i')}^{\min\{1, \eta \beta + o(1)\}} \right).
\end{align*}
Recall that $N_{m-1}$ is a constant for a fixed controlled $Z$-rotation $\mczrot$, allowing it to be safely absorbed into the asymptotic notation.

We begin with the case when auxiliary code families exhibit sub-linear growth in their logical dimension (i.e., $0< \alpha < 1$).

\emph{Sub-linear Target Support Sizes $(0 < \epsilon < 1)$:}
Since $\epsilon < 1$, the number of active control qubit tuples also grow sub-linearly, implying $0 < \epsilon' < 1$. By choosing any constant $0 < \eta < \min\left\{1-\epsilon',\frac{1-\epsilon}{1-\alpha}\right\}$, the overall code parameters scale as $[[\Theta(n_i'), \Theta(n_i'), \Omega((n_i')^{\eta \beta+ o(1)})]]$ (as $\eta, \beta \le 1$). 

\emph{At Least Linear Target Support Sizes ($\epsilon \in [1,m]$):}
In this regime, the overall code parameters scale as $[[\Theta({(n_i')}^{\zeta + o(1)}), \Omega(n_i'), \Omega( {(n_i')}^{\min\{1, \eta \beta + o(1)\}} )]]$, where $\zeta=\max\{\eta + \epsilon', \eta(1-\alpha) + \epsilon\}$. The dominant term in the minimum distance depends on the choice of $\eta$. If $\eta \in (0, 1/\beta]$, the auxiliary code controls the distance, yielding parameters $[[\Theta({(n_i')}^{\zeta+o(1)}), \Omega(n_i'), \Omega( {(n_i')}^{\eta \beta + o(1)} )]]$. If $\eta \in (1/\beta, \infty)$, the primary code controls the distance, yielding parameters $[[\Theta((n_i')^{\zeta+o(1)}), \Omega(n_i'), \Omega(n_i' )]]$. 

Rescaling these parameters to express them in terms of the final overall blocklength $n_i$, we obtain our final scaling bounds. For $\eta > 0$, the code parameters scale as $[[\Theta(n_i), \Omega({(n_i)}^{1/\zeta + o(1)}), \Omega( {(n_i)}^{\min\{1,\eta \beta\}/\zeta + o(1)} )]]$.

We now proceed to auxiliary code families that exhibit linear growth in their logical dimension.

\emph{Linear Target Support Sizes $(0 < \epsilon \le 1)$ and Sub-linear Active Control Qubit Tuples $(0 < \epsilon' < 1)$:} By choosing any fixed constant $0 < \eta \le 1-\epsilon'$, the overall code parameters scale as $[[\Theta(n_i'), \Theta(n_i'), \Omega((n_i')^{\eta \beta + o(1)})]]$ (since $\eta, \beta \le 1$). 

If the target support size grows super-linearly or if the number of active control qubits grows at least linearly, we can apply the scaling analysis of the at least linear target support sizes regime.

\emph{CSS Families Realizing All Addressable Logical Gates}
We now broaden our scope to constructing code families that transversally realize any addressable logical $\mczrot$ gate. As detailed at the conclusion of Section~\ref{app_targ_addr_log_mc}, the resulting CSS code $(C_{1,i,r}, C_{2,i,r})_{\tn{CSS}}$ achieves the parameters $[[n_{i,r}, k_{i,r}, d_{i,r}]]$, given by
\begin{equation*}
    n_{i,r}=n_{i}' + \widetilde{n}_r M_{i,r} = \Theta(n_{i}' + (n_i')^{\eta} M_{i,r} ), \quad k_{i,r}=\Omega(n_i'), \quad d_{i,r} = \Omega\left( {(n_i')}^{\min\{1, \eta \beta + o(1)\}}\right), 
\end{equation*}
where $M_{i,r} = \binom{k_i'}{m} = \Theta((n_i')^m)$. 

Consequently, achieving an asymptotically constant rate is impossible in this regime. Instead, the overall code parameters scale as $[[\Theta({(n_i')}^{\zeta + o(1)}), \Omega(n_i'), \Omega( {(n_i')}^{\min\{1, \eta \beta + o(1)\}} )]]$, where $\zeta=m+\eta$. If $\eta \in (0, 1/\beta]$, the auxiliary code controls the distance, yielding parameters $[[\Theta((n_i')^{\zeta}), \Omega(n_i'), \Omega( {(n_i')}^{\eta \beta + o(1)} )]]$. If $\eta \in (1/\beta, \infty)$, the primary code controls the distance, yielding parameters $[[\Theta({(n_i')}^{\zeta}), \Omega(n_i'), \Omega( n_i' )]]$. Rescaling these parameters to express them in terms of the final overall blocklength $n_i$, we obtain our final scaling bounds. For $\eta  > 0$, the code parameters scale as $[[\Theta(n_i), \Omega({(n_i)}^{1/\zeta}), \Omega( {(n_i)}^{\min\{1,\eta \beta\}/\zeta + o(1)} )]]$. 

Analogous to the construction of CSS code families for addressable logical $1$-controlled $Z$-rotations detailed in Section~\ref{sec_css_families_1_controlled}, we can construct CSS code families realizing addressable logical multi-controlled $Z$-rotations. This is achieved by incorporating auxiliary code families that realize logical transversal $R_Z\left(\frac{\pi}{2^{\ell+m-1}}\right)$ gate into the construction framework developed in Section~\ref{gen_family_framework_mcz}. The constructed CSS families realizing addressable logical $(m-1)$-controlled $Z$-rotations are summarized in Table~\ref{mcz_families}.

\renewcommand{\arraystretch}{1}
\begin{table*}[!h]
\begin{center}
\begin{tabular}{| p{9.5cm} | p{6cm} |}
    \hline
\textbf{Families} & \textbf{Logical gates}  \\
 \hline
\makecell[l]{ $m \ge 2$ and $\ell \ge 0$:\\
$[[n, \Theta(n), \Omega(n^{\eta \beta_{i,\ell+m-1}(\alpha_{i,\ell+m-1}) + o(1)}]]$, \\ for $0 < \eta < \min\{1, \frac{1-\epsilon}{1-\alpha_{i,\ell+m-1}}\}$, \\ $0 < \alpha_{i,\ell+m-1} < H_2\left(\frac{1}{\ell+m-1}\right)$, \\ explicit for $i=1$, non-explicit for $i=2$}  & \makecell[l]{Targeted addressable logical \\ $C_{A_i}^{(m-1)}\left(\sqzrot\right)$ gates, \\ for $\epsilon, \epsilon' \in (0,1)$} \\
\hline
\makecell[l]{ $m \ge 2$ and $\ell \ge 0$:\\
$[[n, \Omega(n^{1/\zeta_i+o(1)}), \Omega(n^{\min\left\{1,\eta \beta_{i,\ell+m-1}(\alpha_{i,\ell+m-1}) \right\}/\zeta_i + o(1)})]]$, \\ $\zeta_i=\max\{\eta+\epsilon', \eta (1-\alpha_{i,\ell+m-1}) + \epsilon\}$, \\ for $\eta > 0$, $0 < \alpha_{i,\ell+m-1} < H_2\left(\frac{1}{\ell+m-1}\right)$, \\ explicit for $i=1$, non-explicit for $i=2$}  & \makecell[l]{Targeted addressable logical \\ $C_{A_i}^{(m-1)}\left(\sqzrot\right)$ gates, \\ for $\epsilon \in (0,m]$ and $\epsilon' \in (0,m-1]$} \\
\hline
\makecell[l]{ Explicit:
$[[n, \Theta(n), \Omega(2^{(c_{R_3}+o(1)) \sqrt{\eta \log(n)}})]]$, \\ for $0 < \eta \le 1-\epsilon'$, $R_3 \in (0,1/3)$}  & \makecell[l]{Targeted addressable logical \\ $C_{A_i}^{(2)}(Z)$ gates, \\ for $\epsilon \in (0,1]$ and $\epsilon' \in (0,1)$} \\
\hline
\makecell[l]{Explicit:
$[[n, \Omega(n^{1/\zeta}), \Omega(2^{((c_{R_3}/\zeta)+o(1)) \sqrt{(\eta/\zeta) \log(n)}})]]$, \\ $\zeta=\max\{\eta+\epsilon', \epsilon\}$, for $\eta > 0$, $R_3 \in (0,1/3)$}  & \makecell[l]{Targeted addressable logical \\ $C_{A_i}^{(2)}(Z)$ gates, \\ for $\epsilon \in (0,3]$ and $\epsilon' \in (0,2]$} \\
\hline
\makecell[l]{ $m \ge 2$ and $\ell \ge 0$:\\
$[[n, \Omega(n^{1/\zeta+o(1)}), \Omega(n^{\min\left\{1,\eta \beta_{i,\ell+m-1}(\alpha_{i,\ell+m-1}) \right\}/\zeta + o(1)})]]$, \\ $\zeta=m+\eta$, for $\eta> 0$, $0 < \alpha_{i,\ell+m-1} < H_2\left(\frac{1}{\ell+m-1}\right)$, \\ explicit for $i=1$, non-explicit for $i=2$}  & \makecell[l]{Any addressable logical \\ $\mczrot$ gate} \\
\hline
\end{tabular}
\vspace{0.5cm}
\caption{CSS families realizing addressable logical $(m-1)$-controlled $Z$-rotations}\label{mcz_families}
\end{center}
\end{table*}
\subsection{Relaxed Modular Equations to Realize Addressable Logical Multi-Controlled $Z$-Rotations}\label{sec:mc_relaxed}

As in Section~\ref{sec:sq_relaxed}, we demonstrate that the system of modular equations in Proposition~\ref{prop:target_addr_log_mc_basis} can be relaxed. Applying Lemma~\ref{lem:relaxed} in Appendix~\ref{sec:relaxed} to Proposition~\ref{prop:target_addr_log_mc_basis} yields the following alternative characterization:

\begin{proposition}\label{prop:target_addr_log_mc_basis_relaxed}        
    Let ${(C_1, C_2)}_{\tn{CSS}}$ be an $[[n,k]]$ CSS code, where $\{x_1, \ldots, x_{k_2}\}$ and $\{y_1, \ldots, y_k\}$ are bases for $C_2$ and $C_1/C_2$, respectively. For any order-$m$ target address tensor $A \in \mathbb{F}_2^{k \times \cdots \times k}$, the code realizes the addressable logical gate $\addmczrot$ via a transversal physical $Z$-rotation $U(\ell+m-1,w')$ for some $w' \in \mathbb{Z}^n$ if and only if there exists a vector $w \in \mathbb{Z}^n$ such that for all $a \in \mathbb{F}_2^k$, $p \in [k_2]$, ordered indices $1 \le p_1 < \cdots < p_q \le k_2$ with $q \in \{2, \ldots, k_2\}$, $i \in [k]$, ordered indices $1 \le i_1 < \cdots < i_j \le k$, the following conditions hold:
    \begin{align*}
        &w \cdot x_p = 0 \pmod{2^{\ell+m-1}}, \qquad \qquad  w \cdot (x_{p_1} \ast \cdots \ast  x_{p_q}) = 0 \pmod{2^{\ell + m - q +1}}, \\
        \qquad
        &\qquad \qquad \qquad  \; \; \; \   w \cdot {(x \ast y_a)} = 0 \pmod{2^{\ell+m-1}}, \\
        &\qquad\qquad w \cdot \left( y_{i_1} \ast \cdots \ast y_{i_{j}}\right) = \begin{cases} 
            {(-1)}^{m-1} A(i_1, \ldots,i_m) \pmod{2^{\ell+1}}, & \text{if } j=m, \\ 
            0 \pmod{2^{\ell+m-1}}, & \text{if }j=1, \\
            0 \pmod{2^{\ell+m-j+1}}, & \text{otherwise.} 
        \end{cases}
    \end{align*}
\end{proposition}

\subsection{Proof of Lemma~\ref{lem:schur_basis_vectors_for_mc}}\label{sec:technical_decomp_mod_mc}
\begin{proof}       
    We establish necessity via a two-phase induction on $j \in [\ell+m]$: the first phase covers $j=1$ through $m-1$ (with base case $j=1$), while the second spans $j=m$ through $\ell+m$ (with base case $j=m$). We begin with the base case of the first phase: for any $i \in [k]$, evaluating~\eqref{eq:orig_mod_code_mc} at $a = e_{i}$ establishes the base case
    \begin{equation*}
        w \cdot x_i = 0 \pmod{2^{\ell+m}},
    \end{equation*}        
    where the last equality holds because $m \ge 2$. For the induction step, assume the result holds for all $p$ where $1 \le p \le j < m-1$. Then, for any ordered indices $1 \le i_1 < \cdots < i_{p} \le k$, we have
    \begin{equation}\label{eq:ind_hyp_mc}
        w \cdot \left(x_{i_1} \ast \cdots \ast x_{i_p} \right) = 0 \pmod{2^{\ell+m-p+1}}.
    \end{equation}
    Now, for any ordered indices $1 \le i_1 <  \cdots < i_{j+1} \le k$, evaluating~\eqref{eq:orig_mod_code_mc} at $a = \bigoplus_{p=1}^{j+1} e_{i_p}$ yields
    \begin{equation}\label{eq:eval_1_mc}
        w \cdot \left( \bigoplus_{p=1}^{j+1} x_{i_p} \right) = 0 \pmod{2^{\ell+m}}.
    \end{equation}    
    By expanding the binary addition on the left-hand side of~\eqref{eq:eval_1_mc} and applying the induction hypothesis~\eqref{eq:ind_hyp_mc}, we obtain
    \begin{equation}\label{eq:eval2_mc}
        w \cdot \left( \bigoplus_{p=1}^{j+1} x_{i_p} \right) = {(-2)}^{j} w \cdot \left( x_{i_1} \ast \cdots \ast x_{i_{j+1}} \right) \pmod{2^{\ell+m}}. 
    \end{equation}
    Equating~\eqref{eq:eval_1_mc} and~\eqref{eq:eval2_mc}, we complete the phase one induction argument:
    \begin{equation}\label{eq:ind_interim_mc}
         w \cdot \left( x_{i_1} \ast \cdots \ast x_{i_{j+1}} \right) = 0 \pmod{2^{\ell+m-j}}.
    \end{equation}  
    
    We now turn to the base case of the second phase ($j=m$). For any ordered indices $1 \le i_1 < \cdots < i_m \le k$, evaluating~\eqref{eq:orig_mod_code_mc} at $a = \bigoplus_{p=1}^{m} e_{i_p}$ yields
    \begin{equation}\label{eq:eval_3_mc}
        w \cdot \left( \bigoplus_{p=1}^{m} x_{i_p} \right) = 2^{m-1}A(i_1, \ldots, i_m) \pmod{2^{\ell+m}}. 
    \end{equation}      
    Expanding the binary addition on the left-hand side of~\eqref{eq:eval_3_mc} and incorporating the modular conditions for $1 \le j \le m-1$, we arrive at the base case, $j=m$:
    \begin{equation}
        w \cdot \left( x_{i_1} \ast \cdots \ast x_{i_m} \right) = {(-1)}^{m-1} A(i_1, \ldots, i_m) \pmod{2^{\ell + 1}}.
    \end{equation}    
    For the second phase induction, assume the result holds for all $p$ where $m \le p \le j < \ell + m$. Then, for any ordered indices $1 \leq i_1 < \cdots < i_p \le k$, we have
    \begin{equation}\label{eq:ind_hyp_phase_2_mc}
        w \cdot (x_{i_1} \ast \cdots \ast x_{i_p}) =  \begin{cases} 
    {(-1)}^{m-1} A(i_1, \ldots, i_m) \pmod{2^{\ell+1}}, & \text{if } p=m, \\ 
                 0 \pmod{2^{\ell+m-p+1}}, & \text{otherwise.} 
    \end{cases} 
    \end{equation}
    Now, for any ordered indices $1 \le i_1 < \cdots < i_{j+1} \le k$, evaluating~\eqref{eq:orig_mod_code_mc} at $a=\bigoplus_{p=1}^{j+1}e_{i_p}$, we obtain
    \begin{equation}\label{eq:eval_4_mc}
        w \cdot \left( \bigoplus_{p=1}^{j+1} x_{i_p} \right) = 2^{m-1} \sum_{1 \le q_1 < \cdots < q_m \leq j+1} A(i_{q_1}, \ldots, i_{q_m}) \pmod{2^{\ell+m}}.
    \end{equation}    
    By expanding the binary addition on the left-hand side of~\eqref{eq:eval_4_mc} and applying the induction hypothesis~\eqref{eq:ind_hyp_phase_2_mc}, as well as the first phase modular conditions, we obtain
    \begin{equation*}
        w \cdot \left( x_{i_1} \ast \cdots \ast x_{i_{j+1}} \right) = 0 \pmod{2^{\ell+m-j}},
    \end{equation*}
    completing the induction and establishing necessity. To prove sufficiency, we begin by expanding the binary addition on the left-hand side of~\eqref{eq:orig_mod_code_mc} into an integer sum:
    \begin{equation*}
        w \cdot \left(\bigoplus_{i=1}^{k} a(i) x_i \right) = \sum_{j=1}^{k} (-2)^{j-1} \left( \sum_{1 \le i_1 < \cdots < i_j \le k} \left( \prod_{p=1}^{j} a({i_p}) \right) w \cdot (x_{i_1} \ast \cdots \ast x_{i_j}) \right).
    \end{equation*}
    
    Because we evaluate this expression modulo $2^{\ell+m}$, any term in the outer summation with $j \geq \ell+m+1$ vanishes. Substituting~\eqref{eq:basis_mc} into the remaining terms completes the proof:
    \begin{equation*}
        w \cdot {\left(\bigoplus_{i=1}^{k} a(i) x_i \right)} 
        = 2^{m-1} \sum_{1 \leq i_1 < \cdots < i_m \le k} \left( \prod_{p=1}^m a({i_p}) \right) A(i_1, \ldots, i_m) \pmod{2^{\ell+m}}.
    \end{equation*}
\end{proof}
\subsection{Proof of Proposition~\ref{prop:schur_sum_basis_mc_secc_ode}}\label{sec:schur_sum_basis_mc_secc_ode}
\begin{proof}
By the inclusion-exclusion principle,
\begin{align*}
    \widetilde{w} \cdot \left(\bigoplus_{i=1}^{\widetilde{k}} \widetilde{y}_i \right) &= \sum_{j=1}^{\widetilde{k}} {(-2)}^{j-1} \left( \sum_{1 \le i_1 < \cdots < i_j \le \widetilde{k}} \widetilde{w} \cdot (\widetilde{y}_{i_1} \ast \cdots \ast \widetilde{y}_{i_j}) \right) \\
    &= \widetilde{k} \pmod{2^{\ell+m}},
\end{align*}
where the last equality follows by~\eqref{eq:secondary_css_basis_mod_mc}. We now turn to prove the second claim: For any $p \in [\widetilde{k}]$, expanding the left-hand side of~\eqref{eq:schur_prop_sc_code_mc_1} via the inclusion-exclusion principle alongside the idempotent property $\widetilde{y}_p \ast \widetilde{y}_p = \widetilde{y}_p$ yields
\begin{align*}
   \widetilde{w} \cdot \left( \left(\bigoplus_{i=1}^{\widetilde{k}} \widetilde{y}_i \right) \ast \widetilde{y}_p \right) &= \widetilde{w} \cdot \left(\widetilde{y}_p \oplus \left(\bigoplus_{i \ne p} \widetilde{y}_i \ast \widetilde{y}_p\right)\right) \\
   &= \widetilde{w} \cdot \widetilde{y}_p + \widetilde{w} \cdot \left( \bigoplus_{i\ne p} \widetilde{y}_i \ast \widetilde{y}_p \right) - 2 \widetilde{w} \cdot \left( \widetilde{y}_p \ast \left( \bigoplus_{i\ne p} \widetilde{y}_i \ast \widetilde{y}_p \right) \right) \\
   &= \widetilde{w} \cdot \widetilde{y}_p - \widetilde{w} \cdot \left( \bigoplus_{i\ne p} \widetilde{y}_i \ast \widetilde{y}_p \right) 
\end{align*}
Expanding the final term above via the inclusion-exclusion formula, we obtain
\begin{align*}
    \widetilde{w} \cdot \left( \bigoplus_{i\ne p} \widetilde{y}_i \ast \widetilde{y}_p \right)  &= \sum_{q=1}^{\widetilde{k}-1} {(-2)}^{q-1} \left( \sum_{\substack{1 \le r_1 < \cdots < r_q \le \widetilde{k}: \\
    r_s \ne p \ \forall s \in [q] }} \widetilde{w} \cdot (\widetilde{y}_{r_1} \ast \cdots \ast \widetilde{y}_{r_q} \ast \widetilde{y}_p) \right) \\
    &= 0 \pmod{2^{\ell+m-1}},
\end{align*}
by virtue of~\eqref{eq:secondary_css_basis_mod_mc} again. Substituting this result back into our main computation gives
\begin{equation*}
   \widetilde{w} \cdot \left( \left(\bigoplus_{i=1}^{\widetilde{k}} \widetilde{y}_i \right) \ast \widetilde{y}_p \right) = \widetilde{w} \cdot \widetilde{y}_p \pmod{2^{\ell+m-1}}. 
\end{equation*}
Since we already know from~\eqref{eq:secondary_css_basis_mod_mc} that $\widetilde{w} \cdot \widetilde{y}_p=-1\pmod{2^{\ell+m}}$, reducing the modulus establishes
\begin{equation*}
    \widetilde{w} \cdot \left( \left(\bigoplus_{i=1}^{\widetilde{k}} \widetilde{y}_i \right) \ast \widetilde{y}_p \right) = 1 \pmod{2^{\ell+m-1}},
\end{equation*}
completing the proof of the second claim~\eqref{eq:schur_prop_sc_code_mc_1}.

To prove the third claim~\eqref{eq:schur_prop_sc_code_mc_2}, consider an arbitrary ordered sequence of distinct indices $1 \le p_1 < \cdots < p_q \leq \widetilde{k}$ with $q \ge 2$, and let $P=\{p_1, \ldots, p_q\}$. Evaluating the left-hand side of~\eqref{eq:schur_prop_sc_code_mc_2} based on the parity of $q$, we have
\begin{equation*}
   \widetilde{w} \cdot \left( \left(\bigoplus_{i=1}^{\widetilde{k}} \widetilde{y}_i \right) \ast \left( \widetilde{y}_{p_1} \ast \cdots \ast \widetilde{y}_{p_q}  \right) \right) = \begin{cases} 
        \widetilde{w} \cdot \left(\bigoplus_{i \in [\widetilde{k}] \setminus P} \widetilde{y}_i \ast \widetilde{y}_{p_1} \ast \cdots \ast \widetilde{y}_{p_q} \right), & \text{if $q$ is even}, \\ 
         \widetilde{w} \cdot \left( \left( \widetilde{y}_{p_1} \ast \cdots \ast \widetilde{y}_{p_q} \right) \oplus \bigoplus_{i \in [\widetilde{k}] \setminus P} \widetilde{y}_i \ast \widetilde{y}_{p_1} \ast \cdots \ast \widetilde{y}_{p_q} \right), & \text{otherwise.} 
    \end{cases} 
\end{equation*} 
For odd $q$, applying inclusion-exclusion yields
\begin{align*}
    \widetilde{w} \cdot &\left( \left( \widetilde{y}_{p_1} \ast \cdots \ast \widetilde{y}_{p_q} \right) \oplus \left( \bigoplus_{i \in [\widetilde{k}]\setminus P} \widetilde{y}_i \ast \widetilde{y}_{p_1} \ast \cdots \ast \widetilde{y}_{p_q} \right) \right) \\
    &\qquad \qquad \qquad =\widetilde{w} \cdot (\widetilde{y}_{p_1} \ast \cdots \ast \widetilde{y}_{p_q}) - \widetilde{w} \cdot \left( \bigoplus_{i \in [\widetilde{k}]\setminus P} \widetilde{y}_i \ast \widetilde{y}_{p_1} \ast \cdots \ast \widetilde{y}_{p_q} \right) \\
    &\qquad \qquad \qquad = - \widetilde{w} \cdot \left( \bigoplus_{i \in [\widetilde{k}]\setminus P} \widetilde{y}_i \ast \widetilde{y}_{p_1} \ast \cdots \ast \widetilde{y}_{p_q} \right) \pmod{2^{\ell+m-q}},
\end{align*}
via~\eqref{eq:secondary_css_basis_mod_mc}. Note that showing that the last expression above is equal to $0 \pmod{2^{\ell+m-q}}$ also resolves the case of even $q$. 
To this end, applying inclusion-exclusion once again gives us
\begin{align*}
    \widetilde{w} \cdot \left( \bigoplus_{i \in [\widetilde{k}] \setminus P} \widetilde{y}_i \ast \widetilde{y}_{p_1} \ast \cdots \ast \widetilde{y}_{p_q} \right) &= \sum_{r=1}^{\widetilde{k}-q} {(-2)}^{r-1} \left( \sum_{\substack{1 \le s_1 < \cdots < s_r \le \widetilde{k}: \\s_t \notin P \ \forall t \in [r]}} \widetilde{w} \cdot (\widetilde{y}_{s_1} \ast \cdots \ast \widetilde{y}_{s_r} \ast \widetilde{y}_{p_1} \ast \cdots \ast \widetilde{y}_{p_q} ) \right) \\
    &= 0 \pmod{2^{\ell+m-q}},
\end{align*}
again by~\eqref{eq:secondary_css_basis_mod_mc}. This completes the proof of the third claim~\eqref{eq:schur_prop_sc_code_mc_2}.
\end{proof}

\section{Lifting Modular Equations}\label{sec:relaxed}

In this appendix, we prove the lemma that provides the mechanism for obtaining Propositions~\ref{prop:target_addr_log_sq_basis_relaxed} and~\ref{prop:target_addr_log_mc_basis_relaxed} from Propositions~\ref{prop:target_addr_log_sq_basis} and~\ref{prop:target_addr_log_mc_basis}, respectively. 
Specifically, the lemma demonstrates that if an integer vector $w$ satisfies certain modulo-$2^\ell$ equations over a code's basis vectors, alongside modular constraints on their $t$-fold Schur products (for $t \ge 2$), then we can systematically refine $w$ into a new vector $w'$ in such a way that the basis modular equations are lifted to their modulo-$2^{\ell+1}$ counterparts while leaving unchanged the modular constraints on all the $t$-fold ($t \ge 2$) Schur products.

\begin{lemma}\label{lem:relaxed}    
    Let $C \subseteq \mathbb{F}_{2}^n$ be a binary code with basis $\{x_1, \ldots, x_k\}$, and let $A \in \mathbb{F}_2^k$ be a fixed vector. Suppose there exists a vector $w \in \mathbb{Z}^n$ such that for all $i \in [k]$ and ordered indices $1 \leq i_1 < \cdots < i_{j} \le k$ with $j\in \{2, \ldots, \ell+1\}$, the following conditions hold:
    \begin{align} 
        w \cdot x_{i} &= A(i) \pmod{2^{\ell}}, \label{eq:basis_1_relaxed} \\ 
        w \cdot (x_{i_1} \ast \cdots \ast x_{i_{j}}) &= 0 \pmod{2^{\ell-j+2}}. \label{eq:basis_2_relaxed}  
    \end{align}
    Then, there exists an integer vector $w' \in \mathbb{Z}^n$ such that
    \begin{align} 
        w' \cdot x_{i} &= A(i) \pmod{2^{\ell+1}}, \label{eq:basis_3_relaxed} \\ 
        w' \cdot (x_{i_1} \ast \cdots \ast x_{i_{j}}) &= 0 \pmod{2^{\ell-j+2}}. \label{eq:basis_4_relaxed}
    \end{align}
\end{lemma}
\begin{proof}
    By Lemma~\ref{lem:schur_basis_vectors}, conditions~\eqref{eq:basis_3_relaxed} and~\eqref{eq:basis_4_relaxed} hold if and only if
    \begin{equation*}
        w' \cdot x_a  = w_H(a \ast A) \pmod{2^{\ell+1}}. 
    \end{equation*}
    To proceed, we introduce a function $f_1$ defined over the code $C$ as follows:
    \begin{equation*} 
        f_1(x_a) = w \cdot x_a \pmod{2^{\ell+1}},
    \end{equation*}
    where $x_a = \bigoplus_{p=1}^k a(p) x_p$. Expanding the dot product yields
    \begin{align}
        w \cdot x_a = \sum_{j=1}^k {(-2)}^{j-1} \left( \sum_{1 \le i_1 < \cdots < i_j \le k} \left( \prod_{p=1}^j a({i_p}) \right) w \cdot \left( x_{i_1} \ast \cdots \ast x_{i_j} \right) \right).
    \end{align}
        From~\eqref{eq:basis_2_relaxed}, $f_1(x_a)$ simplifies to
    \begin{equation*}
        f_1(x_a) = \sum_{i=1}^k a(i) w \cdot x_i \pmod{2^{\ell+1}}.
    \end{equation*}    
    Furthermore,~\eqref{eq:basis_1_relaxed} requires that for each $i \in [k]$, the difference $w \cdot x_{i} - A(i)$ must be either $0$ or $2^{\ell} \pmod{2^{\ell+1}}$. Factoring $2^\ell$ from these terms, the expression evaluates to
    \begin{align}
        f_1(x_a) &= \left( w_H(a \ast A) \pmod{2^{\ell+1}} \right) +  \left( \sum_{i=1}^k a(i) \left( w \cdot x_i - A(i) \right) \pmod{2^{\ell+1}} \right) \nonumber \\
        &=\left( w_H(a \ast A) \pmod{2^{\ell+1}} \right) + \left( 2^{\ell}\left( \sum_{i=1}^k a(i) \frac{\left( w \cdot x_i - A(i) \right)}{2^\ell} \pmod{2} \right) \pmod{2^{\ell+1}} \right).\label{eq45b}
    \end{align}
    Next, we isolate the modulo-$2$ summation by defining the function $f_2$ as
    \begin{equation*}
        f_2 \left(\bigoplus_{i=1}^k a(i) x_i \right) = \sum_{i=1}^k a(i) \frac{w \cdot x_i - A(i)}{2^\ell} \pmod{2}.
    \end{equation*}    
    As $f_2$ is a linear transformation from $C$ to $\mathbb{F}_2$, there exists a vector $w_1\in \mathbb{F}_2^n$ such that
    \begin{equation*}
        f_2(x_a) =  w_1 \cdot x_a \pmod{2}.
    \end{equation*}
    Substituting this dot product representation back into~\eqref{eq45b} yields
    \begin{equation*}
        (w - 2^\ell w_1) \cdot x_a = w_H(a \ast A) \pmod{2^{\ell+1}}.
    \end{equation*}
    Setting $w' = w - 2^{\ell} w_1$ satisfies the target modular equations~\eqref{eq:basis_3_relaxed} and~\eqref{eq:basis_4_relaxed}, thereby completing the proof.
\end{proof}

It should be pointed out that these modular equations cannot be further relaxed. If we weaken the Schur product condition $w \cdot (x_i \ast x_j) = 0 \pmod{2^{\ell}}$ to modulo $2^{\ell-1}$, the expansion of $f_1$ in~\eqref{eq45b} gains an additional residual term:
    \begin{equation*}
        2^{\ell} \sum_{1 \leq i < j \leq k} a(i) a(j) \frac{w \cdot (x_i \ast x_j)}{2^{\ell-1}} \pmod{2}.
    \end{equation*}
The cross-terms $a(i) a(j)$ make this expression non-linear over $C$. Consequently, it cannot be expressed as a dot product with a fixed vector, which prevents the linear refinement of $w$. While this residual could be interpreted as a bi-linear map over $C \times C$, physically implementing such a map requires physical two-qubit $CZ$ gates, destroying fault tolerance.

Furthermore, weakening the constraint $w \cdot x_i = 0 \pmod{2^{\ell+1}}$ to  modulo-$2^{\ell-1}$ produces the residual term:
    \begin{equation*}
        2^{\ell-1} \sum_{i=1}^{k} a(i) \frac{w \cdot x_i-A(i)}{2^{\ell-1}} \pmod{4}.
    \end{equation*}    
Because arithmetic modulo $4$ is not linear over $C$, we once again forfeit the linearity required to refine $w$.

\section{Proof of Lemma~\ref{lem:z_rotation_parameterization}}\label{app:z_rotation_parameterization_proof}
\begin{proof}
    We first verify sufficiency. For any non-negative integer $p$ and integer vector $w$, the operator $U$ defined in~\eqref{eq:z_rotation_diag} is precisely the transversal $Z$-rotation
    \begin{equation*}
        U = \bigotimes_{i=1}^n R_Z \left( \frac{\pi w(i)}{2^p} \right).
    \end{equation*}
    Conversely, consider an arbitrary dyadic transversal $Z$-rotation as described in Remark~\ref{rmk:transversalzrotations}:
    \begin{equation*}
        U = \bigotimes_{i=1}^n R_Z\left( \frac{\pi q_i}{2^{p_i}} \right).
    \end{equation*}
    Let $p_{\max} = \max \{ p_i : i \in [n] \}$, and define the integer vector $w = (w(1), \ldots, w(n)) \in \mathbb{Z}^n$ by setting $w(i) = 2^{p_{\max} - p_i} q_i$.
    By construction, at least one component of $w$ is odd---specifically, any $w(i)$ for which $p_i = p_{\max}$. For any computational basis state $\ket{x}$ with $x \in \mathbb{F}_2^n$, the action of $U$ evaluates to 
\begin{align*}
    U \ket{x} &= \left( \prod_{i=1}^n  \exp{ \iota \frac{\pi}{2^{p_i}} q_i x_i } \right)  \ket{x} \\
    &= \exp{ \iota \frac{\pi}{2^{p_{\max}}} \left( w \cdot x \right)} \ket{x}.
\end{align*}
Consequently, $\physicalU$ admits the diagonal representation
\begin{equation*}
U = \diag\left( \exp{\iota \frac{\pi}{2^{p_{\max}}} \left( w \cdot x \right) } : x \in \mathbb{F}_2^n \right).
\end{equation*}
To establish uniqueness, suppose there exist two such parameterizations:
\begin{equation}\label{eq:diagonal_equality}
    U = \diag \left( \exp{\iota \frac{\pi}{2^{p}} \left( w_1 \cdot x \right) } : x \in \mathbb{F}_2^n \right) = \diag\left( \exp{\iota \frac{\pi}{2^{q}} \left( w_2 \cdot x \right) } : x \in \mathbb{F}_2^n \right),
\end{equation}
where $p, q \ge 0$ are integers, and $w_1, w_2 \in \mathbb{Z}^n$ each contain at least one odd component. Let $i$ be an index such that $w_1(i)$ is odd. Evaluating both the representations at the standard basis vector $x = e_i$ yields
\begin{equation*}
    \exp{\iota \frac{\pi}{2^{p}} w_1(i)} = \exp{\iota \frac{\pi}{2^{q}} w_2(i)}.
\end{equation*}
Because $w_1(i)$ is odd, phase equality mandates that $q \geq p$. A symmetric argument applied to an odd component of $w_2$ implies that $p \geq q$, forcing $p = q$. Substituting $p = q$ in the diagonal equality~\eqref{eq:diagonal_equality} and evaluating at $x = e_i$ for every index $i \in [n]$ demonstrates that
\begin{equation*}
w_1(i) = w_2(i) \pmod{2^{p+1}},
\end{equation*}
which completes the proof.
\end{proof}

\section{Proof of Proposition~\ref{prop:addressable_multi_controlled}}\label{app:prop_addressable_multi_controlled}
\begin{proof}
    By Eq.~\eqref{eq:orig_char} in the proof of Theorem~\ref{thm:refined_fault_tolerant_logical_diag_gates}, the physical gate $U(p,w)$ realizes the logical gate $\logicalU$ if and only if $p \geq \ell$ and for all $x \in C_2$ and $a \in \mathbb{F}_2^k$, the following condition holds:
    \begin{equation}\label{eq:strict_mod_condition}
    w \cdot \left( x \oplus y_a \right) = 2^{p - \ell} \left( \sum_{1 \leq i_1 < \cdots < i_m \leq k} A(i_1, \ldots, i_m)  a(i_1) \cdots  a(i_m) \right) \pmod{2^{p+1}}.
    \end{equation}
    We now show that condition~\eqref{eq:strict_mod_condition} implies the stricter bound $p \geq \ell + m - 1$. We prove by contradiction, assuming $p \leq \ell+m-2$. Because the addressability tensor $A$ possesses at least one odd component, we may permute the indices such that $A(1, \ldots, m)$ is odd.

    Let $\{y_1, \ldots, y_k\}$ denote the basis for the coset space $C_1/C_2$. For any non-empty subset $S \subseteq [m-1]$, setting $x=0$ and $a=\bigoplus_{i \in S} e_i$ in condition~\eqref{eq:strict_mod_condition}, we obtain
    \begin{equation}\label{eq:necessity_expanded}
    w \cdot y_{a} = w \cdot \left( \bigoplus_{i \in S} y_i \right) = 0 \pmod{2^{p+1}}.
    \end{equation}
    As a consequence of~\eqref{eq:necessity_expanded}, for any non-empty subset $\{i_1, \ldots, i_r\} \subseteq [m-1]$, it follows that
    \begin{equation}\label{eq:schur_product_simplified}
    (-2)^{r-1} \left( w \cdot (y_{i_1} \ast \cdots \ast y_{i_r}) \right) = 0 \pmod{2^{p+1}}.
    \end{equation}
    Applying~\eqref{eq:schur_product_simplified} to the evaluation at $a = \bigoplus_{i=1}^m e_i$, we deduce that
    \begin{equation}\label{eq:necessity_simplified}
    w \cdot y_{a} = (-2)^{m-1} \left( w \cdot (y_1 \ast \cdots \ast y_m) \right) \pmod{2^{p+1}}.
    \end{equation}
    On the other hand, directly evaluating condition~\eqref{eq:strict_mod_condition} at $a = \bigoplus_{i=1}^m e_i$ (with $x = 0$) yields
    \begin{equation}\label{eq:strict_mod_inference} 
        w \cdot y_{a} = 2^{p - \ell} A(1, \ldots, m) \pmod{2^{p+1}}. 
    \end{equation}
    Combining~\eqref{eq:necessity_simplified} and~\eqref{eq:strict_mod_inference}, we can express this modular congruence as an integer equation. For some $c \in \mathbb{Z}$, we have
    \begin{align*}
        2^{p - \ell} A(1, \ldots, m) &=  {(-2)}^{m-1} \left( w \cdot (y_1 \ast \cdots \ast y_m) \right)  + 2^{p+1} c \\[5pt]
        \implies A(1, \ldots, m) &= {(-1)}^{m-1} 2^{m-p+\ell-1} \left( w \cdot (y_1 \ast \cdots \ast y_m) \right)  + 2^{\ell+1} c.
    \end{align*}
    Because we assumed $p \leq m + \ell -2$, the exponent $m - p + \ell -1$ is strictly positive. Furthermore, since $\ell \geq 0$, the exponent $\ell+1$ is also strictly positive. Thus, the right-hand side of the equation is even, implying that $A(1, \ldots, m)$ must be an even integer. This directly contradicts our initial premise that $A(1,\ldots, m)$ is odd. We therefore conclude that $p \geq \ell+m-1$.
\end{proof}

\section{Proof of Lemma~\ref{lem:schur_basis_vectors}}\label{app:lem_schur_basis_vectors}
\begin{proof}
    We begin by proving necessity. For any $i \in [k]$, evaluating~\eqref{eq:orig_mod_code} at $a = e_{i}$ yields~\eqref{eq:basis_1}:
    \begin{equation}\label{eq:basis_1a}
        w \cdot x_{i} = A(i) \pmod{2^{\ell+1}}.
    \end{equation}
    
    We prove~\eqref{eq:basis_2} by induction on $j \in [\ell+1]$, where the base case $j=1$ is already established by~\eqref{eq:basis_1a}. For the induction step, assume the result holds for all $1 \le p \le j < \ell+1$. Then, for any ordered indices $1 \le i_1 < \cdots < i_p \le k$, if $p=1$ we have~\eqref{eq:basis_1a} and if $p \ge 2$ we have
    \begin{equation}\label{eq:ind_hyp}
        w \cdot \left( x_{i_1} \ast \cdots \ast x_{i_p} \right) = 0 \pmod{2^{\ell -p +2}}.
    \end{equation}
    Now, for any ordered indices $1 \le i_1 <  \cdots < i_{j+1} \le k$, evaluating~\eqref{eq:orig_mod_code} at $a = \bigoplus_{p=1}^{j+1} e_{i_p}$ yields
    \begin{equation}\label{eq:eval_1}
        w \cdot \left( \bigoplus_{p=1}^{j+1} x_{i_p} \right) = \sum_{p=1}^{j+1} A(i_p) \pmod{2^{\ell+1}}. 
    \end{equation}    
    By expanding the binary addition, the left-hand side of~\eqref{eq:eval_1} can be rewritten as
    \begin{equation*}
        w \cdot \left( \bigoplus_{p=1}^{j+1} x_{i_p} \right) = \sum_{p=1}^{j+1} {(-2)}^{p-1} \sum_{1\leq q_1 < \cdots < q_p \leq j+1} w \cdot \left( x_{i_{q_1}} \ast \cdots \ast x_{i_{q_p}} \right) \pmod{2^{\ell+1}}.
    \end{equation*}
    Applying the induction hypothesis~\eqref{eq:basis_1a} and~\eqref{eq:ind_hyp}, we obtain
    \begin{equation}\label{eq:eval2}
        w \cdot \left( \bigoplus_{p=1}^{j+1} x_{i_p} \right) = \sum_{p=1}^{j+1} A(i_p) + {(-2)}^{j} w \cdot \left( x_{i_1} \ast \cdots \ast x_{i_{j+1}} \right) \pmod{2^{\ell+1}}. 
    \end{equation}
    Equating~\eqref{eq:eval_1} and~\eqref{eq:eval2}, we finally obtain
    \begin{equation*}
         w \cdot \left( x_{i_1} \ast \cdots \ast x_{i_{j+1}} \right) = 0 \pmod{2^{\ell-j+1}}, 
    \end{equation*}
    completing the induction and establishing necessity. To prove sufficiency, we begin by expanding the binary addition on the left-hand side of~\eqref{eq:orig_mod_code} into an integer sum:
    \begin{equation*}
        w \cdot \left(\bigoplus_{i=1}^{k} a(i) x_i \right) = \sum_{j=1}^{k} (-2)^{j-1} \left( \sum_{1 \le i_1 < \cdots < i_j \le k} \left( \prod_{p=1}^{j} a({i_p}) \right) w \cdot (x_{i_1} \ast \cdots \ast x_{i_j}) \right).
    \end{equation*}
    
    Because we evaluate this expression modulo $2^{\ell+1}$, any term in the outer summation with $j \geq \ell+2$ contains a coefficient of $(-2)^{j-1}$, which is a multiple of $2^{\ell+1}$ and therefore vanishes. Thus, the expansion strictly truncates to the first $\ell+1$ terms. Substituting~\eqref{eq:basis_1} and~\eqref{eq:basis_2} into these remaining terms causes all higher-order Schur products ($2 \le j \le \ell+1$) to vanish modulo $2^{\ell+1}$ as well, leaving only the linear terms:
    \begin{align*}
        w \cdot {\left(\bigoplus_{i=1}^{k} a(i) x_i \right)} 
        &= \sum_{i=1}^k a(i) A(i) \pmod{2^{\ell+1}}\\
        &= w_H(a \ast A) \pmod{2^{\ell+1}}. 
    \end{align*}
    This completes the proof.
\end{proof}

\section{Proof of Lemma~\ref{lem:params}}\label{app:lem_params}
\begin{proof}
    The block length and logical dimension of the derived CSS code follow directly from the construction. To establish the minimum distance, we begin by lower-bounding the $X$-distance. For any $x \in C_1 \setminus C_2$, we can decompose $x$ as $x = \widehat{x} + y$, where $\widehat{x} \in C_2$ and $y \in C_1/C_2$ represents a non-zero coset element. Partitioning these vectors according to their primary and appended coordinates, we write $\widehat{x} = (x', \widetilde{x})$ and $y = (y', \widetilde{y})$, where $x' \in C_2'$ and $y' \in C_1'/C_2'$. Because the appended submatrix $\left[
        G_1^{(1)} \; \cdots \; G_1^{(m)}
    \right]$ merely extends the linearly independent rows of $G_{C_1' / C_2'}$, the non-zero condition on $y$ enforces that $y' \neq 0$. Consequently, the weight of $x$ is lower-bounded by the $X$-distance of the primary CSS code, yielding
    \begin{equation*}
        d_X = d_{\min}(C_1 \setminus C_2) \geq d_{\min}(C_1' \setminus C_2') \geq d'.
    \end{equation*}
    Next, we turn to the $Z$-distance, $d_Z$. By the block-diagonal structure of $G_{C_2}$, the dual code $C_2^\perp$ decomposes into a direct sum of its constituent spaces:
    \begin{equation*}
    C_2^\perp = (C_2')^\perp \oplus (C_2^{(1)})^\perp \oplus \cdots \oplus (C_2^{(m)})^\perp.
    \end{equation*}
    Consequently, we obtain the lower bound
    \begin{equation*}
    d_Z = d_{\min}(C_2^\perp \setminus C_1^\perp) \geq d_{\min}( C_2^\perp ) \geq \min\{ d_{\min}((C_2')^\perp), d_{\min}((C_2^{(1)})^\perp), \ldots, d_{\min}((C_2^{(m)})^\perp) \}.
    \end{equation*}
    This completes the proof.
\end{proof}

\bibliographystyle{ieeetr}
\bibliography{references}

@article{RCNP2020,
  title = {On optimality of {CSS} codes for transversal {$T$}},
  author = {Rengaswamy, N. and Calderbank, R. and Newman, M. and Pfister, H.D.},
  journal = {IEEE Journal on Selected Areas in Information Theory},
  volume = {1},
  pages = {499--514},
  year = {2020},
  doi = {10.1109/JSAIT.2020.3012914}
  }

@article{ADP2014,
  title = {Fault-Tolerant Conversion between the {S}teane and {Reed}-{Muller} Quantum Codes},
  author = {Anderson, Jonas T. and Duclos-Cianci, Guillaume and Poulin, David},
  journal = {Phys. Rev. Lett.},
  volume = {113},
  issue = {8},
  pages = {080501},
  numpages = {5},
  year = {2014},
  month = {Aug},
  publisher = {American Physical Society},
  doi = {10.1103/PhysRevLett.113.080501},
  url = {https://link.aps.org/doi/10.1103/PhysRevLett.113.080501}
}

@misc{Z2025,
      title={A topological theory for q{LDPC}: non-Clifford gates and magic state fountain on homological product codes with constant rate and beyond the {$N^{1/3}$} distance barrier}, 
      author={Guanyu Zhu},
      year={2025},
      eprint={2501.19375},
      archivePrefix={arXiv},
      primaryClass={quant-ph},
      url={https://arxiv.org/abs/2501.19375},
      note={arXiv:2501.19375}
}

@article{ZSPWB2025,
  title = {{N}on-{C}lifford and Parallelizable Fault-Tolerant Logical Gates on Constant and Almost-Constant Rate Homological Quantum Low-Density Parity-Check Codes via Higher Symmetries},
  author = {Zhu, Guanyu and Sikander, Shehryar and Portnoy, Elia and Cross, Andrew W. and Brown, Benjamin J.},
  journal = {PRX Quantum},
  volume = {6},
  issue = {4},
  pages = {040361},
  numpages = {43},
  year = {2025},
  month = {Dec},
  publisher = {American Physical Society},
  doi = {10.1103/wcxs-w69t},
  url = {https://link.aps.org/doi/10.1103/wcxs-w69t}
}

@misc{ET2026,
      title={Quantum Codes with Transversal {$CCZ$} Gates and Sublinear {$Z$}-Stabilizers}, 
      author={Ohad Elishco and Itzhak Tamo},
      year={2026},
      eprint={2606.22472},
      archivePrefix={arXiv},
      primaryClass={cs.IT},
      url={https://arxiv.org/abs/2606.22472},
      note={arXiv:2606.22472}
}

@misc{G2025,
      title={Good quantum codes with addressable and parallelizable transversal non-clifford gates}, 
      author={Virgile Guémard},
      year={2025},
      eprint={2510.19809},
      archivePrefix={arXiv},
      primaryClass={quant-ph},
      url={https://arxiv.org/abs/2510.19809}, 
      note={arXiv:2510.19809}
}

@article{BC2026,
      title={Coxeter codes: extending the {Reed}–{Muller} family}, 
      author={Coble N.J. and Barg A.},
      year={2026},
      volume={94},
   journal={Des. Codes Cryptogr.},
}

@misc{LWC2026,
      title={Asymptotically Good Quantum Codes with Addressable Transversal {T} Gates}, 
      author={Tongyin Lin and Bujiao Wu and Bin Cheng},
      year={2026},
      eprint={2609.35129},
      archivePrefix={arXiv},
      primaryClass={quant-ph},
      arXiv={arxiv:2609.35129}, 
}

@article{BCHDK2026,
   title={Geometric Structure and Transversal Logic of Quantum Reed–Muller Codes},
   volume={72},
   ISSN={1557-9654},
   url={http://dx.doi.org/10.1109/TIT.2025.3592631},
   DOI={10.1109/tit.2025.3592631},
   number={1},
   journal={IEEE Transactions on Information Theory},
   publisher={Institute of Electrical and Electronics Engineers (IEEE)},
   author={Barg, Alexander and Coble, Nolan J. and Hangleiter, Dominik and Kang, Christopher},
   year={2026},
   month=Jan, pages={415–436} }

@misc{TB2026,
      title={Copy-cup Gates in Tensor Products of Group Algebra Codes}, 
      author={Ryan Tiew and Nikolas P. Breuckmann},
      year={2026},
      eprint={2602.23307},
      archivePrefix={arXiv},
      primaryClass={quant-ph},
      url={https://arxiv.org/abs/2602.23307}, 
      note={arXiv:2602.23307}
}

@misc{GVG2025,
      title={Near-Asymptotically-Good Quantum Codes with Transversal {CCZ} Gates and Sublinear-Weight Parity-Checks}, 
      author={Louis Golowich and Venkatesan Guruswami},
      year={2025},
      eprint={2510.06798},
      archivePrefix={arXiv},
      primaryClass={quant-ph},
      url={https://arxiv.org/abs/2510.06798},
      note={arXiv:2510.06798}
}

@misc{JCD2026,
      title={Single-Shot Decoding and Fault-tolerant Gates with Trivariate Tricycle Codes}, 
      author={Abraham Jacob and Campbell McLauchlan and Dan E. Browne},
      year={2026},
      eprint={2508.08191},
      archivePrefix={arXiv},
      primaryClass={quant-ph},
      url={https://arxiv.org/abs/2508.08191}, 
      note={arXiv:2508.08191}
}

@misc{YZZQ2025,
      title={Poincar\'e Duality and Multiplicative Structures on Quantum Codes}, 
      author={Yiming Li and Zimu Li and Zi-Wen Liu and Quynh T. Nguyen},
      year={2025},
      eprint={2512.21922},
      archivePrefix={arXiv},
      primaryClass={quant-ph},
      url={https://arxiv.org/abs/2512.21922}, 
      note={arXiv:2512.21922}
}

@article{MBJMRADL2026,
  title = {Magic Tricycles: Efficient Magic-State Generation with Finite Block-Length Quantum {LDPC} Codes},
  author = {Menon, Varun and Bonilla Ataides, J. Pablo and Mehta, Rohan and Gu, Andi and Tan, Daniel Bochen and Lukin, Mikhail D.},
  journal = {Phys. Rev. X},
  volume = {16},
  issue = {2},
  pages = {021014},
  numpages = {37},
  year = {2026},
  month = {Apr},
  publisher = {American Physical Society},
  doi = {10.1103/ghhp-cytl},
  url = {https://link.aps.org/doi/10.1103/ghhp-cytl}
}

@article{CT2021,
  title = {Four-dimensional toric code with non-clifford transversal gates},
  author = {Jochym-O'Connor, Tomas and Yoder, Theodore J.},
  journal = {Phys. Rev. Res.},
  volume = {3},
  issue = {1},
  pages = {013118},
  numpages = {19},
  year = {2021},
  month = {Feb},
  publisher = {American Physical Society},
  doi = {10.1103/PhysRevResearch.3.013118},
  url = {https://link.aps.org/doi/10.1103/PhysRevResearch.3.013118}
}

@misc{THLGH2025,
      title={Single-Shot Universality in Quantum {LDPC} Codes via Code-Switching}, 
      author={Shi Jie Samuel Tan and Yifan Hong and Ting-Chun Lin and Michael J. Gullans and Min-Hsiu Hsieh},
      year={2025},
      eprint={2510.08552},
      archivePrefix={arXiv},
      primaryClass={quant-ph},
      url={https://arxiv.org/abs/2510.08552},
      note={arXiv:2510.08552}
}

@misc{GT2024,
      title={Quantum {LDPC} Codes with Transversal non-clifford Gates via Products of Algebraic Codes}, 
      author={Louis Golowich and Ting-Chun Lin},
      year={2024},
      eprint={2410.14662},
      archivePrefix={arXiv},
      primaryClass={quant-ph},
      url={https://arxiv.org/abs/2410.14662}, 
      note={arXiv:2410.14662}
}

@article{HLC2022,
  author={Hu, J. and Liang, Q. and Calderbank, R.},
  journal={Quantum}, 
  title={Designing the quantum channels induced by diagonal gates}, 
  year={2022},
  volume={6},
  pages={802},
  doi={10.22331/q-2022-09-08-802}
  }

@article{BDMJL2025,
  title = {Asymptotically good {CSS}-{T} codes and a new construction of triorthogonal codes},
  author = {Berardini, E. and Dastbasteh, R. and Martinez, J.E. and Jain, S. and Larrarte, O.S.},
  journal = {IEEE Journal on Selected Areas in Information Theory},
  volume = {6},
  pages = {189--198},
  year = {2025},
  doi = {10.1109/JSAIT.2025.3582156}
  }

@article{JA2025,
      title={Transversal {Clifford} and {$T$}-Gate Codes of Short Length and High Distance}, 
      author={Jain, S.P. and Albert, V.V.},
      journal = {IEEE Journal on Selected Areas in Information Theory},
      volume = {6},
      pages = {127--137},
      year = {2025},
      doi = {10.1109/JSAIT.2025.3570832}
}

@article{SMT1972,
  title = {Good self dual codes exist},
  author = {Sloane, N.J.A. and MacWilliams, F.J. and Thompson, J.G.},
  journal = {Discrete Mathematics},
  volume = {3},
  pages = {153--162},
  year = {1972},
  doi = {10.1016/0012-365X(72)90030-1}
}

@article{WQAS2023,
doi = {10.1088/1367-2630/acfc5f},
url = {https://doi.org/10.1088/1367-2630/acfc5f},
year = {2023},
month = {Oct},
publisher = {IOP Publishing},
volume = {25},
number = {10},
pages = {103018},
author = {Webster, Mark A and Quintavalle, Armanda O and Bartlett, Stephen D},
title = {Transversal diagonal logical operators for stabiliser codes},
journal = {New Journal of Physics}
}

@article{S1995,
  title = {Scheme for reducing decoherence in quantum computer memory},
  author = {Shor, Peter W.},
  journal = {Phys. Rev. A},
  volume = {52},
  issue = {4},
  pages = {R2493(R)--R2496(R)},
  numpages = {0},
  year = {1995},
  month = {Oct},
  publisher = {American Physical Society},
  doi = {10.1103/PhysRevA.52.R2493},
  url = {https://link.aps.org/doi/10.1103/PhysRevA.52.R2493}
}

@article{PW2007,
author = {Martinez-Perez, C. and Willems, W.},
title = {Self-Dual Doubly Even 2-Quasi-Cyclic Transitive Codes Are Asymptotically Good},
year = {2007},
issue_date = {November 2007},
publisher = {IEEE Press},
volume = {53},
number = {11},
issn = {0018-9448},
url = {https://doi.org/10.1109/TIT.2007.907500},
doi = {10.1109/TIT.2007.907500},
journal = {IEEE Trans. Inf. Theor.},
month = nov,
pages = {4302–4308},
numpages = {7}
}

@misc{AM2026,
      title={Constant Rate Codes with Fully Addressable Transversal {T}: Good Codes, Sparse Checks}, 
      author={Alexander Meiburg},
      year={2026},
      eprint={2609.38996},
      archivePrefix={arXiv},
      primaryClass={quant-ph},
      note={arxiv:2609.38996}, 
}

@article{CLMRSS2024,
  title = {An algebraic characterization of binary {{CSS}-T} codes and cyclic {{CSS}-T} codes for quantum fault tolerance},
  author = {Camps-Moreno, E. and López, H.H. and Matthews, G.L. and Ruano, D. and San-José, R. and Soprunov, I.},
  journal = {Quantum Information Processing},
  volume = {23},
  pages = {230},
  year = {2024},
  doi = {10.1007/s11128-024-04427-5}
  }

@article{Haah_2018,
   title={Towers of generalized divisible quantum codes},
   volume={97},
   number={4},
   journal={Phys. Rev. A},
   publisher={American Physical Society (APS)},
   author={Haah, Jeongwan},
   year={2018},
   month=Apr }

@article{HLC2025,
  title = {Climbing the diagonal {Clifford} hierarchy},
  author = {Hu, J. and Liang, Q. and Calderbank, R.},
  journal = {Academia Quantum},
  volume = {2},
  year = {2025},
  doi = {10.20935/AcadQuant7957}
}

@article{KT2019,
  title = {Towards Low Overhead Magic State Distillation},
  author = {Krishna, A. and Tillich, J.P.},
  journal = {Phys. Rev. Lett.},
  volume = {123},
  issue = {7},
  year = {2019},
  publisher = {American Physical Society},
  doi = {10.1103/PhysRevLett.123.070507},
}

@misc{TTF2025,
  title={Clifford gates with logical transversality for self-dual {CSS} codes},
  author={Tansuwannont, T. and Takada, Y. and Fujii, K.},
  year={2025},
  eprint={2503.19790},
  archivePrefix={arXiv},
  primaryClass={quant-ph},
  note={arXiv:2503.19790}
}

@article{EK2009,
  title = {Restrictions on Transversal Encoded Quantum Gate Sets},
  author = {Eastin, Bryan and Knill, Emanuel},
  journal = {Phys. Rev. Lett.},
  volume = {102},
  issue = {11},
  pages = {110502},
  numpages = {4},
  year = {2009},
  month = {Mar},
  publisher = {American Physical Society},
  doi = {10.1103/PhysRevLett.102.110502},
  url = {https://link.aps.org/doi/10.1103/PhysRevLett.102.110502}
}

@book{NC2010,
author = {Nielsen, Michael A. and Chuang, Isaac L.},
title = {Quantum Computation and {Quantum} Information: 10th {Anniversary} Edition},
year = {2011},
isbn = {1107002176},
publisher = {Cambridge University Press},
address = {USA},
edition = {10th}
}

@article{CTV2017,
  title = {Roads towards fault-tolerant universal quantum computation},
  author = {Campbell, E. T. and Vuillot, C.},
  journal = {Nature},
  volume = {549},
  pages = {172--179},
  year = {2017},
  doi = {10.1038/nature23460}
}

@article{BH2012,
  title = {Magic-state distillation with low overhead},
  author = {Bravyi, Sergey and Haah, Jeongwan},
  journal = {Phys. Rev. A},
  volume = {86},
  issue = {5},
  pages = {052329},
  numpages = {10},
  year = {2012},
  month = {Nov},
  publisher = {American Physical Society},
  doi = {10.1103/PhysRevA.86.052329},
  url = {https://link.aps.org/doi/10.1103/PhysRevA.86.052329}
}

@ARTICLE{BADEFSM2025,
  author={Bolkema, Jessalyn and Andrade, Emma and Dexter, Thomas and Eggers, Harrison and Fisher, Victoria L. and Szramowsky, Luke and Manganiello, Felice},
  journal={IEEE Journal on Selected Areas in Information Theory}, 
  title={{CSS}-{T} Codes From {Reed}-{Muller} Codes}, 
  year={2025},
  volume={6},
  number={},
  pages={199-204},
  doi={10.1109/JSAIT.2025.3583217}}

@misc{GK2026,
      title={Symmetric Self-Dual Quantum Codes on High Dimensional Expanders}, 
      author={Kyle Gulshen and Tali Kaufman},
      year={2026},
      eprint={2510.07864},
      archivePrefix={arXiv},
      primaryClass={quant-ph},
      url={https://arxiv.org/abs/2510.07864}, 
      note={arXiv:2510.07864}
}

@article{HH2018,
  title = {Distillation with Sublogarithmic Overhead},
  author = {Hastings, Matthew B. and Haah, Jeongwan},
  journal = {Phys. Rev. Lett.},
  volume = {120},
  issue = {5},
  pages = {050504},
  numpages = {3},
  year = {2018},
  month = {Jan},
  publisher = {American Physical Society},
  doi = {10.1103/PhysRevLett.120.050504},
  url = {https://link.aps.org/doi/10.1103/PhysRevLett.120.050504}
}

@article{HHccz18,
  doi = {10.22331/q-2018-06-07-71},
  url = {https://doi.org/10.22331/q-2018-06-07-71},
  title = {Codes and protocols for distilling {$T$}, Controlled-{$S$}, and {T}offoli gates},
  author = {Haah, Jeongwan and Hastings, Matthew B.},
  journal = {{Quantum}},
  issn = {2521-327X},
  publisher = {{Verein zur F{\"{o}}rderung des Open Access Publizierens in den Quantenwissenschaften}},
  volume = {2},
  pages = {71},
  year = {2018}
}

@article{WHY2025,
  title = {Constant-overhead magic state distillation},
  author = {Wills, A. and Hsieh, MH. and Yamasaki, H},
  journal = {Nat. Phys.},
  volume = {21},
  pages = {1842--1846},
  year = {2025},
  doi = {10.1038/s41567-025-03026-0}
}

@inproceedings{GG2025,
author = {Golowich, Louis and Guruswami, Venkatesan},
title = {Asymptotically Good Quantum Codes with Transversal non-clifford Gates},
year = {2025},
isbn = {9798400715105},
publisher = {Association for Computing Machinery},
url = {https://doi.org/10.1145/3717823.3718234},
doi = {10.1145/3717823.3718234},
booktitle = {Proceedings of the 57th Annual ACM Symposium on Theory of Computing},
pages = {707–717},
numpages = {11},
series = {STOC '25}
}

@inproceedings{N2025,
author = {Nguyen, Quynh T.},
title = {Good Binary Quantum Codes with Transversal {$CCZ$} Gate},
year = {2025},
isbn = {9798400715105},
publisher = {Association for Computing Machinery},
url = {https://doi.org/10.1145/3717823.3718186},
doi = {10.1145/3717823.3718186},
booktitle = {Proceedings of the 57th Annual ACM Symposium on Theory of Computing},
pages = {697–706},
numpages = {10},
series = {STOC '25}
}

@inproceedings{RCNPisit2020,
author = {Rengaswamy, Narayanan and Calderbank, Robert and Newman, Michael and Pfister, Henry D.},
title = {Classical Coding Problem from Transversal {$T$} Gates},
year = {2020},
publisher = {IEEE Press},
url = {https://doi.org/10.1109/ISIT44484.2020.9174408},
doi = {10.1109/ISIT44484.2020.9174408},
booktitle = {2020 IEEE International Symposium on Information Theory (ISIT)},
pages = {1891–1896},
numpages = {6},
location = {Los Angeles, CA, USA}
}

@misc{HLC2022div,
  title={Divisible Codes for Quantum Computation},
  author={Hu, J. and Liang, Q. and Calderbank, R.},
  year={2022},
  eprint={2204.13176},
  archivePrefix={arXiv},
  primaryClass={quant-ph},
  note={arXiv:2204.13176}
}

@misc{G1997,
      title={Stabilizer Codes and Quantum Error Correction}, 
      author={Daniel Gottesman},
      year={1997},
      eprint={quant-ph/9705052},
      archivePrefix={arXiv},
      primaryClass={quant-ph},
      url={https://arxiv.org/abs/quant-ph/9705052}, 
      note={arXiv:quant-ph/9705052}
}

@misc{MN2025,
      title={Asymptotically good {CSS} codes that realize the logical transversal {Clifford} group fault-tolerantly}, 
      author={K. Sai Mineesh Reddy and Navin Kashyap},
      year={2026},
      eprint={2601.08568},
      archivePrefix={arXiv},
      primaryClass={quant-ph},
      note={arXiv:2601.08568}, 
}

@misc{CLMRS2026,
      title={Transversal gates for quantum {CSS} codes}, 
      author={Eduardo Camps-Moreno and Hiram H. López and Gretchen L. Matthews and Narayanan Rengaswamy and Rodrigo San-José},
      year={2026},
      eprint={2601.21514},
      archivePrefix={arXiv},
      primaryClass={cs.IT},
      note={arXiv:2601.21514}, 
}

@article{AT2016,
author = {Anderson, Jonas T. and Jochym-O'Connor, Tomas},
title = {Classification of transversal gates in qubit stabilizer codes},
year = {2016},
issue_date = {July 2016},
publisher = {Rinton Press, Incorporated},
address = {Paramus, NJ},
volume = {16},
number = {9–10},
issn = {1533-7146},
journal = {Quantum Info. Comput.},
month = jul,
pages = {771–802},
numpages = {32}
}

@misc{TTR2026,
      title={Construction of the full logical {Clifford} group for high-rate quantum {Reed}-{Muller} codes using only transversal and fold-transversal gates}, 
      author={Theerapat Tansuwannont and Tim Chan and Ryuji Takagi},
      year={2026},
      eprint={2602.09788},
      archivePrefix={arXiv},
      primaryClass={quant-ph},
      url={https://arxiv.org/abs/2602.09788},
      note={arXiv:2602.09788}
}

@article{BB2024,
  doi = {10.22331/q-2024-06-13-1372},
  url = {https://doi.org/10.22331/q-2024-06-13-1372},
  title = {Fold-transversal {C}lifford Gates for Quantum Codes},
  author = {Breuckmann, Nikolas P. and Burton, Simon},
  journal = {{Quantum}},
  issn = {2521-327X},
  publisher = {{Verein zur F{\"{o}}rderung des Open Access Publizierens in den Quantenwissenschaften}},
  volume = {8},
  pages = {1372},
  month = jun,
  year = {2024}
}

@book{CLRS2009,
author = {Cormen, Thomas H. and Leiserson, Charles E. and Rivest, Ronald L. and Stein, Clifford},
title = {Introduction to Algorithms, Third Edition},
year = {2009},
isbn = {0262033844},
publisher = {The MIT Press},
edition = {3rd}
}

@misc{LLL2026a,
      title={Transversal non-clifford gates on almost-good quantum {LDPC} and quantum locally testable codes}, 
      author={Yiming Li and Zimu Li and Zi-Wen Liu},
      year={2026},
      eprint={2604.01874},
      archivePrefix={arXiv},
      primaryClass={quant-ph},
      note={arxiv:2604.01874}, 
}

@misc{LWHL2026,
      title={Non-abelian sheaf quantum {LDPC} codes: Good and magical}, 
      author={Zimu Li and Fuchuan Wei and Zhengyi Han and Zi-Wen Liu},
      year={2026},
      eprint={2609.30159},
      archivePrefix={arXiv},
      primaryClass={quant-ph},
      note={arXiv:2609.30159},
      url={https://arxiv.org/abs/2609.30159}, 
}

@misc{SJ2026,
      title={Asymptotically good binary triorthogonal codes and higher-level transversal gates}, 
      author={Rodrigo San-José},
      year={2026},
      eprint={2609.08203},
      archivePrefix={arXiv},
      primaryClass={cs.IT},
      note={arXiv:2609.08203},
      url={https://arxiv.org/abs/2609.08203}, 
}

@misc{GTZ2026,
      title={Improved Transversal non-clifford Gates from Cup Products}, 
      author={Louis Golowich and Itzhak Tamo and Guanyu Zhu},
      year={2026},
      eprint={2609.27801},
      archivePrefix={arXiv},
      primaryClass={quant-ph},
      note={arXiv:2609.27801}, 
}

@misc{LLL2026,
      title={Transversal non-clifford gates on good quantum locally testable codes}, 
      author={Yiming Li and Zimu Li and Zi-Wen Liu},
      year={2026},
      eprint={2609.26691},
      archivePrefix={arXiv},
      primaryClass={quant-ph},
      note={arXiv:2609.26691}, 
}

@misc{HVAR2025,
      title={Asymptotically Good Quantum Codes with Addressable and Transversal non-clifford Gates}, 
      author={Zhiyang He and Vinod Vaikuntanathan and Adam Wills and Rachel Yun Zhang},
      year={2025},
      eprint={2507.05392},
      archivePrefix={arXiv},
      primaryClass={quant-ph},
      url={https://arxiv.org/abs/2507.05392}, 
      note={arXiv:2507.05392}
}

@misc{HVAR20251,
      title={Quantum Codes with Addressable and Transversal non-clifford Gates}, 
      author={Zhiyang He and Vinod Vaikuntanathan and Adam Wills and Rachel Yun Zhang},
      year={2025},
      eprint={2502.01864},
      archivePrefix={arXiv},
      primaryClass={quant-ph},
      url={https://arxiv.org/abs/2502.01864}, 
      note={arXiv:2502.01864}
}

@article{S1996,
  title = {Simple quantum error-correcting codes},
  author = {Steane, A. M.},
  journal = {Phys. Rev. A},
  volume = {54},
  issue = {6},
  pages = {4741--4751},
  numpages = {0},
  year = {1996},
  month = {Dec},
  publisher = {American Physical Society},
  doi = {10.1103/PhysRevA.54.4741},
  url = {https://link.aps.org/doi/10.1103/PhysRevA.54.4741}
}

@article{B2025,
  title = {Color Code with a Logical Control-{$S$} Gate Using Transversal {$T$} Rotations},
  author = {Brown, Benjamin J.},
  journal = {Phys. Rev. Lett.},
  volume = {135},
  issue = {7},
  pages = {070602},
  numpages = {6},
  year = {2025},
  month = {Aug},
  publisher = {American Physical Society},
  doi = {10.1103/lwxd-fdlb},
  url = {https://link.aps.org/doi/10.1103/lwxd-fdlb}
}

@article{ALT2001,
  title = {Asymptotically good quantum codes},
  author = {Ashikhmin, Alexei and Litsyn, Simon and Tsfasman, Michael A.},
  journal = {Phys. Rev. A},
  volume = {63},
  issue = {3},
  pages = {032311},
  numpages = {5},
  year = {2001},
  month = {Feb},
  publisher = {American Physical Society},
  doi = {10.1103/PhysRevA.63.032311},
  url = {https://link.aps.org/doi/10.1103/PhysRevA.63.032311}
}

@misc{CHGE2024,
      title={Toward Quantum {CSS}-{T} Codes from Sparse Matrices}, 
      author={Eduardo Camps-Moreno and Hiram H. López and Gretchen L. Matthews and Emily McMillon},
      year={2024},
      note={arXiv:2406.00425},
      archivePrefix={arXiv},
}

@article{BCR2024,
   title={Structure of {CSS} and {CSS}-{T} quantum codes},
   volume={92},
   journal={Des. Codes Cryptogr.},
   publisher={Springer Science and Business Media LLC},
   author={Berardini, Elena and Caminata, Alessio and Ravagnani, Alberto},
   year={2024},
   month=May }

@misc{AOHEM2026,
      title={On Constructing and Decoding Quantum Triorthogonal Codes}, 
      author={Alessio Baldelli and Olai A. Mostad and Hsuan-Yin Lin and Eirik Rosnes and Massimo Battaglioni},
      year={2026},
      note={arXiv:2605.24519},
      archivePrefix={arXiv},
}

@misc{PCR2025,
      title={{CSS-T} codes over Binary Extension Fields and their Physical Foundations}, 
      author={Jasper J. Postema and F. Conca and A. Ravagnani},
      year={2025},
      note={arXiv:2507.17611},
      archivePrefix={arXiv},
}

@misc{BHMR2026,
      title={The {S}chur product of evaluation codes and its application to {CSS}-{T} quantum codes and private information retrieval}, 
      author={Şeyma Bodur and Fernando Hernando and Edgar Martínez-Moro and Diego Ruano},
      year={2026},
      note={arXiv:2505.10068},
      archivePrefix={arXiv},
}

@INPROCEEDINGS{CHGDRI2024,
  author={Camps-Moreno, Eduardo and López, Hiram H. and Matthews, Gretchen L. and Ruano, Diego and San–José, Rodrigo and Soprunov, Ivan},
  booktitle={2024 60th Annual Allerton Conference on Communication, Control, and Computing}, 
  title={Binary Triorthogonal and {CSS-T} Codes for Quantum Error Correction}, 
  year={2024},
  volume={},
  number={},
  pages={01-06},
  doi={10.1109/Allerton63246.2024.10735332}}

@misc{BA2015,
      title={Doubled Color Codes}, 
      author={Sergey Bravyi and Andrew Cross},
      year={2015},
      eprint={1509.03239},
      archivePrefix={arXiv},
      primaryClass={quant-ph},
      url={https://arxiv.org/abs/1509.03239}, 
      note={arXiv:1509.03239}
}

@article{B2015,
doi = {10.1088/1367-2630/17/8/083002},
url = {https://doi.org/10.1088/1367-2630/17/8/083002},
year = {2015},
month = {aug},
publisher = {IOP Publishing},
volume = {17},
number = {8},
pages = {083002},
author = {Bombín, Héctor},
title = {Gauge color codes: optimal transversal gates and gauge fixing in topological stabilizer codes},
journal = {New Journal of Physics},
}

@article{KB2015,
  title = {Universal transversal gates with color codes: A simplified approach},
  author = {Kubica, Aleksander and Beverland, Michael E.},
  journal = {Phys. Rev. A},
  volume = {91},
  issue = {3},
  pages = {032330},
  numpages = {12},
  year = {2015},
  month = {Mar},
  publisher = {American Physical Society},
  doi = {10.1103/PhysRevA.91.032330},
  url = {https://link.aps.org/doi/10.1103/PhysRevA.91.032330}
}

@article{VB2019,
  title = {Three-dimensional surface codes: Transversal gates and fault-tolerant architectures},
  author = {Vasmer, Michael and Browne, Dan E.},
  journal = {Phys. Rev. A},
  volume = {100},
  issue = {1},
  pages = {012312},
  numpages = {20},
  year = {2019},
  month = {Jul},
  publisher = {American Physical Society},
  doi = {10.1103/PhysRevA.100.012312},
  url = {https://link.aps.org/doi/10.1103/PhysRevA.100.012312}
}

@article{KYP2015,
doi = {10.1088/1367-2630/17/8/083026},
url = {https://doi.org/10.1088/1367-2630/17/8/083026},
year = {2015},
month = {aug},
publisher = {IOP Publishing},
volume = {17},
number = {8},
pages = {083026},
author = {Kubica, Aleksander and Yoshida, Beni and Pastawski, Fernando},
title = {Unfolding the color code},
journal = {New Journal of Physics},
}

@article{PR2013,
  title = {Universal Fault-Tolerant Quantum Computation with Only Transversal Gates and Error Correction},
  author = {Paetznick, Adam and Reichardt, Ben W.},
  journal = {Phys. Rev. Lett.},
  volume = {111},
  issue = {9},
  pages = {090505},
  numpages = {5},
  year = {2013},
  month = {Aug},
  publisher = {American Physical Society},
  doi = {10.1103/PhysRevLett.111.090505},
  url = {https://link.aps.org/doi/10.1103/PhysRevLett.111.090505}
}

@misc{DOCJ2026,
      title={Quantum Codes with Arbitrary {Z}-Rotation logical Gates and Applications to Fault-Tolerant Code Switching}, 
      author={Reza Dastbasteh and Ruben M. Otxoa and Pedro M. Crespo and Josu Etxezarreta Martinez},
      year={2026},
      eprint={2608.11160},
      archivePrefix={arXiv},
      primaryClass={quant-ph},
      url={https://arxiv.org/abs/2608.11160}, 
      note={arXiv:2608.11160}
}

@ARTICLE{CYCXSZ2026,
  author={Chen, Hao and Yu, Cong and Chen, Yaqi and Xie, Conghui and Sun, Zhonghua and Zhu, Shixin},
  journal={IEEE Transactions on Information Theory}, 
  title={Optimal, Best-Known Divisible Binary Codes and Related Quantum Codes}, 
  year={2026},
  doi={10.1109/TIT.2026.3719641}}

@inproceedings{G1996,
author = {Grover, Lov K.},
title = {A fast quantum mechanical algorithm for database search},
year = {1996},
isbn = {0897917855},
publisher = {Association for Computing Machinery},
address = {New York, NY, USA},
url = {https://doi.org/10.1145/237814.237866},
doi = {10.1145/237814.237866},
booktitle = {Proceedings of the Twenty-Eighth Annual ACM Symposium on Theory of Computing},
pages = {212–219},
numpages = {8},
location = {Philadelphia, Pennsylvania, USA},
series = {STOC '96}
}

@article{S1997,
author = {Shor, Peter W.},
title = {Polynomial-Time Algorithms for Prime Factorization and Discrete Logarithms on a Quantum Computer},
journal = {SIAM Journal on Computing},
volume = {26},
number = {5},
pages = {1484-1509},
year = {1997},
doi = {10.1137/S0097539795293172},
URL = {https://doi.org/10.1137/S0097539795293172},
eprint = {https://doi.org/10.1137/S0097539795293172
}
}

@article{GJ2026,
  doi = {10.22331/q-2026-06-30-2145},
  url = {https://doi.org/10.22331/q-2026-06-30-2145},
  title = {On the {A}ddressability {P}roblem on {CSS} {C}odes},
  author = {Guyot, J{\'{e}}r{\^{o}}me and Jaques, Samuel},
  journal = {{Quantum}},
  issn = {2521-327X},
  publisher = {{Verein zur F{\"{o}}rderung des Open Access Publizierens in den Quantenwissenschaften}},
  volume = {10},
  pages = {2145},
  month = jun,
  year = {2026}
}

@article{QWV2023,
  doi = {10.22331/q-2023-10-24-1153},
  url = {https://doi.org/10.22331/q-2023-10-24-1153},
  title = {Partitioning qubits in hypergraph product codes to implement logical gates},
  author = {Quintavalle, Armanda O. and Webster, Paul and Vasmer, Michael},
  journal = {{Quantum}},
  issn = {2521-327X},
  publisher = {{Verein zur F{\"{o}}rderung des Open Access Publizierens in den Quantenwissenschaften}},
  volume = {7},
  pages = {1153},
  month = oct,
  year = {2023}
}

@article{PB2025,
  doi = {10.22331/q-2025-08-29-1842},
  url = {https://doi.org/10.22331/q-2025-08-29-1842},
  title = {Targeted {C}lifford logical gates for hypergraph product codes},
  author = {Patra, Adway and Barg, Alexander},
  journal = {{Quantum}},
  issn = {2521-327X},
  publisher = {{Verein zur F{\"{o}}rderung des Open Access Publizierens in den Quantenwissenschaften}},
  volume = {9},
  pages = {1842},
  month = aug,
  year = {2025}
}

@misc{W2026,
      title={Improved Quantum Codes with Transversal {T} Gates}, 
      author={Adam Wills},
      year={2026},
      eprint={2608.24000},
      archivePrefix={arXiv},
      primaryClass={quant-ph},
      url={https://arxiv.org/abs/2608.24000}, 
      note={arXiv:2608.24000}
}

@misc{CB2026,
      title={Quantum Lego Power-up: Designing Transversal Gates with Tensor Networks}, 
      author={ChunJun Cao and Brad Lackey},
      year={2026},
      eprint={2603.03542},
      archivePrefix={arXiv},
      primaryClass={quant-ph},
      url={https://arxiv.org/abs/2603.03542},
      note={arXiv:2603.03542}
}

@misc{GPRW2026,
      title={Magic State Distillation via Codes over Binary Extension Fields}, 
      author={Anqi Gong and Christopher A. Pattison and Patrick Rall and Adam Wills},
      year={2026},
      eprint={2608.09727},
      archivePrefix={arXiv},
      primaryClass={quant-ph},
      url={https://arxiv.org/abs/2608.09727}, 
      note={arXiv:2608.09727}
}

\end{document}